\documentclass[11pt,reqno]{amsart}
\usepackage[T1]{fontenc}
\usepackage{lmodern}
\usepackage{microtype}
\usepackage[a4paper,textwidth=6.25in,textheight=9.05in,centering]{geometry}
\usepackage{amsmath,amssymb,mathtools,bm}
\usepackage{enumitem,booktabs,array,tikz,xcolor}
\usepackage[normalem]{ulem}
\usepackage{placeins}
\usetikzlibrary{arrows.meta,positioning,calc}
\usepackage{aliascnt}
\usepackage[hidelinks,bookmarksnumbered]{hyperref}
\usepackage[capitalise,noabbrev]{cleveref}
\numberwithin{equation}{section}
\newtheorem{theorem}{Theorem}[section]
\newaliascnt{proposition}{theorem}
\newtheorem{proposition}[proposition]{Proposition}
\aliascntresetthe{proposition}
\newaliascnt{lemma}{theorem}
\newtheorem{lemma}[lemma]{Lemma}
\aliascntresetthe{lemma}
\newaliascnt{corollary}{theorem}
\newtheorem{corollary}[corollary]{Corollary}
\aliascntresetthe{corollary}
\theoremstyle{definition}
\newaliascnt{assumption}{theorem}

\aliascntresetthe{assumption}
\newaliascnt{condition}{theorem}
\newtheorem{condition}[condition]{Condition}
\aliascntresetthe{condition}
\newaliascnt{definition}{theorem}

\aliascntresetthe{definition}
\newaliascnt{background}{theorem}

\aliascntresetthe{background}
\theoremstyle{remark}
\newaliascnt{remark}{theorem}
\newtheorem{remark}[remark]{Remark}
\aliascntresetthe{remark}
\crefname{background}{background theorem}{background theorems}
\Crefname{background}{Background theorem}{Background theorems}
\crefname{condition}{condition}{conditions}
\Crefname{condition}{Condition}{Conditions}
\DeclareMathOperator{\Tr}{Tr}
\DeclareMathOperator{\supp}{supp}
\DeclareMathOperator{\dist}{dist}
\DeclareMathOperator{\ad}{ad}
\DeclareMathOperator{\Ad}{\mathsf{Ad}}
\DeclareMathOperator{\sech}{sech}
\DeclareMathOperator{\ran}{ran}

\DeclareMathOperator{\Ent}{Ent}

\DeclareMathOperator{\id}{id}

\newcommand{\cL}{\mathcal L}

\newcommand{\one}{\mathbf 1}
\newcommand{\R}{\mathbb R}

\newcommand{\Z}{\mathbb Z}

\newcommand{\Dtwo}{\widetilde D_2}

\DeclareMathAlphabet{\mathsfgreek}{OT1}{lmss}{m}{n}
\newcommand{\sGamma}{\mathsfgreek{\Gamma}}
\newcommand{\norm}[1]{\left\lVert#1\right\rVert}

\setlist[enumerate]{label=\textup{(\roman*)},leftmargin=2.25em,itemsep=3pt,topsep=4pt}

\newcommand{\Ball}{\mathrm B}

\usepackage{todonotes}

\hypersetup{pdftitle={Size-independent logarithmic Sobolev inequalities for quantum Gibbs samplers},pdfsubject={A general proof framework with one-dimensional and high-temperature results}}
\makeatletter
\renewcommand\paragraph{\@startsection{paragraph}{4}{\z@}%
  {1.2ex \@plus .3ex \@minus .2ex}{-.7em}{\normalfont\bfseries}}
\newcommand\proofstep{\@startsection{paragraph}{4}{\z@}%
  {1.2ex \@plus .3ex \@minus .2ex}{-.7em}{\normalfont\itshape}*}
\makeatother

\allowdisplaybreaks[2]
\title[MLSI for quantum Gibbs samplers]{Size-independent logarithmic Sobolev inequalities for quantum Gibbs samplers}
\author{Hongrui Chen}
\address{Department of Mathematics, Stanford University, Stanford, CA 94305, USA}
\email{hongrui@stanford.edu}
\author{Zhiyan Ding}
\address{Department of Mathematics, University of Michigan, Ann Arbor, MI 48109, USA}
\email{zyding@umich.edu}
\author{Jianfeng Lu}
\address{Department of Mathematics, Duke University, Durham, NC 27708, USA}
\email{jianfeng@math.duke.edu}
\date{October 1, 2026}

\begin{document}
\begin{abstract}
We propose a new method for proving $L^2$ and modified logarithmic Sobolev inequalities with constants independent of system size for quantum Gibbs samplers whose Lindblad generators satisfy Kubo-Martin-Schwinger (KMS) detailed balance. We first establish the $L^2$ logarithmic Sobolev inequality by transporting it along a temperature-dependent path from the infinite-temperature product state to the target Gibbs state while maintaining a size-independent decay rate for the sandwiched R\'enyi-2 divergence. A tensor-power argument, combined with the quantum asymptotic equipartition property, then transfers this decay bound to relative entropy, yielding the modified logarithmic Sobolev inequality. This approach provides a general framework for proving size-independent logarithmic Sobolev inequalities for quantum Gibbs samplers. We apply it to two regimes: one-dimensional quantum spin chains at every finite temperature, and finite-range quantum spin systems at sufficiently high temperature in any fixed dimension.


\end{abstract}
\maketitle
\enlargethispage{2pt}


\tableofcontents

\section{Introduction}\label{sec:introduction}

The modified logarithmic Sobolev inequality (MLSI) characterizes exponential decay of relative entropy under quantum Markov dynamics. By relating relative entropy to its instantaneous dissipation, it provides a central tool for understanding mixing and obtaining quantitative convergence bounds~\cite{KT,Carbone2014}. For many-body systems, a basic question is whether this entropy decay occurs at a rate independent of system size.

Quantum Gibbs sampling provides a natural setting for this question: dissipative dynamics drive the system toward the Gibbs state of an interacting Hamiltonian. Such dynamics provide both algorithms for preparing thermal states on quantum computers and a natural class of processes for studying many-body thermalization~\cite{CKG,RFAThermalization}. For commuting Hamiltonians, Davies generators give a standard choice of Gibbs sampler satisfying Gelfand-Naimark-Segal (GNS) detailed balance~\cite{KCAR}. For non-commuting Hamiltonians, Davies dynamics can become nonlocal and difficult to implement efficiently. Recent constructions address this difficulty through quasi-local Gibbs samplers satisfying Kubo-Martin-Schwinger (KMS) detailed balance, which can be efficiently simulated on quantum computers, such as the Gaussian KMS Gibbs sampler~\cite{CKG} and the broader family in~\cite{DLL}.

Proving a size-independent MLSI for these non-commuting KMS samplers, however, remains challenging. Earlier many-body MLSI results largely concerned GNS-symmetric dynamics and commuting interactions~\cite{CRF,KCAR}, where proofs exploit local conditional expectations and spatial factorization of relative entropy that do not carry over to non-commuting systems. Moreover, while GNS-symmetric dynamics admit a Riemannian gradient-flow formulation of relative entropy~\cite{CM2017,CM2020}, general KMS-symmetric dynamics lack this geometric structure. Consequently, even establishing the existence of a nontrivial MLSI constant in the general KMS-symmetric setting has been far from straightforward.

Recent work has nevertheless made substantial progress in understanding the convergence of non-commuting KMS samplers. Polynomial-time (fast) mixing has been established via spectral-gap bounds~\cite{RFAThermalization,TZ,SMBBFast,BC2025}, while polylogarithmic-time (rapid) mixing has been proved using oscillator-norm contraction~\cite{RFA,SMBBRapid}, quantum Wasserstein contraction~\cite{BLMT,BMT,Kim2026,RFDS2026}, R\'enyi-2 divergence decay~\cite{CD2026}, and entropy quasi-factorization~\cite{CPS2026}. The latter two approaches have yielded MLSI bounds for specific samplers of one-dimensional non-commuting systems, but with constants that deteriorate with system size. For $n$-site systems, \cite{CPS2026} establishes an MLSI lower bound of $\Omega(1/\log n)$ for heat-bath and regularised heat-bath dynamics, while \cite{CD2026} obtains $\Omega(1/\log^4 n)$ for the Gaussian KMS Gibbs sampler~\cite{CKG}.
Obtaining a size-independent rate, both in one dimension and for broader classes of non-commuting systems, therefore remains open. This motivates the following question:

\medskip
\noindent\emph{Can we obtain size-independent MLSI for KMS-symmetric Gibbs samplers of broad classes of non-commuting Hamiltonians?}
\medskip

Addressing this question requires controlling relative-entropy dissipation, a nonlinear quantity that is difficult to estimate directly for interacting quantum systems. A useful alternative is to pass through the $L^2$ logarithmic Sobolev inequality ($L^2$ LSI), which bounds the $L^2$ relative entropy (a rescaled relative entropy; see~\eqref{eq:l2-relative-entropy}) by a Dirichlet form and yields exponential decay of the sandwiched R\'enyi-2 divergence~\cite{OZ1999,KT,MHF}; one then transfers this R\'enyi-2 decay to relative-entropy decay. Within this two-step framework, \cite{CD2026} obtains R\'enyi-2 decay in one dimension by exploiting the stronger dissipation of states far from equilibrium, and transfers it to MLSI via hypercontractivity~\cite{GW2026}; however, both steps incur logarithmic losses in system size.

In this work, we overcome the logarithmic losses in both steps through a new proof strategy, obtaining size-independent bounds for the $L^2$ LSI and MLSI. First, we establish a size-independent $L^2$ LSI by transporting the known $L^2$ LSI at infinite temperature to the target Gibbs state along a unitary flow. The key is to preserve the comparison between the $L^2$ relative entropy and the Dirichlet form as both quantities change along the path. Local estimates on their variations ensure that the resulting constant does not deteriorate as the system grows.
Second, with size-independent R\'enyi-2 decay in hand, we transfer this decay to relative entropy via a tensor-power argument without losing the rate. We use the quantum asymptotic equipartition property~\cite{TCR2009,DMHB} to derive relative-entropy bounds for a single system from R\'enyi-2 bounds on many independent copies. Since the same decay rate holds for every number of copies, the rate is preserved.

This proof strategy provides a general framework for proving modified logarithmic Sobolev inequalities and can apply to various classes of non-commuting Hamiltonians, as shown in our main theorem:


\begin{theorem}[Informal; Theorem~\ref{thm:main-mlsi}]\label{thm:informal-results}
For quantum spin systems with finite-range interactions of bounded strength and fixed local dimension, the Gaussian KMS Gibbs sampler of~\cite{CKG} satisfies both the $L^2$ LSI and MLSI with constants independent of system size in the following regimes:
\begin{itemize}
\item In one dimension, at any fixed inverse temperature $0<\beta<\infty$.
\item In any fixed higher dimension, at sufficiently high temperature, with a temperature threshold independent of system size.
\end{itemize}
\end{theorem}

  Theorem~\ref{thm:informal-results} implies size-independent decay rates for both the R\'enyi-2 divergence and relative entropy in one dimension at every finite temperature and in higher dimensions at sufficiently high temperature. 
Beyond these quantitative bounds, there are several conceptual implications of our work, which we discuss in the following:
\begin{enumerate}
\item Size-independent MLSI is not necessarily tied to the existence of a gradient flow structure.

GNS-symmetric quantum Markov semigroups admit a Riemannian gradient-flow formulation of relative entropy~\cite{CM2017,CM2020}. More generally, for finite-dimensional ergodic semigroups with a faithful invariant state, such a formulation exists if and only if Bogoliubov-Kubo-Mori (BKM) detailed balance holds~\cite{CM2020,BM2024}. Although KMS detailed balance alone does not guarantee this structure, our transport method surprisingly establishes size-independent MLSI for the Gaussian KMS Gibbs sampler in the one-dimensional and high-temperature regimes of Theorem~\ref{thm:main-mlsi}, without imposing BKM detailed balance or using a gradient-flow formulation. This confirms that the existence of a gradient-flow structure is not necessary for size-independent MLSI.

\item A quantum counterpart of the classical passage from the $L^2$ LSI to entropy decay.

For reversible classical Markov chains, the $L^2$ LSI implies MLSI and hence exponential decay of relative entropy~\cite{DSC1996,BT2006}. In the quantum setting, $L_p$-regularity yields a corresponding comparison without system-size loss~\cite{KT,Carbone2014}, but the required regularity is not known to follow from KMS detailed balance alone. \cite{GW2026} derives MLSI from hypercontractivity without this assumption. However, directly applying their result to the Gaussian KMS Gibbs sampler will incur a logarithmic loss in the MLSI bound.

Our result overcomes this difficulty by establishing an $L^2$ LSI for the tensor-sum generator on every number of copies, with a constant independent of the number of copies and valid for all positive inputs on the joint system. This inequality yields sandwiched R\'enyi-2 decay~\cite{MHF}, and Theorem~\ref{thm:transfer} converts this into relative-entropy decay with the same rate when the bound holds uniformly on every tensor power, including for inputs correlated across copies.

\item The transport method provides a new approach to proving mixing time bounds.

A size-independent MLSI constant directly implies $\mathcal O
\left(\log(n)+\log(1/\epsilon)\right)$ mixing time for the Gaussian KMS Gibbs sampler in the one-dimensional and high-temperature regimes of Theorem~\ref{thm:main-mlsi}. Our method differs from many existing approaches for proving quantum mixing bounds, which establish functional inequalities directly for a Gibbs sampler at the target temperature. Our method obtains mixing bounds by transporting an $L^2$ LSI from infinite temperature to the target Gibbs state via a unitary path.


The unitary path is a quantum analogue of Moser's construction of transport maps~\cite{Moser}. Classically, the Gibbs densities $p_s\propto e^{-sH}p_0$ are transported by a gradient flow whose potential solves a Poisson equation for the Langevin generator, and any density carried by the same flow keeps its relative entropy to $p_s$. On square-root densities this flow is generated by the commutator of a multiplication operator with the generator. Our quantum flow keeps this commutator form, which lets us control how relative entropy changes along the path; see Remark~\ref{rem:classical-counterpart}.

\end{enumerate}





\subsection*{Related work}

\paragraph*{Mixing time of quantum Gibbs samplers}
For commuting Hamiltonians, spectral gap analysis has been used to establish fast mixing of Davies generators, with the required spectral gap often guaranteed by decay of correlations in the Gibbs state~\cite{KB2016}.
In the case of commuting spin chains at every finite temperature, higher-dimensional systems at high temperature~\cite{CRF,KCAR}, and quantum error-correcting-code Hamiltonians~\cite{CSSMLSI,AbelianMLSI,HaahRapid2026}, size-independent MLSI constants have been established through spatial mixing and entropy factorization. This ensures rapid mixing of the Davies generator in these commuting cases.

For non-commuting Hamiltonians, quasi-local Gibbs samplers satisfying KMS detailed balance~\cite{CKG,DLL} have made it possible to study both fast and rapid mixing in several regimes. In perturbative regimes, including high temperature and weak interactions, spectral-gap estimates yield fast mixing~\cite{RFAThermalization,TZ,SMBBFast,CZZ2026}, while subsequent results establish rapid mixing through oscillator-norm contraction~\cite{RFA,SMBBRapid} or quantum Wasserstein path coupling~\cite{BLMT,BMT}. For general one-dimensional spin systems beyond perturbative regimes, \cite{BC2025} established fast mixing for the Gaussian KMS Gibbs sampler~\cite{CKG} at every finite temperature via spectral-gap estimates. More recently, \cite{CD2026} proves rapid mixing for the Gaussian KMS Gibbs sampler from spectral analysis beyond the slowest mode. \cite{Kim2026,RFDS2026} establish rapid mixing for block and smoothed heat-bath dynamics, respectively, using Wasserstein contraction and influence decay. \cite{CPS2026} obtains MLSI bounds for heat-bath and regularised heat-bath dynamics by combining weak quasi-factorization of relative entropy with spectral-gap estimates.

\paragraph*{Entropy decay of quantum Markov semigroups}
Beyond these results for specific Gibbs samplers, a broader literature studies entropy decay through logarithmic Sobolev inequalities for quantum Markov semigroups. One approach to establishing $L^2$ LSIs is through hypercontractivity~\cite{OZ1999,KT}. Tools for establishing hypercontractivity include noncommutative integration and matrix norm inequalities~\cite{CL1993,King}. Interpolation provides another route from finite-time norm estimates and a spectral gap~\cite{TPK2014}. Alternatively, the noncommutative Holley--Stroock method transfers logarithmic Sobolev inequalities between GNS-symmetric generators with common jump operators and different invariant states~\cite{JLR2019}. Our approach instead transports the inequality along a unitary flow, as discussed above.

To obtain MLSI from an $L^2$ LSI, $L_p$-regularity compares the Dirichlet form with relative-entropy dissipation~\cite{KT,Carbone2014}; noncommutative Stroock--Varopoulos inequalities provide such comparisons under GNS detailed balance~\cite{BDR2018}. Without assuming $L_p$-regularity, \cite{GW2026} derives MLSI from hypercontractive estimates; this approach is used in \cite{CD2026} to transfer R\'enyi-2 decay to MLSI with a polylogarithmic loss in system size. Our tensor-power argument uses the quantum asymptotic equipartition property~\cite{TCR2009,DMHB} to transfer R\'enyi-2 decay to relative entropy without losing the rate.

Geometric approaches relate entropy convexity to decay through the GNS gradient-flow framework~\cite{CM2017,CM2020,CLL2018}, whose Riemannian formulation extends precisely to dynamics satisfying Bogoliubov-Kubo-Mori (BKM) detailed balance~\cite{CM2020,BM2024}. Within this framework, curvature bounds combined with mixing estimates yield MLSI~\cite{BGJ2020}. More recently, \cite{GG2026} derives MLSI from BKM coercivity through integral representations of entropy, without assuming detailed balance.

\paragraph*{Concurrent work}
In independent and concurrent work, \cite{gao2026completemodifiedlogarithmicsobolev} proves a size-independent complete MLSI for the Gaussian KMS Gibbs sampler of one-dimensional spin systems at every finite temperature. This overlaps with our one-dimensional MLSI result, but the approaches are substantially different: their proof uses spatial entropy factorization, whereas ours combines unitary transport of an $L^2$ LSI with a tensor-power argument. More broadly, our framework reduces the proof to local mixing and local energy conditions in arbitrary spatial dimension (Theorem~\ref{prop:general-criterion}), separating the general argument from its verification for specific systems. This framework also yields MLSI for non-commuting spin systems at high temperature in arbitrary spatial dimension, which we establish here for the first time.

\section{Notation and preliminaries}\label{sec:notation}
Section~\ref{sec:notation-vectorization} fixes notation for matrices, superoperators, and vectorization. Section~\ref{sec:notation-kms} introduces KMS detailed balance and the weighted generator. Section~\ref{sec:notation-entropies} introduces the divergences and entropy functionals, and Section~\ref{sec:notation-inequalities} introduces the functional inequalities and the entropy decay they imply.

\subsection{Matrices, superoperators, and vectorization}\label{sec:notation-vectorization}
We work on a finite-dimensional Hilbert space $\mathcal H$ and write $\mathcal B(\mathcal H)$ for its matrix algebra and $I$ for the identity matrix. A state is a positive semidefinite matrix of trace one; it is faithful when it is positive definite. We denote the Hilbert--Schmidt inner product $\langle X,Y\rangle=\Tr(X^*Y)$, and $\|X\|_p=(\Tr|X|^p)^{1/p}$ denotes the Schatten $p$-norm for $1\leq p<\infty$, with $\|X\|_\infty$ the operator norm. For simplicity, we also write $\|X\|:=\|X\|_\infty$ for the operator norm. All logarithms are natural, and functions of positive matrices use the convention $0\log0=0$.

To represent superoperators on $\mathcal B(\mathcal H)$ as matrices, fix an orthonormal basis of $\mathcal H$ and define vectorization and its inverse by
\[
 |Y\rangle:=\operatorname{vec}(Y)=\sum_{i,j}Y_{ij}|i\rangle\otimes|j\rangle,
 \qquad \operatorname{vec}^{-1}(|Y\rangle)=Y.
\]
It identifies the matrix algebra, equipped with the Hilbert--Schmidt inner product, with $\mathcal H\otimes\mathcal H$:
\[
 \langle X|Y\rangle=\Tr(X^*Y),\qquad
 \bigl\||Y\rangle\bigr\|_2=\|Y\|_2.
\]
A matrix $Y$ and its vectorization $|Y\rangle$ are distinguished throughout, and $\one$ denotes the identity operator on $\mathcal H\otimes\mathcal H$. In the chosen basis, left and right multiplication satisfy
\[
 |AYB\rangle=(A\otimes B^{\mathrm T})|Y\rangle,
\]
where $^{\mathrm T}$ denotes transpose.

Operators acting on vectorizations are written in sans-serif letters. If $\mathcal T$ is a superoperator on matrices, its vectorized representative $\mathsf{T}$ is defined by
\[
 \mathsf{T}|Y\rangle=|\mathcal T(Y)\rangle,\qquad
 \mathsf{T}=\operatorname{vec}\circ\mathcal T\circ\operatorname{vec}^{-1}.
\]
For a superoperator, $\mathcal T^*$ denotes its Hilbert--Schmidt adjoint; its vectorized representative is $\mathsf{T}^*$. The corresponding induced norms agree:
\begin{equation}\label{eqn:induced-norms_def}
 \|\mathcal T\|_{2\to2}
 =\sup_{\|Y\|_2=1}\|\mathcal T(Y)\|_2=\|\mathsf{T}\|.
\end{equation}
The matrix commutator map and symmetrized multiplication are
\begin{equation}\label{eq:symmetrized-multiplication}
 \ad_V(Y)=[V,Y],\qquad \mathcal M_f(Y)=\tfrac12(fY+Yf).
\end{equation}
Their vectorized representatives are denoted by $\Ad_V$ and $\mathsf{M}_f$, respectively:
\[
 \Ad_V|Y\rangle=|[V,Y]\rangle,\qquad
 \mathsf{M}_f|Y\rangle=|\mathcal M_f(Y)\rangle.
\]
Thus functions of $\ad_V$ applied to matrices are superoperators, whereas $\Ad_V$ acts on vectorizations.

\subsection{KMS detailed balance and the weighted generator}\label{sec:notation-kms}
For a faithful state $\sigma$ we follow~\cite{KT} and write
\[
 \Gamma_\sigma(X)=\sigma^{1/2}X\sigma^{1/2},\qquad
 \Gamma_\sigma^{p}(X)=\sigma^{p/2}X\sigma^{p/2}\quad(p\in\R).
\]
The KMS inner product and its norm are
\begin{equation}\label{eqn:KMS_norm}
 \langle X,Y\rangle_\sigma
 =\langle X,\Gamma_\sigma(Y)\rangle
 =\langle\Gamma_\sigma^{1/2}(X),\Gamma_\sigma^{1/2}(Y)\rangle,\qquad
 \|X\|_{2,\sigma}=\|\Gamma_\sigma^{1/2}(X)\|_2.
\end{equation}
Thus $\Gamma_\sigma^{1/2}$ is an isometry from the KMS inner product space to Hilbert--Schmidt space. Its vectorized representative is written with a sans-serif $\sGamma$:
\[
 \sGamma_\sigma^{1/2}|X\rangle=|\Gamma_\sigma^{1/2}(X)\rangle,\qquad
 \sGamma_\sigma^{1/2}=\sigma^{1/4}\otimes(\sigma^{1/4})^{\mathrm T},
\]
and $\sGamma_\sigma^{-1/2}=(\sGamma_\sigma^{1/2})^{-1}$. The operator $\sGamma_\sigma^{1/2}$ is positive definite, so $(\sGamma_\sigma^{1/2})^*=\sGamma_\sigma^{1/2}$. The isometry above uses the KMS norm on the domain; $\sGamma_\sigma^{1/2}$ need not be unitary in the unweighted inner product. The map $\Gamma_\sigma^{1/2}$ sends the identity matrix to the Gibbs square-root matrix:
\begin{equation}\label{eq:omega_sigma}
 \Omega_\sigma:=\Gamma_\sigma^{1/2}(I)=\sqrt\sigma,\qquad
 |\Omega_\sigma\rangle=\sGamma_\sigma^{1/2}|I\rangle.
\end{equation}
The vector $|\Omega_\sigma\rangle$ is normalized because $\Tr\sigma=1$.

Let $\mathcal L$ be a Lindblad generator acting on observables, with trace adjoint $\mathcal L_*$ and invariant state $\sigma$, so $\mathcal L_*(\sigma)=0$. It is primitive if $e^{t\mathcal L_*}(\rho)\to\sigma$ for every state $\rho$. KMS detailed balance, or KMS symmetry, means
\begin{equation}\label{eqn:L-KMS}
 \langle X,\mathcal L(Y)\rangle_\sigma
 =\langle\mathcal L(X),Y\rangle_\sigma
 \qquad(X,Y\in\mathcal B(\mathcal H)),
\end{equation}
or equivalently $\Gamma_\sigma\circ\mathcal L=\mathcal L_*\circ\Gamma_\sigma$~\cite[Definition~6]{KT}. The Dirichlet form associated with $\mathcal L$ is
\[
 \mathcal E_{2,\sigma}(X):=-\operatorname{Re}\langle X,\mathcal L(X)\rangle_\sigma.
\]

For a primitive KMS-symmetric observable generator $\mathcal L$, let $\mathsf{L}$ and $\mathsf{L}_*$ be the vectorized representatives of $\mathcal L$ and its trace adjoint $\mathcal L_*$, respectively. We define the weighted generator on matrices and its vectorized representative by
\begin{equation}\label{eq:weighted-generator}
 \begin{aligned}
 \mathcal L_\sigma
 &=\Gamma_\sigma^{1/2}\mathcal L\Gamma_\sigma^{-1/2}
  =\Gamma_\sigma^{-1/2}\mathcal L_*\Gamma_\sigma^{1/2},\\
 \mathsf{L}_\sigma
 &=\sGamma_\sigma^{1/2} \mathsf{L} \sGamma_\sigma^{-1/2}=\sGamma_\sigma^{-1/2}\mathsf{L}_*\sGamma_\sigma^{1/2}.
 \end{aligned}
\end{equation}
The equalities between the observable and state forms follow from KMS symmetry: composing $\Gamma_\sigma\circ\mathcal L=\mathcal L_*\circ\Gamma_\sigma$ with $\Gamma_\sigma^{-1/2}$ on both sides gives $\Gamma_\sigma^{1/2}\mathcal L\Gamma_\sigma^{-1/2}=\Gamma_\sigma^{-1/2}\mathcal L_*\Gamma_\sigma^{1/2}$. Since $\mathcal L$ preserves Hermiticity, its trace adjoint $\mathcal L_*$ is also its Hilbert--Schmidt adjoint, so $\mathsf{L}_*=\mathsf{L}^*$. We reserve $\mathsf{L}$ for the original vectorized generator and $\mathsf{L}_\sigma$ for the weighted generator; its parent Hamiltonian is $-\mathsf{L}_\sigma$. They satisfy
\begin{equation}\label{eq:parent-action}
 \mathsf{L}_\sigma \sGamma_\sigma^{1/2}=\sGamma_\sigma^{1/2} \mathsf{L},\qquad
 \mathsf{L}_\sigma=\mathsf{L}_\sigma^*\leq0,\qquad
 \mathsf{L}_\sigma|\Omega_\sigma\rangle=0.
\end{equation}
The first identity is~\eqref{eq:weighted-generator}. Taking adjoints in its second form, with $\mathsf L_*^*=\mathsf L$ and $(\sGamma_\sigma^{\pm1/2})^*=\sGamma_\sigma^{\pm1/2}$, gives $\mathsf L_\sigma^*=\sGamma_\sigma^{1/2}\mathsf L\sGamma_\sigma^{-1/2}=\mathsf L_\sigma$. The last identity follows from $\mathcal L(I)=0$ and~\eqref{eq:omega_sigma}. For the sign, $\mathsf L_\sigma$ is similar to $\mathsf L$, whose eigenvalues have nonpositive real part because the unital channels $e^{t\mathcal L}$ are bounded uniformly in $t$; since $\mathsf L_\sigma$ is self-adjoint, $\mathsf L_\sigma\leq0$. Finally, by~\eqref{eqn:KMS_norm} and the first identity,
\[
 \langle X,\mathcal L(X)\rangle_\sigma
 =\langle X|\sGamma_\sigma^{1/2}\sGamma_\sigma^{1/2}\mathsf L|X\rangle
 =\langle X|\sGamma_\sigma^{1/2}\mathsf L_\sigma\sGamma_\sigma^{1/2}|X\rangle,
\]
which is real because $\mathsf L_\sigma$ is self-adjoint. Consequently,
\begin{equation}\label{eqn:l2-dirichlet}
 \mathcal E_{2,\sigma}(X)
 =\langle X|\sGamma_\sigma^{1/2}(-\mathsf{L}_\sigma)\sGamma_\sigma^{1/2}|X\rangle.
\end{equation}

\subsection{Divergences and entropy functionals}\label{sec:notation-entropies}
For a state $\rho$ and a faithful state $\sigma$, the relative density of $\rho$ is $\Gamma_\sigma^{-1}(\rho)=\sigma^{-1/2}\rho\sigma^{-1/2}$; it satisfies $\Tr(\sigma\Gamma_\sigma^{-1}(\rho))=1$. Umegaki relative entropy and the sandwiched R\'enyi-2 divergence are
\begin{align*}
 D(\rho\Vert\sigma)&=\Tr[\rho(\log\rho-\log\sigma)],\\
 \widetilde D_2(\rho\Vert\sigma)
 &=\log\|\Gamma_\sigma^{-1/2}(\rho)\|_2^2
 =\log\|\Gamma_\sigma^{-1}(\rho)\|_{2,\sigma}^2.
\end{align*}
They satisfy $D\leq\widetilde D_2$~\cite{MLDS,FL}.

Following~\cite{KT}, the $L^2$ relative entropy of a positive definite matrix $X$ is
\begin{equation}\label{eq:l2-relative-entropy}
 \Ent_{2,\sigma}(X)
 =\tfrac12\Tr\Bigl[\bigl(\Gamma_\sigma^{1/2}(X)\bigr)^2
   \Bigl(\log\bigl(\Gamma_\sigma^{1/2}(X)\bigr)^2-\log\sigma\Bigr)\Bigr]
  -\tfrac12\|X\|_{2,\sigma}^2\log\|X\|_{2,\sigma}^2.
\end{equation}
It is a rescaled relative entropy. Since $\|\Gamma_\sigma^{1/2}(X)\|_2=\|X\|_{2,\sigma}$, the matrix
\[
 \eta_X:=\bigl(\Gamma_\sigma^{1/2}(X)\bigr)^2/\|X\|_{2,\sigma}^2
\]
is a state, with square root $\Gamma_\sigma^{1/2}(X)/\|X\|_{2,\sigma}$, and
\begin{equation}\label{eqn:l2-entropy}
 2\Ent_{2,\sigma}(X)=\|X\|_{2,\sigma}^2\,D(\eta_X\Vert\sigma),
 \qquad X>0.
\end{equation}
Indeed, $(\Gamma_\sigma^{1/2}(X))^2=\|X\|_{2,\sigma}^2\,\eta_X$ with $\Tr\eta_X=1$, so $\log(\Gamma_\sigma^{1/2}(X))^2=\log\eta_X+\log\|X\|_{2,\sigma}^2\,I$. Substituting this into the definition, the first term of $\Ent_{2,\sigma}(X)$ equals $\tfrac12\|X\|_{2,\sigma}^2\bigl(D(\eta_X\Vert\sigma)+\log\|X\|_{2,\sigma}^2\bigr)$, and the second term cancels the logarithm of the norm.

\subsection{Functional inequalities and entropy decay}\label{sec:notation-inequalities}
Let $\mathcal L$ be a primitive KMS-symmetric generator with invariant state $\sigma$. We use three functional inequalities for $\mathcal L$: the spectral gap inequality, the $L^2$ LSI, and the modified logarithmic Sobolev inequality. For each, we record its optimal constant and the decay estimate it yields. Theorem~\ref{thm:transfer} states how R\'enyi-2 contraction uniform over independent copies yields relative-entropy contraction.

\paragraph*{Spectral gap}
For an initial state $\rho_0$, define the centered relative density by
\[
 X_t=\Gamma_\sigma^{-1}\bigl(e^{t\mathcal L_*}(\rho_0-\sigma)\bigr).
\]
KMS detailed balance~\eqref{eqn:L-KMS} implies $\Gamma_\sigma^{-1}\mathcal L_*=\mathcal L\Gamma_\sigma^{-1}$, so $\partial_tX_t=\mathcal L(X_t)$. The Dirichlet form therefore gives the decay rate of the squared KMS norm:
\[
 \frac d{dt}\|X_t\|_{2,\sigma}^2=-2\mathcal E_{2,\sigma}(X_t).
\]
The KMS spectral gap of $\mathcal L$ is
\[
 \lambda(\mathcal L)
 =\inf_{\substack{X\ne0\\\Tr(\sigma X)=0}}
   \frac{\mathcal E_{2,\sigma}(X)}{\|X\|_{2,\sigma}^2}>0.
\]
Since $\Tr(\sigma X_t)=0$, the spectral gap gives the $L^2$ relaxation bound
\[
 \frac d{dt}\|X_t\|_{2,\sigma}^2\leq -2\lambda(\mathcal{L})\|X_t\|_{2,\sigma}^2\ \Rightarrow\ \|X_t\|_{2,\sigma}
 \leq e^{-\lambda(\mathcal L)t}
       \|X_0\|_{2,\sigma}
 \quad(t\geq0).
\]

\paragraph*{$L^2$ logarithmic Sobolev inequality}
The $L^2$ LSI constant $\alpha_2(\mathcal L)$ is the largest constant such that
\begin{equation}\label{eq:l2-lsi}
 \alpha_2(\mathcal L)\,\Ent_{2,\sigma}(X)\leq\mathcal E_{2,\sigma}(X)
 \qquad\text{for every }X>0.
\end{equation}
This is the $2$-log-Sobolev inequality of~\cite{KT}, with the same normalization. The following lemma shows that \eqref{eq:l2-lsi} yields exponential decay of the sandwiched R\'enyi-2 divergence.

\begin{lemma}[$L^2$ LSI and R\'enyi-2 decay {\cite[Theorems~3.2 and~4.3]{MHF}}]\label{lem:l2-renyi}\label{cor:endpoint-decay}
For every state $\rho$ and every $t\geq0$,
\[
 \widetilde D_2(e^{t\mathcal L_*}\rho\Vert\sigma)
 \leq e^{-\alpha_2(\mathcal L)t}\widetilde D_2(\rho\Vert\sigma).
\]
\end{lemma}

\paragraph*{Modified logarithmic Sobolev inequality}
The modified logarithmic Sobolev constant $\alpha_1(\mathcal L)$ is the largest constant such that
\begin{equation}\label{eq:mlsi-definition}
 2\alpha_1(\mathcal L)\,D(\rho\Vert\sigma)\leq-\Tr\bigl[\mathcal L_*(\rho)(\log\rho-\log\sigma)\bigr]
 \qquad\text{for every faithful state }\rho.
\end{equation}
The right-hand side is the entropy production $-\frac{d}{dt}D(e^{t\mathcal L_*}\rho\Vert\sigma)\big|_{t=0}$. The following lemma expresses this inequality as an exponential decay bound.

\begin{lemma}[MLSI and relative-entropy decay {\cite[Lemma~21]{KT}}]\label{lem:mlsi-decay}
{Let $\mathcal L$ satisfy KMS detailed balance with invariant state $\sigma$. Then~\eqref{eq:mlsi-definition} holds if and only if, for every faithful state $\rho$ and every $t\geq0$,
\[
 D(e^{t\mathcal L_*}\rho\Vert\sigma)
 \leq e^{-2\alpha_1(\mathcal L)t}D(\rho\Vert\sigma).
\]}
\end{lemma}

\section{Setup and main results}\label{sec:setup-main-result}
We specify the spin systems and Gibbs samplers under consideration and state our main results on entropy decay.

\paragraph*{Quantum spin systems}
We consider quantum spin systems on a finite region $\Lambda\subset\mathbb Z^d$ of the $d$-dimensional lattice. Each site carries $q$ qubits, so the Hilbert space is
\begin{equation}\label{eqn:hilbert-space}
 \mathcal H=\bigotimes_{x\in\Lambda}\mathbb C^{2^q}.
\end{equation}
For a region $I$, write $\mathcal H_I=\bigotimes_{x\in I}\mathbb C^{2^q}$, so $\mathcal B(\mathcal H_I)=\bigotimes_{x\in I}\mathcal B(\mathbb C^{2^q})$. Each factor $\mathcal B(\mathbb C^{2^q})$ has dimension $4^q$.
We choose the basis used for vectorization to be a product basis. Regrouping the doubled factors by site gives
\[
 \mathcal H_I^{\mathrm{dbl}}:=\bigotimes_{x\in I}(\mathbb C^{2^q}\otimes\mathbb C^{2^q}),
 \qquad \operatorname{vec}(\mathcal B(\mathcal H_I))=\mathcal H_I^{\mathrm{dbl}}.
\]
{For a region $I\subseteq\Lambda$, the support conditions for matrices $A\in\mathcal B(\mathcal H)$, superoperators $\mathcal A:\mathcal B(\mathcal H)\to\mathcal B(\mathcal H)$, and operators $\mathsf A\in\mathcal B(\mathcal H\otimes\mathcal H)$ on the doubled space are, respectively,}
\begin{equation}\label{eqn:support}
 \begin{aligned}
  \supp A\subseteq I
  &\quad\Longleftrightarrow\quad A=A_I\otimes I_{I^c},\\
  \supp\mathcal A\subseteq I
  &\quad\Longleftrightarrow\quad
    \mathcal A=\mathcal A_I\otimes\id_{\mathcal B(\mathcal H_{I^c})},\\
  \supp\mathsf A\subseteq I
  &\quad\Longleftrightarrow\quad
    \mathsf A=\mathsf A_I\otimes\one_{I^c}^{\mathrm{dbl}}.
 \end{aligned}
\end{equation}
Here $A_I$ acts on $\mathcal H_I$, $\mathcal A_I$ acts on $\mathcal B(\mathcal H_I)$, and $\mathsf A_I$ acts on $\mathcal H_I^{\mathrm{dbl}}$; the factors on $I^c$ are the corresponding identities.

We use the nearest-neighbor distance $\dist$ in $\mathbb Z^d$ and write $\Ball_x(r)=\{y\in\Lambda:\dist(x,y)\leq r\}$. We consider finite-range Hamiltonians with Hermitian interactions $h_Z$ supported on $Z\subseteq\Lambda$, satisfying
\begin{equation}\label{eq:growth}
 H=\sum_{\substack{Z\subseteq\Lambda\\\operatorname{diam}Z\leq R_0}}h_Z,\qquad
 \sup_{x\in\Lambda}\sum_{Z\ni x}\|h_Z\|_\infty\leq J.
\end{equation}
The range $R_0$ and strength $J$ are fixed.

Fix an inverse temperature $\beta>0$ and consider the Gibbs state
\[
 \sigma_\beta=\frac{e^{-\beta H}}{\Tr e^{-\beta H}}.
\]
We study two settings in this manuscript as examples:
\begin{itemize}
\item \emph{One-dimensional systems.} We consider $d=1$ at any fixed inverse temperature $0<\beta<\infty$. Following~\cite{BC2025}, by grouping a fixed number of neighboring sites and rescaling the interaction strength, we write the Hamiltonian on an open chain $\Lambda=\{1,\ldots,n\}\subset\mathbb Z$ as
\begin{equation}\label{eq:chain-setting}
 H=\sum_{i=1}^{n-1}h_{i,i+1},\qquad \|h_{i,i+1}\|_\infty\leq1.
\end{equation}
\item \emph{High-temperature systems.} On finite boxes $\Lambda\subset\mathbb Z^d$, we consider the finite-range Hamiltonians in \eqref{eq:growth} with $\beta J$ sufficiently small. The threshold depends only on $d,R_0,q$, as specified in Theorem~\ref{thm:main-mlsi}\textup{(ii)}.
\end{itemize}

Unless stated otherwise, constants are independent of the system size $|\Lambda|$. They may depend on $q,R_0,J,d$, the inverse temperature $\beta$, and the quantitative bounds in the stated assumptions.

\paragraph*{Quantum Gibbs samplers}
To prepare $\sigma_\beta$, we consider the Gaussian KMS Gibbs sampler of~\cite{CKG}, an instance of the KMS family in~\cite{DLL}, with coupling operators given by the nonidentity on-site Pauli matrices $S_{x,u}$, where $x$ labels the site and $u$ ranges over its $4^q-1$ nonidentity Pauli matrices.
The generator acting on states is
\begin{equation}\label{eq:CKG}
 \cL_{\beta,*}(\rho)=-i[H_\beta^{\mathrm{coh}},\rho]+
 \sum_{x,u}\int_\R\gamma_\beta(\omega)
 \left(J_{x,u}^{\beta}(\omega)\rho J_{x,u}^{\beta}(\omega)^*
      -\tfrac12\{J_{x,u}^{\beta}(\omega)^*J_{x,u}^{\beta}(\omega),\rho\}\right)d\omega,
\end{equation}
where $\gamma_\beta(\omega)=e^{-(\beta\omega+1)^2/2}$ and the Gaussian-filtered jump operators are
\begin{equation}\label{eq:CKG-jump}
 J_{x,u}^{\beta}(\omega)
 =\frac{(2/\pi)^{1/4}}{\sqrt{2\pi\beta}}
       \int_\R e^{-t^2/\beta^2-i\omega t}e^{itH}S_{x,u}e^{-itH}\,dt.
\end{equation}
The coherent correction is defined by
\begin{equation}\label{eq:CKG-correction}
 H_\beta^{\mathrm{coh}}=\frac i2\tanh\!\left(\frac\beta4\ad_H\right)Q_\beta,\qquad
 Q_\beta=\sum_{x,u}\int_\R\gamma_\beta(\omega)J_{x,u}^{\beta}(\omega)^*J_{x,u}^{\beta}(\omega)\,d\omega,
\end{equation}
with $\ad_H(A)=[H,A]$.

This sampler is primitive and KMS-symmetric with invariant state $\sigma_\beta$~\cite[Theorem~I.1]{CKG}.

Let $\mathcal L_\beta$ be the trace adjoint of $\mathcal L_{\beta,*}$, acting on observables. The corresponding operators on $\mathcal H\otimes\mathcal H$ are defined by
\[
 \mathsf{L}_\beta|Y\rangle=|\mathcal L_\beta(Y)\rangle,\qquad
 \mathsf{L}_{\beta,*}|Y\rangle=|\mathcal L_{\beta,*}(Y)\rangle.
\]
The weighted generator on $\mathcal B(\mathcal H)$ and its corresponding operator on $\mathcal H\otimes\mathcal H$ are
\[
 \mathcal L_{\sigma_\beta}(Y):=\Gamma_{\sigma_\beta}^{1/2}\bigl(\mathcal L_\beta(\Gamma_{\sigma_\beta}^{-1/2}(Y))\bigr),
 \qquad \mathsf{L}_{\sigma_\beta}|Y\rangle:=|\mathcal L_{\sigma_\beta}(Y)\rangle.
\]
By \eqref{eq:parent-action}, $-\mathsf{L}_{\sigma_\beta}$ is positive semidefinite; it is the parent Hamiltonian of the sampler.

\smallskip

We recall the decomposition of the sampler's Dirichlet form into squared norms. Building on the Bohr-frequency form of~\cite[Lemma~C.2]{RFAThermalization}, Chen and Rouz\'e~\cite[Lemma~X.2]{CR2025} wrote the Dirichlet form of the Gaussian KMS Gibbs sampler as an integral of squared KMS norms of commutators (see also~\cite[Lemma~VI.1]{BC2025}). We use the corresponding factorization of the parent Hamiltonian and its parametrization from~\cite{LJL}.
\begin{lemma}[Dirichlet-form decomposition {\cite[Lemma~X.2]{CR2025}; \cite[Theorem~3, Proposition~6, and Appendix~A.4.2]{LJL}}]\label{lem:native-dirichlet-decomposition}
For the Gaussian KMS Gibbs sampler \eqref{eq:CKG}, the weighted generator and its Dirichlet form satisfy, for $X\in\mathcal B(\mathcal H)$,
\begin{equation}\label{eq:native-dirichlet-form}
 \begin{aligned}
 -\mathsf{L}_{\sigma_\beta}&=\sum_{x,u}\int_{\R^2}
       \widehat{\mathsf{B}}_{\beta,x,u}(v,z)^*\widehat{\mathsf{B}}_{\beta,x,u}(v,z)\,d\mu(v,z),\\
 \mathcal E_{2,\sigma_\beta}(X)
 &=\sum_{x,u}\int_{\R^2}
       \bigl\|\widehat{\mathsf{B}}_{\beta,x,u}(v,z)\sGamma_{\sigma_\beta}^{1/2}|X\rangle\bigr\|_2^2\,d\mu(v,z),
 \end{aligned}
\end{equation}
where the operators $\widehat{\mathsf{B}}_{\beta,x,u}(v,z)$ on $\mathcal H\otimes\mathcal H$ are defined by
\[
 \widehat{\mathsf{B}}_{\beta,x,u}(v,z)|Y\rangle
 :=|\widehat{\mathcal B}_{\beta,x,u}(v,z)(Y)\rangle,
 \qquad Y\in\mathcal B(\mathcal H),
\]
with the superoperators and measure given by
\begin{equation}\label{eq:native-gradient}
 \begin{aligned}
 \widehat{\mathcal B}_{\beta,x,u}(v,z)(Y)
 &=\bigl(e^{i\beta vH}\widehat J^\beta_{x,u}(z)e^{-i\beta vH}\bigr)Y-Y\bigl(e^{i\beta vH}\sigma_\beta^{-1/2}
       \widehat J^\beta_{x,u}(z)\sigma_\beta^{1/2}e^{-i\beta vH}\bigr),\\
 d\mu(v,z)&=\frac{e^{-(z+1)^2/2}}{\cosh(2\pi v)}\,dv\,dz.
 \end{aligned}
\end{equation}
Here the rescaled jumps are
\begin{equation}\label{eq:native-rescaled-jump}
 \widehat J^\beta_{x,u}(z)
 :=\beta^{-1/2}J^\beta_{x,u}(z/\beta)
 =\frac{(2/\pi)^{1/4}}{\sqrt{2\pi}}
   \int_\R e^{-r^2-izr}e^{i\beta rH}S_{x,u}e^{-i\beta rH}\,dr.
\end{equation}
Each operator $\widehat{\mathsf{B}}_{\beta,x,u}(v,z)$ satisfies
\[
 \widehat{\mathsf{B}}_{\beta,x,u}(v,z)|\Omega_{\sigma_\beta}\rangle=0.
\]
\end{lemma}

\paragraph*{Main result}
The Gaussian KMS Gibbs sampler satisfies entropy decay at a rate independent of system size in both settings above.

\begin{theorem}[Size-independent MLSI for the Gaussian KMS Gibbs sampler]\label{thm:main-mlsi}
For the Gaussian KMS Gibbs sampler~\cite{CKG} in \eqref{eq:CKG}, there is a rate $\kappa>0$ in each of the following settings:
\begin{enumerate}[label=\textup{(\roman*)}]
\item the one-dimensional systems in \eqref{eq:chain-setting} at any fixed $0<\beta<\infty$, where $\kappa$ depends only on $\beta,q$;
\item the finite-range systems in \eqref{eq:growth} on finite boxes $\Lambda\subset\mathbb Z^d$ with $\beta J\leq\beta_0$, where $\beta_0>0$ and $\kappa$ depend only on $d,R_0,q$.
\end{enumerate}
For every state $\rho$ and $t\geq0$,
\[
 \begin{aligned}
 \widetilde D_2(e^{t\mathcal L_{\beta,*}}\rho\Vert\sigma_\beta)
 &\leq e^{-\kappa t}\widetilde D_2(\rho\Vert\sigma_\beta),\\
 D(e^{t\mathcal L_{\beta,*}}\rho\Vert\sigma_\beta)
 &\leq e^{-\kappa t}D(\rho\Vert\sigma_\beta).
 \end{aligned}
\]
In particular, $\alpha_1(\mathcal L_\beta)\geq\kappa/2$; moreover, $\alpha_2(\mathcal L_\beta)\geq\kappa$.
\end{theorem}

\paragraph*{From R\'enyi-2 decay to MLSI}
The proof of Theorem~\ref{thm:main-mlsi} passes from R\'enyi-2 decay to relative-entropy decay using the following theorem, proved in Section~\ref{sec:transfer}.
\begin{theorem}[Uniform R\'enyi-2 contraction implies relative-entropy contraction]
\label{thm:transfer}
Let $\Phi$ be a channel fixing a faithful state $\sigma$.
Suppose that, for a fixed $a\in[0,1]$, all $m\geq1$ and
all states $\nu$ satisfy
\begin{equation}\label{eq:transfer-renyi-contraction}
 \Dtwo(\Phi^{\otimes m}(\nu)\Vert\sigma^{\otimes m})
       \leq a\Dtwo(\nu\Vert\sigma^{\otimes m}).
\end{equation}
Then $D(\Phi(\rho)\Vert\sigma)\leq aD(\rho\Vert\sigma)$
for every state $\rho$.
\end{theorem}
The proof of Theorem~\ref{thm:main-mlsi} is completed in Section~\ref{sec:main-proof}, after establishing the transport framework and verifying its hypotheses in the two settings.

\section{Proof idea: transporting \texorpdfstring{$L^2$ LSIs}{L2 LSIs} along a unitary path}\label{sec:vectorization}

Theorem~\ref{thm:main-mlsi} follows from an $L^2$ LSI for $\mathcal L_\beta$ whose constant is independent of the system size. Section~\ref{sec:strategy-mlsi} explains this reduction and outlines the proof of Theorem~\ref{thm:transfer}, which converts R\'enyi-2 decay into relative-entropy decay. We prove the $L^2$ LSI by transporting it from infinite temperature, where the Gibbs state is a product state, along a unitary path in the doubled Hilbert space. Sections~\ref{sec:vectorized-target}--\ref{sec:transport-principle} set up this transport and state the transport principle, Proposition~\ref{prop:transport}. Section~\ref{sec:proof-outline} outlines how the remaining sections carry out the proof.

\subsection{From the $L^2$ LSI to MLSI}\label{sec:strategy-mlsi}
We pass from~\eqref{eq:l2-lsi} to MLSI through R\'enyi-2 decay. By Lemma~\ref{lem:l2-renyi}, \eqref{eq:l2-lsi} yields R\'enyi-2 decay at rate $\alpha_2(\mathcal L_\beta)$; Theorem~\ref{thm:transfer} converts R\'enyi-2 decay into relative-entropy decay; and by Lemma~\ref{lem:mlsi-decay}, relative-entropy decay at rate $\kappa$ gives $\alpha_1(\mathcal L_\beta)\geq\kappa/2$.

The proof of Theorem~\ref{thm:transfer} compares the two divergences on many copies. They are ordered as $D\leq\widetilde D_2\leq D_{\max}$, where $D_{\max}$ is the max-relative entropy (see Section~\ref{sec:transfer} for details). By the quantum asymptotic equipartition property~\cite[Theorem~14]{DMHB}, $\rho^{\otimes m}$ can be approximated by a state $\nu_m$ whose max-relative entropy per copy satisfies $\frac1mD_{\max}(\nu_m\Vert\sigma^{\otimes m})\lesssim D(\rho\Vert\sigma)$ for large $m$. Applying the R\'enyi-2 contraction~\eqref{eq:transfer-renyi-contraction} assumed in Theorem~\ref{thm:transfer} to the approximating state $\nu_m$ gives
\[
 D(\Phi(\rho)\Vert\sigma)
 \lesssim\frac1m\widetilde D_2(\Phi^{\otimes m}(\nu_m)\Vert\sigma^{\otimes m})
 \leq\frac am\widetilde D_2(\nu_m\Vert\sigma^{\otimes m})
 \leq\frac amD_{\max}(\nu_m\Vert\sigma^{\otimes m})
 \lesssim aD(\rho\Vert\sigma).
\]
The first step uses $\nu_m\approx\rho^{\otimes m}$, superadditivity of relative entropy, and $D\leq\widetilde D_2$ to pass to a single copy. The second step is the contraction hypothesis applied to $\nu_m$, the third uses $\widetilde D_2\leq D_{\max}$. The final step follows from the quantum asymptotic equipartition property. Since approximating a product state may introduce correlations, the contraction estimate must hold even when the input state $\nu$ is correlated across copies. See Section~\ref{sec:transfer} for the full proof.

To apply Theorem~\ref{thm:transfer} to our Gibbs sampler, we need a R\'enyi-2 decay rate that remains unchanged when we take independent copies of the system. On $m$ independent copies, the evolution $(e^{t\mathcal L_{\beta,*}})^{\otimes m}$ is generated by the Gaussian KMS Gibbs sampler of the Hamiltonian $H^{(m)}=\sum_{j=1}^mH^{(j)}$ of the copies, and $H^{(m)}$ belongs to the same class with the same local parameters (Section~\ref{sec:independent-copies}). A rate that depends only on these parameters thus applies to $H^{(m)}$ for every $m$ and to every state of the copies, and Corollary~\ref{cor:class-renyi-mlsi} gives MLSI. Our task is therefore to bound $\alpha_2(\mathcal L_\beta)$ in~\eqref{eq:l2-lsi} from below, independently of the system size.

\subsection{The vectorized target inequality}\label{sec:vectorized-target}
We first rewrite~\eqref{eq:l2-lsi} for $\mathcal L_\beta$ as an inequality between a relative entropy and a quadratic form on the doubled Hilbert space. Both sides of~\eqref{eq:l2-lsi} are homogeneous of degree two in $X$, so we may assume $\|X\|_{2,\sigma_\beta}=1$. Define $Y=\Gamma_{\sigma_\beta}^{1/2}(X)$, equivalently $|Y\rangle=\sGamma_{\sigma_\beta}^{1/2}|X\rangle$. Then $Y>0$, $\|Y\|_2=\|X\|_{2,\sigma_\beta}=1$, and $\eta_X=Y^2$ in~\eqref{eqn:l2-entropy}. Since $\sGamma_{\sigma_\beta}^{1/2}$ is self-adjoint, \eqref{eqn:l2-entropy} and~\eqref{eqn:l2-dirichlet} become
\[
 2\Ent_{2,\sigma_\beta}(X)=D(Y^2\Vert\sigma_\beta),\qquad
 \mathcal E_{2,\sigma_\beta}(X)=\langle Y|(-\mathsf{L}_{\sigma_\beta})|Y\rangle.
\]
It therefore suffices to find a constant $C_\beta$, independent of $|\Lambda|$, such that
\begin{equation}\label{eq:vectorized-target}
 D(Y^2\Vert\sigma_\beta)\leq C_\beta\langle Y|(-\mathsf{L}_{\sigma_\beta})|Y\rangle
 \qquad\text{for every }Y=Y^*\text{ with }\|Y\|_2=1.
\end{equation}
This inequality gives $\alpha_2(\mathcal L_\beta)\geq2/C_\beta$, so it remains to bound $C_\beta$ independently of the system size.

\subsection{Interpolation from infinite temperature}\label{sec:temperature-path}
We prove \eqref{eq:vectorized-target} by transporting the inequality from infinite temperature. To this end, we interpolate in the inverse temperature. For $0\leq s\leq\beta$, let
$$
 \sigma_s=\frac{e^{-sH}}{\Tr e^{-sH}},
$$
and for $0<s\leq\beta$, let $\mathcal L_{s,*}$ be the generator \eqref{eq:CKG} with $\beta$ replaced by $s$ throughout \eqref{eq:CKG}--\eqref{eq:CKG-correction}. Thus $\mathcal L_{s,*}$ is the sampler for $\sigma_s$: its filter, jump operators, and coherent correction all depend on $s$. The maps $\mathcal L_s$, $\mathsf L_s$, $\mathcal L_{\sigma_s}$, $\mathsf L_{\sigma_s}$, the rescaled jumps $\widehat J^s_{x,u}$, and the operators $\widehat{\mathsf B}_{s,x,u}(v,z)$ are defined by the same replacement. At $s=0$, we define $\mathcal L_{0,*}$ by taking the limit $s\to0^+$:
\[
 \mathcal L_{0,*}(\rho):=\lim_{s\to0^+}\mathcal L_{s,*}(\rho)
 =\frac{e^{-1/4}}{\sqrt2}\sum_{x,u}\bigl(S_{x,u}\rho S_{x,u}-\rho\bigr).
\]
This is a sum of single-site depolarizing generators whose evolution converges to the maximally mixed state $\sigma_0=I/2^{q|\Lambda|}$ and for which the desired $L^2$ LSI holds (Lemma~\ref{lem:product-entropy}):
\begin{equation}\label{eq:product-state-reference}
 D(Y^2\Vert\sigma_0)\leq C_0\langle Y|(-\mathsf{L}_{\sigma_0})|Y\rangle
 \quad(Y=Y^*,\ \|Y\|_2=1),\qquad C_0=\frac{2\sqrt2e^{1/4}q}{4^q}.
\end{equation}

\subsection{Transport along a unitary path}\label{sec:transport-principle}
To carry~\eqref{eq:product-state-reference} from $s=0$ to $s=\beta$, we seek a Hermiticity-preserving unitary path $\mathsf{U}_s$ on the doubled Hilbert space, with $\mathsf{U}_0=\one$ and $\mathsf{U}_s|\Omega_{\sigma_0}\rangle=|\Omega_{\sigma_s}\rangle$. Define the associated superoperator $\mathcal U_s$ by $|\mathcal U_s(Y)\rangle=\mathsf{U}_s|Y\rangle$; it preserves the Hilbert--Schmidt inner product. For a Hermitian $Y_0$ with $\|Y_0\|_2=1$, the matrices $Y_s=\mathcal U_s(Y_0)$, equivalently $|Y_s\rangle=\mathsf{U}_s|Y_0\rangle$, are Hermitian with $\|Y_s\|_2=1$, so $Y_s^2$ is a state. To maintain the comparison of entropy and energy in the target inequality~\eqref{eq:vectorized-target} along the entire unitary trajectory, we need to control how both quantities change along the path. Writing $E_s=\langle Y_s|(-\mathsf{L}_{\sigma_s})|Y_s\rangle$, it suffices to establish the following derivative bounds:

\begin{equation}\label{eq:variation-targets}
 \begin{aligned}
 \left|\frac d{ds}D(Y_s^2\Vert\sigma_s)\right|
 &\leq C_S\langle Y_s|(-\mathsf{L}_{\sigma_s})|Y_s\rangle,\\
 \left|\frac d{ds}\langle Y_s|(-\mathsf{L}_{\sigma_s})|Y_s\rangle\right|
 &\leq C_E\langle Y_s|(-\mathsf{L}_{\sigma_s})|Y_s\rangle,
 \end{aligned}
 \qquad 0\leq s\leq\beta,
\end{equation}
for every Hermitian $Y_0$ with $\|Y_0\|_2=1$, with $C_S,C_E\geq0$. This transports the product-state inequality~\eqref{eq:product-state-reference} to $s=\beta$, as illustrated in Figure~\ref{fig:unitary-transport} and stated in the following proposition.

\begin{proposition}[Transport of the $L^2$ LSI]\label{prop:transport}\label{thm:transport-principle}
Suppose $\mathsf{U}_s$ is a $C^1$ Hermiticity-preserving unitary path with $\mathsf{U}_0=\one$ and $\mathsf{U}_s|\Omega_{\sigma_0}\rangle=|\Omega_{\sigma_s}\rangle$, satisfying \eqref{eq:variation-targets} for every Hermitian $Y_0$ with $\|Y_0\|_2=1$, with constants $C_S,C_E$. Then \eqref{eq:vectorized-target} holds with $C_\beta=(C_0+\beta C_S)e^{\beta C_E}$, where $C_0$ is from \eqref{eq:product-state-reference}. Consequently,
\begin{equation}\label{eq:transport-lsi}
 \alpha_2(\mathcal L_\beta)\geq\frac{2e^{-\beta C_E}}{C_0+\beta C_S}.
\end{equation}
\end{proposition}
\begin{proof}
Fix a Hermitian $Y_0$ with $\|Y_0\|_2=1$ and let $|Y_s\rangle=\mathsf{U}_s|Y_0\rangle$. By the second inequality in \eqref{eq:variation-targets}, $\frac{d}{ds}E_s\geq-C_EE_s$, so $s\mapsto e^{C_Es}E_s$ is nondecreasing and
\[
 E_s\leq e^{(\beta-s)C_E}E_\beta,\qquad 0\leq s\leq\beta.
\]
Integrating the first inequality in \eqref{eq:variation-targets} from $0$ to $\beta$, and applying \eqref{eq:product-state-reference} to $Y_0$, gives
\[
 D(Y_\beta^2\Vert\sigma_\beta)
 \leq C_0E_0+C_S\int_0^\beta E_s\,ds
 \leq C_0e^{\beta C_E}E_\beta+C_S\int_0^\beta e^{(\beta-s)C_E}E_\beta\,ds
 \leq(C_0+\beta C_S)e^{\beta C_E}E_\beta.
\]
Now let $Y$ be Hermitian with $\|Y\|_2=1$ and take $|Y_0\rangle=\mathsf{U}_\beta^*|Y\rangle$. Since $\mathsf U_\beta$ is unitary and maps the real space of Hermitian matrices into itself, it maps that space onto itself, so $\mathsf U_\beta^*=\mathsf U_\beta^{-1}$ also preserves Hermiticity. Hence $Y_0$ is Hermitian with $\|Y_0\|_2=1$, and $Y_\beta=Y$. This proves \eqref{eq:vectorized-target}, and \eqref{eq:transport-lsi} follows from $\alpha_2(\mathcal L_\beta)\geq2/C_\beta$.
\end{proof}

\begin{figure}[htbp]
\centering
\begin{tikzpicture}[
 font=\footnotesize,
 endpoint/.style={draw,rounded corners=3pt,align=center,minimum width=2.5cm,minimum height=.8cm,inner sep=4pt},
 edgelabel/.style={fill=white,inner sep=4pt,align=center},
 arr/.style={-{Latex[length=2mm]},thick}]
\node[endpoint,draw=blue!75!black,fill=blue!7] (energy0) at (0,0) {$E_0$};
\node[endpoint,draw=blue!75!black,fill=blue!7] (energybeta) at (9.5,0) {$E_\beta$};
\node[endpoint,draw=orange!75!black,fill=orange!9] (entropy0) at (0,-3.1) {$D(Y_0^2\Vert\sigma_0)$};
\node[endpoint,draw=orange!75!black,fill=orange!9] (entropybeta) at (9.5,-3.1) {$D(Y_\beta^2\Vert\sigma_\beta)$};
\node[anchor=east,text=blue!75!black] at (-1.5,0) {\textbf{Energy}};
\node[anchor=east,text=orange!75!black] at (-1.5,-3.1) {\textbf{Entropy}};
\draw[arr,blue!75!black] (energy0.east) -- node[edgelabel,above,text=black] (energybound)
 {$\left|\dfrac{d}{ds}E_s\right|\leq C_EE_s$} (energybeta.west);
\node[align=center,inner sep=1pt,above=1mm of energybound] (transport) {\textbf{Unitary transport}\\[3pt]$|Y_s\rangle=\mathsf{U}_s|Y_0\rangle$};
\node[inner sep=1pt,above=2mm of energy0] {$s=0$};
\node[inner sep=1pt,above=2mm of energybeta] {$s=\beta$};
\draw[arr,orange!75!black] (entropy0.east) -- node[edgelabel,below,text=black]
 {$\left|\dfrac{d}{ds}D(Y_s^2\Vert\sigma_s)\right|\leq C_SE_s$} (entropybeta.west);
\draw[arr] (energy0.south) -- node[edgelabel,right,text width=3.9cm]
 {Product-state comparison\\[5pt]$D(Y_0^2\Vert\sigma_0)\leq C_0E_0$} (entropy0.north);
\draw[arr] (energybeta.south) -- node[edgelabel,left,text width=4.3cm]
 {Transported comparison\\[5pt]$D(Y_\beta^2\Vert\sigma_\beta)\leq C_\beta E_\beta$} (entropybeta.north);
\end{tikzpicture}
\caption{Transporting the product-state inequality along a unitary path (Proposition~\ref{prop:transport}). Here $E_s=\langle Y_s|(-\mathsf{L}_{\sigma_s})|Y_s\rangle$ and $C_\beta=(C_0+\beta C_S)e^{\beta C_E}$.}
\label{fig:unitary-transport}
\end{figure}
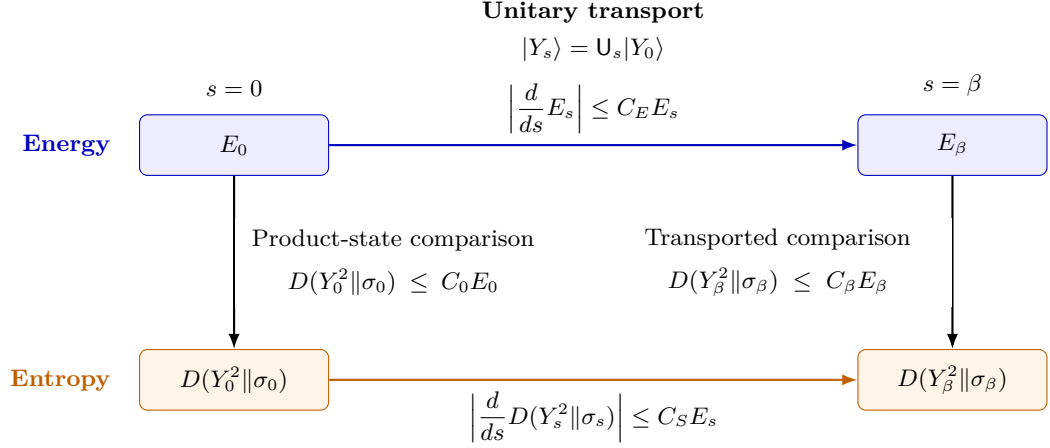

\subsection{Outline of the proof}\label{sec:proof-outline}
In the rest of the paper, we construct a unitary flow, prove that it satisfies the required bounds on energy and entropy variation in \eqref{eq:variation-targets}, and pass from the resulting R\'enyi-2 decay to MLSI, as follows:
\begin{itemize}
\item \emph{Construct the unitary path.} In Section~\ref{sec:poisson-construction}, we construct the generator $\mathsf{G}_s$ and its unitary flow $\mathsf{U}_s$, which carries the infinite-temperature Gibbs vector $|\Omega_{\sigma_0}\rangle$ to $|\Omega_{\sigma_s}\rangle$ (Proposition~\ref{prop:exact-continuation}). A centered Poisson equation and a Sylvester equation determine the transport generator. Its integral formula supplies the locality estimates, and the classical construction explains the commutator form. Section~\ref{sec:flow-locality} states the local rapid-mixing condition (Condition~\ref{cond:transport-mixing}) and proves locality of the transport generator in Proposition~\ref{prop:flow-locality}.
\item \emph{Control energy and entropy along the path.} In Sections~\ref{sec:derivative-comparison}-\ref{sec:conservation-estimates}, we calculate the derivatives of the energy and entropy along $\mathsf{U}_s$. Our goal is to establish the bounds in \eqref{eq:variation-targets}. The calculations reduce this task to controlling the combined action of local operators, denoted by $\mathsf{R}_{s,\alpha}$, through a bound of the form
\[
 \sum_\alpha\|\mathsf{R}_{s,\alpha}|Y\rangle\|_2^2
 \lesssim\langle Y|(-\mathsf{L}_{\sigma_s})|Y\rangle.
\]
The local action-energy inequality \eqref{eq:comparison-target} in Condition~\ref{cond:local-action} makes this requirement precise. Together with locality of the derivative terms, it gives the energy and entropy variation bounds as summarized in Figure~\ref{fig:condition-implications}. Finally, we provide a sufficient condition, i.e. Condition~\ref{ass:local-gap-state-domination}, to verify Condition~\ref{cond:local-action}. This new condition is easier to check in practice and is verified for the two regimes of interest in Section~\ref{sec:regime-verification}.

\item \emph{Obtain a general criterion for R\'enyi-2 decay.} In Section~\ref{sec:renyi-criterion}, we combine these steps: Proposition~\ref{prop:transport} transports the $L^2$ LSI from $s=0$ to $s=\beta$, giving the sampler's $L^2$ LSI and hence size-independent R\'enyi-2 decay (Theorem~\ref{prop:general-criterion}).
\item \emph{Pass from R\'enyi-2 decay to MLSI.} Section~\ref{sec:transfer} proves Theorem~\ref{thm:transfer}. For a Hamiltonian class closed under independent copies, with tensor-sum samplers and bounds uniform in the number of copies, Corollary~\ref{cor:class-renyi-mlsi} in Section~\ref{sec:independent-copies} then gives relative-entropy decay and MLSI. Section~\ref{sec:main-proof} combines these results with the verifications in Section~\ref{sec:regime-verification} to prove Theorem~\ref{thm:main-mlsi}.
\end{itemize}
The technical estimates are deferred to the appendices.

\section{Constructing the unitary from local response}\label{sec:poisson-construction}
Proposition~\ref{prop:transport} asks us to transport a prescribed path of Gibbs vectors by a unitary whose action on entropy and energy can be controlled. The construction consists of a Poisson equation and an invertible symmetrized multiplication. We also calculate the integral formula needed for later estimation and its classical motivation.

\subsection{The continuation equation}\label{sec:continuation-equation}

\paragraph*{Constructing the quantum transport}
Recall from~\eqref{eq:omega_sigma} and~\eqref{eq:weighted-generator} that $\Omega_{\sigma_s}=\sqrt{\sigma_s}$ and that $\mathsf{L}_{\sigma_s}$ is the weighted generator. Since $\|\Omega_{\sigma_s}\|_2=1$ and $\Omega_{\sigma_s}$ is Hermitian, $\langle\Omega_{\sigma_s}|\partial_s|\Omega_{\sigma_s}\rangle=0$. Primitivity makes $-\mathsf{L}_{\sigma_s}$ invertible on the orthogonal complement of $|\Omega_{\sigma_s}\rangle$. Define Hermitian matrices $T_s$ and $f_s$ as the unique solutions to the centered Poisson equation and the Sylvester equation:
\begin{equation}\label{eq:continuation-poisson}
 \begin{aligned}
 \partial_s|\Omega_{\sigma_s}\rangle&=-\mathsf{L}_{\sigma_s}|T_s\rangle,
 &\langle\Omega_{\sigma_s}|T_s\rangle&=0,\\
 \mathcal M_{f_s}(\Omega_{\sigma_s})&=\tfrac12(f_s\Omega_{\sigma_s}+\Omega_{\sigma_s}f_s)=T_s,
 &\mathsf{M}_{f_s}|\Omega_{\sigma_s}\rangle&=|T_s\rangle.
 \end{aligned}
\end{equation}
The Poisson equation has a unique Hermitian solution in that orthogonal complement because $-\mathsf{L}_{\sigma_s}$ is invertible on this subspace. Because $\Omega_{\sigma_s}>0$, the Sylvester equation gives a unique Hermitian $f_s$. Moreover, $\Tr(\sigma_sf_s)=\langle\Omega_{\sigma_s}|T_s\rangle=0$.

Define the matrix transport generator and its vectorized representative by
\begin{equation}\label{eq:continuation-generator}
 \mathcal G_s:=[\mathcal M_{f_s},\mathcal L_{\sigma_s}],\qquad
 \mathsf{G}_s:=[\mathsf{M}_{f_s},\mathsf{L}_{\sigma_s}],\qquad
 \mathsf{G}_s|Y\rangle=|\mathcal G_s(Y)\rangle.
\end{equation}
Both factors in the vector commutator are self-adjoint and preserve Hermiticity, so $\mathsf{G}_s$ is skew-adjoint and preserves Hermiticity. Because $\mathsf{L}_{\sigma_s}|\Omega_{\sigma_s}\rangle=0$, the Poisson equation gives
\[
 \mathsf{G}_s|\Omega_{\sigma_s}\rangle
 =-\mathsf{L}_{\sigma_s}\mathsf{M}_{f_s}|\Omega_{\sigma_s}\rangle
 =-\mathsf{L}_{\sigma_s}|T_s\rangle=\partial_s|\Omega_{\sigma_s}\rangle.
\]
Consequently, the evolution
\begin{equation}\label{eq:continuation-unitary}
 \mathsf{U}_s'=\mathsf{G}_s\mathsf{U}_s,\qquad \mathsf{U}_0=\one,
\end{equation}
is a Hermiticity-preserving unitary carrying $|\Omega_{\sigma_0}\rangle$ to $|\Omega_{\sigma_s}\rangle$. Its matrix counterpart satisfies $\partial_sY_s=\mathcal G_s(Y_s)$ for $Y_s=\mathcal U_s(Y_0)$.

\begin{remark}[Classical counterpart]\label{rem:classical-counterpart}

For a classical Gibbs density $p_s=Z_s^{-1}e^{-sH}p_0$ on $\mathbb R^d$, a gradient velocity $v_s=\nabla\phi_s$ transports $p_s$ when
\[
 \partial_sp_s+\nabla\!\cdot(p_s\nabla\phi_s)=0,
 \qquad
 \mathcal L_s\phi_s=H-\int Hp_s\,dx,
 \qquad \mathcal L_s=\frac1{p_s}\nabla\!\cdot(p_s\nabla).
\]
The centered potential is $\phi_s=-\int_0^\infty e^{t\mathcal L_s}(H-\int Hp_s\,dx)\,dt$ when the integral converges~\cite{Moser}. Passing to square-root densities gives $\partial_s\sqrt{p_s}=G_s\sqrt{p_s}$, where $M_h\psi=h\psi$ denotes scalar multiplication on classical $L^2(dx)$ and $L_{p_s}:=M_{\sqrt{p_s}}\mathcal L_sM_{\sqrt{p_s}}^{-1}$,
\begin{equation}\label{eq:classical-commutator}
 G_s=[M_{\phi_s/2},L_{p_s}]
 =-\nabla\phi_s\cdot\nabla-\tfrac12\Delta\phi_s.
\end{equation}
Thus the quantum choice $\mathsf{G}_s=[\mathsf{M}_{f_s},\mathsf{L}_{\sigma_s}]$ retains the commutator construction on square roots, with scalar multiplication replaced by symmetrized multiplication.
\end{remark}

\begin{remark}[Transported entropy]\label{rem:transported-entropy}
For a Hermitian initial matrix $Y_0$ with $\|Y_0\|_2=1$, the flow of Proposition~\ref{prop:exact-continuation} preserves $\|Y_s\|_2=1$. Thus $\eta_s=Y_s^2$ is a state, and the entropy along the path is $D(\eta_s\Vert\sigma_s)=D(Y_s^2\Vert\sigma_s)$.
For $Y=\Gamma_{\sigma_s}^{1/2}(X)$ with $X>0$ and $\|X\|_{2,\sigma_s}=1$, we have $D(Y^2\Vert\sigma_s)=2\Ent_{2,\sigma_s}(X)$ by~\eqref{eqn:l2-entropy}. Here $s$ varies the inverse temperature, and $\eta_s$ is obtained by squaring the transported matrix. Relative-entropy decay for the physical evolution $e^{t\mathcal L_{\beta,*}}\rho$ at fixed temperature is derived later from R\'enyi-2 decay by the tensor-power argument in Section~\ref{sec:renyi-to-mlsi}.

For classical probability densities transported by the same smooth invertible flow, relative entropy is preserved. For the Gaussian KMS Gibbs sampler and the unitary path constructed here, Appendix~\ref{sec:entropy-reduction} computes the quantum remainder: the entropy derivative has the bilinear representation \eqref{eq:entropy-integral-representation} in weighted commutator actions. Each such map annihilates the Gibbs vector, as expressed in \eqref{eq:derivative-gibbs-cancellation}.
\end{remark}

\paragraph*{The integral formula for $T_s$ and $f_s$}
The inverse of $-\mathsf{L}_{\sigma_s}$ on the orthogonal complement of $|\Omega_{\sigma_s}\rangle$ has the semigroup representation
\begin{equation}\label{eq:continuation-integral}
 \begin{aligned}
 |T_s\rangle&=\int_0^\infty e^{t\mathsf{L}_{\sigma_s}}\partial_s|\Omega_{\sigma_s}\rangle\,dt,\\
 T_s&=\Gamma_{\sigma_s}^{1/2}\left(\int_0^\infty e^{t\mathcal L_s}
       \bigl(\Gamma_{\sigma_s}^{-1/2}(\partial_s\Omega_{\sigma_s})\bigr)\,dt\right).
 \end{aligned}
\end{equation}

For $\sigma_s=e^{-sH}/Z_s$, one has $\partial_s\Omega_{\sigma_s}=-\tfrac12(H-\Tr(\sigma_sH)I)\Omega_{\sigma_s}$. Inverting the Sylvester equation after Gibbs weighting yields
\begin{equation}\label{eq:continuation-sylvester}
 f_s=-\frac12\sech\!\left(\frac{s}{4}\ad_H\right)
 \int_0^\infty e^{t\mathcal L_s}
       \bigl(H-\Tr(\sigma_sH)I\bigr)\,dt.
\end{equation}
This formula separates the two analytic inputs: mixing controls the Poisson integral, and the $\sech$ filter expresses the change of weighting through real-time Hamiltonian evolution.

\begin{proposition}[Poisson-generated continuation]\label{prop:exact-continuation}
For the Gibbs path $\sigma_s$ and the Gaussian KMS Gibbs sampler \eqref{eq:CKG} of Section~\ref{sec:setup-main-result}, the construction \eqref{eq:continuation-poisson}--\eqref{eq:continuation-integral} defines a $C^1$ Hermiticity-preserving unitary path $\mathsf{U}_s$, with $\mathsf{U}_0=\one$, satisfying
\begin{equation}\label{eq:continuation-derivative}
 \mathsf{U}_s|\Omega_{\sigma_0}\rangle=|\Omega_{\sigma_s}\rangle,\qquad 0\leq s\leq\beta.
\end{equation}
\end{proposition}
\begin{proof}[Proof of Proposition~\ref{prop:exact-continuation}]
We verify that the first line of \eqref{eq:continuation-integral} gives the Hermitian solution $T_s$ of the centered Poisson equation, that the corresponding Hermitian $f_s$ is given by \eqref{eq:continuation-sylvester} and depends continuously on $s$, and that the resulting path $\mathsf U_s$ is unitary, preserves Hermiticity, and carries $|\Omega_{\sigma_0}\rangle$ to $|\Omega_{\sigma_s}\rangle$.
By \eqref{eq:parent-action} and the invertibility noted above, $\mathsf{L}_{\sigma_s}$ is strictly negative on the orthogonal complement of $|\Omega_{\sigma_s}\rangle$, an invariant subspace containing $\partial_s|\Omega_{\sigma_s}\rangle$. Hence the first integral in \eqref{eq:continuation-integral} converges to a vector orthogonal to $|\Omega_{\sigma_s}\rangle$, and
\[
 \mathsf{L}_{\sigma_s}|T_s\rangle
 =\left[e^{t\mathsf{L}_{\sigma_s}}\partial_s|\Omega_{\sigma_s}\rangle\right]_{t=0}^{t=\infty}
 =-\partial_s|\Omega_{\sigma_s}\rangle,
\]
so it is the solution of the centered Poisson equation in \eqref{eq:continuation-poisson}. The semigroup $e^{t\mathsf{L}_{\sigma_s}}$ preserves Hermiticity and $\partial_s\Omega_{\sigma_s}$ is Hermitian, so $T_s=T_s^*$. By \eqref{eq:parent-action}, $e^{t\mathsf{L}_{\sigma_s}}=\sGamma_{\sigma_s}^{1/2}e^{t\mathsf{L}_s}\sGamma_{\sigma_s}^{-1/2}$, which gives the second line of \eqref{eq:continuation-integral}.

Since $H$ commutes with $\sigma_s$, the formula for $\partial_s\Omega_{\sigma_s}$ above gives $\Gamma_{\sigma_s}^{-1/2}(\partial_s\Omega_{\sigma_s})=-\tfrac12(H-\Tr(\sigma_sH)I)$. In an eigenbasis $H=\operatorname{diag}(E_i)$, write $\Omega_{\sigma_s}=\operatorname{diag}(\omega_i)$ with $\omega_i>0$, so $\omega_i/\omega_j=e^{-s(E_i-E_j)/2}$. Then $\Gamma_{\sigma_s}^{1/2}$ multiplies the $(i,j)$ entry by $\sqrt{\omega_i\omega_j}$, and the Sylvester equation in \eqref{eq:continuation-poisson} reads $\tfrac12(\omega_i+\omega_j)(f_s)_{ij}=(T_s)_{ij}$. By the second line of \eqref{eq:continuation-integral},
\[
 \begin{aligned}
 (f_s)_{ij}
 &=\frac{2(T_s)_{ij}}{\omega_i+\omega_j}
 =-\frac{\sqrt{\omega_i\omega_j}}{\omega_i+\omega_j}
 \left[\int_0^\infty e^{t\mathcal L_s}
       \bigl(H-\Tr(\sigma_sH)I\bigr)\,dt\right]_{ij}\\
 &=-\tfrac12\sech\!\left(\frac{s(E_i-E_j)}4\right)
 \left[\int_0^\infty e^{t\mathcal L_s}
       \bigl(H-\Tr(\sigma_sH)I\bigr)\,dt\right]_{ij}.
 \end{aligned}
\]
Since $\ad_H$ multiplies the $(i,j)$ entry by $E_i-E_j$, this is \eqref{eq:continuation-sylvester}. Since $T_s$ is Hermitian and the $\omega_i$ are real, the first expression gives $\overline{(f_s)_{ji}}=(f_s)_{ij}$, that is, $f_s=f_s^*$.

It remains to show that $s\mapsto f_s$ is continuous. At fixed volume, $s\mapsto\mathsf L_{\sigma_s}$ is continuous on $[0,\beta]$, and for each $s$ the gap $\lambda_s$ of $-\mathsf L_{\sigma_s}$ on the orthogonal complement of $|\Omega_{\sigma_s}\rangle$ is positive by primitivity. Eigenvalues of a continuous family of self-adjoint operators depend continuously on the parameter, so $\lambda_*:=\min_{0\leq s\leq\beta}\lambda_s>0$. Hence $\|e^{t\mathsf L_{\sigma_s}}\partial_s|\Omega_{\sigma_s}\rangle\|_2\leq e^{-\lambda_*t}\|\partial_s|\Omega_{\sigma_s}\rangle\|_2$ uniformly in $s$, and dominated convergence shows that $s\mapsto T_s$ is continuous. Since each $\omega_i$ is continuous in $s$ and $\omega_i+\omega_j>0$, the formula $(f_s)_{ij}=2(T_s)_{ij}/(\omega_i+\omega_j)$ shows that $f_s$, and hence $\mathsf G_s$, is continuous. Thus \eqref{eq:continuation-unitary} has a unique $C^1$ solution. Since $\mathsf{G}_s^*=-\mathsf{G}_s$ and $\mathsf{G}_s|\Omega_{\sigma_s}\rangle=\partial_s|\Omega_{\sigma_s}\rangle$,
\[
 \frac{d}{ds}\bigl(\mathsf{U}_s^*\mathsf{U}_s\bigr)
 =\mathsf{U}_s^*(\mathsf{G}_s^*+\mathsf{G}_s)\mathsf{U}_s=0,
 \qquad
 \frac{d}{ds}\bigl(\mathsf{U}_s^*|\Omega_{\sigma_s}\rangle\bigr)
 =\mathsf{U}_s^*\bigl(\partial_s|\Omega_{\sigma_s}\rangle-\mathsf{G}_s|\Omega_{\sigma_s}\rangle\bigr)=0.
\]
Hence $\mathsf{U}_s^*\mathsf{U}_s=\one$, so $\mathsf U_s$ is unitary (the space is finite-dimensional), and $\mathsf{U}_s^*|\Omega_{\sigma_s}\rangle=|\Omega_{\sigma_0}\rangle$; applying $\mathsf U_s=(\mathsf U_s^*)^{-1}$ gives~\eqref{eq:continuation-derivative}. Finally, $\mathsf U_s$ preserves Hermiticity. Since $\mathsf G_s$ preserves Hermiticity, $\mathcal G_s(Y^*)=\mathcal G_s(Y)^*$ for every $Y$. Thus, for Hermitian $Y_0$, both $Y_s=\mathcal U_s(Y_0)$ and $Y_s^*$ solve $\partial_sY=\mathcal G_s(Y)$ with initial value $Y_0$, and uniqueness gives $Y_s=Y_s^*$.
\end{proof}

\subsection{Locality of the unitary flow}\label{sec:flow-locality}
The transport constructed above has a quasi-local generator under local rapid mixing. This allows us to control its action through estimates near each site.

To quantify locality of the generator, for any matrix $A$ on $\mathcal H$~\eqref{eqn:hilbert-space} and any site $x\in\Lambda$, we define the quasi-local norm centered at $x$ and the local strength norm, respectively:
\begin{equation}\label{eq:polynorm}
 \|A\|_{p,x}=\inf_{\substack{A=\sum_{r\geq0}A_r\\\supp A_r\subseteq\Ball_x(r)}}\sum_{r\geq0}(1+r)^p\|A_r\|,\qquad p\geq0,
\end{equation}
\begin{equation}\label{eq:local-strength-norm}
 \|A\|_{p,\mathrm{loc}}=\inf_{A=\sum_xA_x}\sup_x\|A_x\|_{p,x}.
\end{equation}
The quasi-local norm measures spatial decay around $x$, while the local strength norm measures strength per site. At $p=0$, the quasi-local norm is $\|A\|_{0,x}=\|A\|$.
{We use these two norms for operators $A\in\mathcal B(\mathcal H)$, superoperators $\mathcal A:\mathcal B(\mathcal H)\to\mathcal B(\mathcal H)$, and their vectorized representatives $\mathsf A\in\mathcal B(\mathcal H\otimes\mathcal H)$. The underlying norms are $\|A\|$ and $\|\mathcal A\|_{2\to2}=\|\mathsf A\|$, respectively, as in~\eqref{eqn:induced-norms_def}, with support defined by~\eqref{eqn:support}. Consequently, vectorization preserves both norms.}

The Poisson formula \eqref{eq:continuation-integral} integrates the dynamics over all times. To control its spatial tails, we require decay of each centered interaction term:
\begin{condition}[Local rapid mixing]\label{cond:transport-mixing}
There are constants $C_{\mathrm{mix}}>0$ and $\nu>0$ such that, for every interaction term $h_Z$, every $0\leq s\leq\beta$, and every $t\geq0$,
\begin{equation}\label{eq:transport-mixing}
 \left\|e^{t\mathcal L_s}
       \bigl(h_Z-\Tr(\sigma_sh_Z)I\bigr)\right\|_\infty
 \leq C_{\mathrm{mix}}(1+t)^{-\nu}\|h_Z\|_\infty.
\end{equation}
\end{condition}

We note that~\eqref{eq:transport-mixing} is different from the usual global mixing because it only considers the evolution of the local interaction terms $h_Z$ instead of the whole matrix algebra. Condition~\ref{cond:transport-mixing} naturally follows from regional rapid mixing by Proposition~\ref{prop:regional-mixing}. Section~\ref{sec:local-rapid-mixing} verifies its hypotheses in the one-dimensional and high-temperature settings.

The locality of the sampler controls propagation during short times, while Condition~\ref{cond:transport-mixing} controls the long-time tail of the Poisson integral. Together they give the following bound.
\begin{proposition}[Locality of the transport generator]\label{prop:flow-locality}
For the Gaussian KMS Gibbs sampler \eqref{eq:CKG} and the finite-range interactions in \eqref{eq:growth}, fix $p>0$. If Condition~\ref{cond:transport-mixing} holds with $\nu>p+6d+10$, then there is a constant $C_G$ such that
\begin{equation}\label{eq:flow-locality-target}
 \sup_{0\leq s\leq\beta}\|\mathsf{G}_s\|_{p,\mathrm{loc}}\leq C_G.
\end{equation}
Here $\mathsf{G}_s$ is the vectorized version of the transport generator defined in \eqref{eq:continuation-generator}.
\end{proposition}
The proof is given in Appendix~\ref{app:transport-coefficients}. {With the local strength norm bound~\eqref{eq:flow-locality-target} holding for sufficiently large $p$,} the Lieb--Robinson bound~\cite[Theorem~3.1]{NSY2019} implies that $\mathsf{U}_s^*\mathsf{O}\mathsf{U}_s$ can be approximated near the support of $\mathsf{O}$, with an error that decays with the truncation radius.


\section{R\'enyi-2 divergence decay from energy and entropy variation bounds}\label{sec:local-control}
Proposition~\ref{prop:transport} reduces the transport of the $L^2$ LSI to bounds on the derivatives of energy and entropy. We now prove these bounds, \eqref{eq:variation-targets}, along $|Y_s\rangle=\mathsf{U}_s|Y_0\rangle$ for the unitary path $\mathsf{U}_s$ constructed in Proposition~\ref{prop:exact-continuation}, with $E_s$ the Dirichlet energy defined before \eqref{eq:variation-targets}. The energy estimates below hold for every initial vector $|Y_0\rangle$, with $E_s=\langle Y_s|(-\mathsf{L}_{\sigma_s})|Y_s\rangle$ given by the same formula; the entropy estimates are for Hermitian $Y_0$ with $\|Y_0\|_2=1$.

\subsection{Derivative bounds and reduction to quadratic-form comparison}\label{sec:derivative-comparison}

In this section, we calculate the energy and entropy derivatives along the unitary flow. We will show that both derivatives can be bounded in terms of quadratic forms built from quasi-local operators. We then reduce the energy and entropy variation bounds in \eqref{eq:variation-targets} to comparisons between these quadratic forms and the Dirichlet form.

\paragraph*{Energy variation}
The sum-of-squares representation \eqref{eq:native-dirichlet-form} gives
\[
 E_s=\sum_{x,u}\int_{\mathbb R^2}
       \|\widehat{\mathsf{B}}_{s,x,u}(v,z)|Y_s\rangle\|_2^2\,d\mu(v,z).
\]
Both $\widehat{\mathsf B}_{s,x,u}(v,z)$ and $|Y_s\rangle$ vary with $s$. By unitarity, we can account for both changes by differentiating $\mathsf U_s^*\widehat{\mathsf B}_{s,x,u}(v,z)\mathsf U_s$ while keeping $|Y_0\rangle$ fixed. Using $\partial_s\mathsf U_s=\mathsf G_s\mathsf U_s$ and $\mathsf G_s^*=-\mathsf G_s$ gives
\begin{equation}\label{eq:moving-frame-derivative}
 \partial_s\bigl(\mathsf U_s^*\widehat{\mathsf B}_{s,x,u}(v,z)\mathsf U_s\bigr)
 =\mathsf U_s^*\bigl(\partial_s\widehat{\mathsf B}_{s,x,u}(v,z)+[\widehat{\mathsf B}_{s,x,u}(v,z),\mathsf G_s]\bigr)\mathsf U_s
 \mathrel{=:}\mathsf U_s^*\mathsf F_{s,x,u}(v,z)\mathsf U_s.
\end{equation}
To differentiate the energy, we now differentiate each squared norm in this fixed frame:
\begin{align*}
 \frac{d}{ds}\|\widehat{\mathsf B}_{s,x,u}(v,z)|Y_s\rangle\|_2^2
 =\frac{d}{ds}\|\mathsf U_s^*\widehat{\mathsf B}_{s,x,u}(v,z)\mathsf U_s|Y_0\rangle\|_2^2=2\operatorname{Re}\langle Y_s|\widehat{\mathsf B}_{s,x,u}(v,z)^*\mathsf F_{s,x,u}(v,z)|Y_s\rangle.
\end{align*}
Summing over $x,u$ and integrating against $d\mu(v,z)$ gives
\begin{equation}\label{eq:energy-derivative-calculation}
 \frac{d}{ds}E_s=2\operatorname{Re}\sum_{x,u}\int_{\mathbb R^2}
 \langle Y_s|\widehat{\mathsf B}_{s,x,u}(v,z)^*\mathsf F_{s,x,u}(v,z)|Y_s\rangle\,d\mu(v,z).
\end{equation}
The Gaussian bounds in Lemmas~\ref{lem:native-gradient-locality} and~\ref{lem:native-gradient-derivatives} allow us to differentiate under the integral. Cauchy--Schwarz then gives the energy derivative bound in Lemma~\ref{lem:derivative-reduction}.

\paragraph*{Entropy variation}
For Hermitian $Y_0$ with $\|Y_0\|_2=1$, both $Y_s$ and $\sigma_s$ vary along the flow, while $\Tr Y_s^2=1$ remains constant. Differentiating the relative entropy $D(Y_s^2\Vert\sigma_s)=\Tr[Y_s^2(\log Y_s^2-\log\sigma_s)]$ gives
\begin{equation}\label{eq:entropy-simultaneous-derivative}
 \frac{d}{ds}D(Y_s^2\Vert\sigma_s)=\underbrace{2\Tr[(\partial_sY_s)Y_s\log Y_s^2]-\Tr[(\partial_sY_s)\{Y_s,\log\sigma_s\}]}_{\text{change in }Y_s}-\underbrace{\Tr[Y_s^2\partial_s\log\sigma_s]}_{\text{change in }\sigma_s}.
\end{equation}
To reformulate the right-hand side of \eqref{eq:entropy-simultaneous-derivative}, we express it in terms of weighted commutators. Define
\[
 \mathcal B_{s,V}(Y)=\Gamma_{\sigma_s}^{1/2}\bigl([V,\Gamma_{\sigma_s}^{-1/2}(Y)]\bigr),
 \qquad \mathsf B_{s,V}|Y\rangle=|\mathcal B_{s,V}(Y)\rangle,
\]
with $\mathsf B_{s,V}=\sGamma_{\sigma_s}^{1/2}\Ad_V\sGamma_{\sigma_s}^{-1/2}$. In these terms, the derivative has the following representation:
\begin{equation}\label{eq:entropy-integral-representation}
 \frac d{ds}D(Y_s^2\Vert\sigma_s)
 =\sum_\gamma\int_{\mathbb R}
 \Tr\!\left[\mathcal B_{s,Q_\gamma(\omega)}(Y_s)
       \mathcal K_{h_\gamma,Y_s}\bigl(\mathcal B_{s,P_\gamma(\omega)}(Y_s)\bigr)\right]d\omega,
 \qquad \|h_\gamma\|_\infty\leq1.
\end{equation}
Each pair $P_\gamma(\omega),Q_\gamma(\omega)$ has an assigned center $x_\gamma\in\Lambda$. The corresponding weighted commutators annihilate the Gibbs vector and will be shown to be quasi-local around this center. For $Y=\sum_\lambda\lambda\Pi_\lambda$, the map $\mathcal K_{h,Y}$ is defined by
\begin{equation}\label{eq:entropy-kernel-contraction}
 \mathcal K_{h,Y}(A)=\sum_{\lambda,\mu}h(\lambda,\mu)\Pi_\lambda A\Pi_\mu,
 \qquad \|\mathcal K_{h,Y}\|_{2\to2}\leq\|h\|_\infty.
\end{equation}
Here, $A\mapsto\Pi_\lambda A\Pi_\mu$ are mutually orthogonal projections on Hilbert--Schmidt space summing to the identity, so for every matrix $A$,
\[
 \|\mathcal K_{h,Y}(A)\|_2^2
 =\sum_{\lambda,\mu}|h(\lambda,\mu)|^2\|\Pi_\lambda A\Pi_\mu\|_2^2
 \leq\|h\|_\infty^2\sum_{\lambda,\mu}\|\Pi_\lambda A\Pi_\mu\|_2^2=\|h\|_\infty^2\|A\|_2^2.
\]
In particular, $\|h\|_\infty\leq1$ makes $\mathcal K_{h,Y}$ a contraction. The construction and proof of \eqref{eq:entropy-integral-representation} are given in Appendix~\ref{sec:entropy-reduction}.

These calculations give the following derivative bounds.

\begin{lemma}[Derivative bounds]\label{lem:derivative-reduction}
Along the unitary path of Proposition~\ref{prop:exact-continuation}, for every initial vector $|Y_0\rangle$,
\begin{equation}\label{eq:energy-derivative-bound}
 \left|\frac{d}{ds}E_s\right|
 \leq2\sqrt{E_s}
 \left(\sum_{x,u}\int_{\mathbb R^2}\|\mathsf{F}_{s,x,u}(v,z)|Y_s\rangle\|_2^2\,d\mu(v,z)\right)^{1/2}.
\end{equation}
For Hermitian $Y_0$ with $\|Y_0\|_2=1$,
\begin{equation}\label{eq:entropy-derivative-bound}
 \left|\frac{d}{ds}D(Y_s^2\Vert\sigma_s)\right|
 \leq\frac12\sum_\gamma\int_{\mathbb R}
 \left(\|\mathsf{B}_{s,P_\gamma(\omega)}|Y_s\rangle\|_2^2+\|\mathsf{B}_{s,Q_\gamma(\omega)}|Y_s\rangle\|_2^2\right)d\omega.
\end{equation}
\end{lemma}
The energy bound \eqref{eq:energy-derivative-bound} follows from \eqref{eq:energy-derivative-calculation} as shown above; the entropy bound \eqref{eq:entropy-derivative-bound} is proved in Appendix~\ref{sec:entropy-reduction}, Step~4.

By Lemma~\ref{lem:derivative-reduction}, to derive~\eqref{eq:variation-targets}, it remains to prove the following two quadratic-form comparisons, with nonnegative constants $C_E,C_S$:
\begin{equation}\label{eq:energy-comparison-target}
 \sum_{x,u}\int_{\mathbb R^2}\mathsf{F}_{s,x,u}(v,z)^*\mathsf{F}_{s,x,u}(v,z)\,d\mu(v,z)
 \leq\frac{C_E^2}{4}(-\mathsf{L}_{\sigma_s}),
\end{equation}
\begin{equation}\label{eq:entropy-comparison-target}
 \sum_\gamma\int_{\mathbb R}\left(\mathsf{B}_{s,P_\gamma(\omega)}^*\mathsf{B}_{s,P_\gamma(\omega)}+\mathsf{B}_{s,Q_\gamma(\omega)}^*\mathsf{B}_{s,Q_\gamma(\omega)}\right)d\omega
 \leq2C_S(-\mathsf{L}_{\sigma_s}).
\end{equation}
Both comparisons require the operators on the left to annihilate the zero-energy Gibbs vector. This cancellation holds for the two derivative families:
\begin{equation}\label{eq:derivative-gibbs-cancellation}
 \mathsf{F}_{s,x,u}(v,z)|\Omega_{\sigma_s}\rangle=0,\qquad
 \mathsf{B}_{s,P_\gamma(\omega)}|\Omega_{\sigma_s}\rangle=\mathsf{B}_{s,Q_\gamma(\omega)}|\Omega_{\sigma_s}\rangle=0.
\end{equation}
For the first identity, insert $\partial_s|\Omega_{\sigma_s}\rangle=\mathsf G_s|\Omega_{\sigma_s}\rangle$ into \eqref{eq:moving-frame-derivative} and use $\widehat{\mathsf B}_{s,x,u}(v,z)|\Omega_{\sigma_s}\rangle=0$ from Lemma~\ref{lem:native-dirichlet-decomposition} for every $s$:
\[
 \mathsf F_{s,x,u}(v,z)|\Omega_{\sigma_s}\rangle
 =\partial_s\bigl(\widehat{\mathsf B}_{s,x,u}(v,z)|\Omega_{\sigma_s}\rangle\bigr)
   -\mathsf G_s\widehat{\mathsf B}_{s,x,u}(v,z)|\Omega_{\sigma_s}\rangle=0.
\]
The other two follow from $\mathsf B_{s,V}|\Omega_{\sigma_s}\rangle=\sGamma_{\sigma_s}^{1/2}|[V,I]\rangle=0$. To obtain \eqref{eq:energy-comparison-target} and \eqref{eq:entropy-comparison-target}, we still need to control the size and spatial spread of these actions relative to the Dirichlet form. Section~\ref{sec:conservation-estimates} supplies sufficient conditions through quasi-locality and a common local action-energy inequality.

\subsection{Sufficient conditions for the comparison}\label{sec:conservation-estimates}\label{sec:response}
We treat the energy and entropy derivatives together. We first show that the operators entering both derivatives are quasi-local, then reduce the bounds \eqref{eq:energy-comparison-target} and \eqref{eq:entropy-comparison-target} to a common local comparison with the Dirichlet form.

\paragraph*{Quasi-locality of the derivative terms}
Each operator entering the derivatives can be approximated by operators supported in balls around its center, with errors that decrease as the balls grow.

Proposition~\ref{prop:flow-locality} controls the generator of the transport. Together with the Gaussian filter estimates, it gives the following bounds for the derivative terms in \eqref{eq:moving-frame-derivative} and \eqref{eq:entropy-integral-representation}, for a fixed exponent $p_0>0$:
\begin{lemma}[Quasi-locality of the derivative terms]\label{lem:comparison-families}
Let $H$ be a finite-range Hamiltonian as in \eqref{eq:growth}, and consider its Gaussian KMS Gibbs sampler \eqref{eq:CKG}. Suppose Condition~\ref{cond:transport-mixing} holds with $\nu>p_0+9d+12$. Then
\begin{align}
 \sup_{s,x}\sum_u\int_{\mathbb R^2}\|\mathsf{F}_{s,x,u}(v,z)\|_{p_0,x}^2\,d\mu(v,z)&\leq M_E^2,\label{eq:energy-family-locality}\\
 \sup_{s,x}\sum_{\gamma:x_\gamma=x}\int_{\mathbb R}
 \left(\|\mathsf{B}_{s,P_\gamma(\omega)}\|_{p_0,x}^2+\|\mathsf{B}_{s,Q_\gamma(\omega)}\|_{p_0,x}^2\right)d\omega&\leq M_S^2.\label{eq:entropy-family-locality}
\end{align}
Here $M_E$ and $M_S$ can be chosen depending only on local parameters such as $q,R_0,J,d,p_0$, temperature $\beta$, and the mixing parameters
$C_{\mathrm{mix}},\nu$ in Condition~\ref{cond:transport-mixing},
independently of $s$, $x$, and the system size $|\Lambda|$.
\end{lemma}
The energy and entropy estimates in Lemma~\ref{lem:comparison-families} are proved in Appendix~\ref{sec:energy-reduction} and Appendix~\ref{sec:entropy-reduction}, Step~5, respectively.

\paragraph*{A general comparison condition}
Lemma~\ref{lem:comparison-families} allows us to formulate a sufficient condition without referring to the particular operators produced by differentiation. It is enough to prove the same comparison for every quasi-local family that annihilates the Gibbs vector, with a bound determined by its local strength.~This is stated in the following condition.
\begin{condition}[local action-energy inequality]\label{cond:local-action}
For a family $\mathsf{R}_{s,\alpha}$ centered at $x_\alpha$, define its local strength $M$ by
\[
 M^2=\sup_{s,x}\sum_{\alpha:x_\alpha=x}\|\mathsf{R}_{s,\alpha}\|_{p_0,x}^2.
\]
There is a constant $C_{\mathrm{loc}}$ such that every such family with $M<\infty$ satisfying $\mathsf{R}_{s,\alpha}|\Omega_{\sigma_s}\rangle=0$ obeys, for any vector $|Y\rangle$,
\begin{equation}\label{eq:comparison-target}
 \sum_\alpha\|\mathsf{R}_{s,\alpha}|Y\rangle\|_2^2\leq C_{\mathrm{loc}}M^2\langle Y|(-\mathsf{L}_{\sigma_s})|Y\rangle.
\end{equation}
\end{condition}
Applying Condition~\ref{cond:local-action} to the two families in Lemma~\ref{lem:comparison-families} gives both comparisons required in~\eqref{eq:energy-comparison-target} and~\eqref{eq:entropy-comparison-target} of Section~\ref{sec:derivative-comparison}.

\paragraph*{Regional sufficient conditions}
Condition~\ref{cond:local-action} is reminiscent of a comparison in the stability theory of frustration-free Hamiltonians~\cite[Proposition~2]{MZ}. For a finite-range frustration-free Hamiltonian $\mathsf{A}$, let $\mathsf{A}_I$ be the sum of its interaction terms supported inside $I$, and let $\mathsf{P}_I$ be the projection onto the ground state space of $\mathsf{A}_I$. Consider a perturbation $\mathsf{W}=\sum_{x,r}\mathsf{W}_x(r)$ whose Hermitian terms are supported on $\Ball_x(r)$ and satisfy
\[
 \|\mathsf{W}_x(r)\|\leq\varepsilon f(r).
\]
Here $\varepsilon\geq0$ bounds the local perturbation strength, and $f$ is a fixed nonnegative decay profile, independent of $\varepsilon$ and the system size. Under the local spectral gap condition $\mathsf{A}_I\geq\gamma_I(\one-\mathsf{P}_I)$, with $f$ decaying sufficiently fast relative to the inverse local spectral gaps and the lattice volume growth,
\begin{equation}\label{eq:ff-relative-bound}
 \mathsf{W}_x(r)\mathsf{P}_{\Ball_x(r)}=0\quad \Rightarrow\quad |\langle Y|\mathsf{W}|Y\rangle|\leq C\varepsilon\langle Y|\mathsf{A}|Y\rangle,
\end{equation}
where $C$ depends on the decay profile, the local spectral gap bounds, and the lattice geometry, but is independent of $\varepsilon$ and the system size.

In Condition~\ref{cond:local-action}, we regard $-\mathsf{L}_{\sigma_s}$ as a frustration-free Hamiltonian with ground state $|\Omega_{\sigma_s}\rangle$ and seek the following bound for every vector $|Y\rangle$ under a weaker annihilation condition:
\[
 \mathsf{R}_{s,\alpha}|\Omega_{\sigma_s}\rangle=0\ \ (\forall\alpha)
 \quad\Longrightarrow\quad
 \sum_\alpha\|\mathsf{R}_{s,\alpha}|Y\rangle\|_2^2
 \leq C_{\mathrm{loc}}M^2\langle Y|(-\mathsf{L}_{\sigma_s})|Y\rangle.
\]
This only requires each local term to annihilate the global Gibbs vector $|\Omega_{\sigma_s}\rangle$, whereas \eqref{eq:ff-relative-bound} requires each perturbation term to annihilate the entire local ground state space. We therefore impose an additional condition to control the action of local operators on local ground state space. Far from the boundary, an operator that annihilates the Gibbs vector should have little effect on those states. This can be formulated as local state domination, stated in Condition~\ref{ass:local-gap-state-domination}\textup{(ii)} below.

\begin{condition}[Local spectral gaps and local state domination]\label{ass:local-gap-state-domination}
There is a constant $c_{\mathrm{cmp}}>0$ such that, for each $s\in[0,\beta]$, there is a positive form
\begin{equation}\label{eq:local-gap-auxiliary}
 \mathsf{A}_s=\sum_{z\in\Lambda}\mathsf{a}_{s,z},\qquad \mathsf{a}_{s,z}\geq0,\qquad
 \ker \mathsf{A}_s=\operatorname{span}\{|\Omega_{\sigma_s}\rangle\},\qquad
 \mathsf{A}_s\leq c_{\mathrm{cmp}}(-\mathsf{L}_{\sigma_s}).
\end{equation}
Each term annihilates the Gibbs vector, so $\mathsf{A}_s$ is frustration free. For $I\subseteq\Lambda$, let $\mathsf{A}_{s,I}=\sum_{z\in I}\mathsf{a}_{s,z}$ and let $\mathsf{P}_{s,I}$ be its Hilbert--Schmidt kernel projection.
There are constants $a,b>0$ such that the following hold.
\begin{enumerate}
\item \emph{Local spectral gaps.} For all $s$, $x$, $\Lambda$, and integer $r\geq0$,
\begin{equation}\label{eq:local-gap-assumption}
 \mathsf{A}_{s,\Ball_x(r)}\geq\gamma_*(1+r)^{-a}(\one-\mathsf{P}_{s,\Ball_x(r)}),
\end{equation}
with $\gamma_*>0$ being a constant.
\item \emph{Local state domination.} For all $s$, $x$, $\Lambda$, and integer $r\geq0$, set $I=\Ball_x(r)$ and $J=\Ball_x(\lfloor r/2\rfloor)$. Every unit vector $|\psi_s\rangle\in\ker \mathsf{A}_{s,I}$ and every operator supported on the sites in $J$ with $0\leq \mathsf{O}_J\leq\one$ satisfy
\begin{equation}\label{eq:local-gap-kernel-assumption}
 \langle\psi_s|\mathsf{O}_J|\psi_s\rangle
 \leq C_{\mathrm{dom}}\langle\Omega_{\sigma_s}|\mathsf{O}_J|\Omega_{\sigma_s}\rangle+\delta_*(1+r)^{-2b},
\end{equation}
with $C_{\mathrm{dom}}\geq1$ and $\delta_*\geq0$ being constants.
\end{enumerate}
\end{condition}

Condition~\ref{ass:local-gap-state-domination} requires a local spectral gap bound for $\mathsf{A}_s$. If we directly take $\mathsf{A}_s=-\mathsf{L}_{\sigma_s}$, however, the sampler's quasi-local tails couple each region to its exterior, making this gap difficult to control. We therefore choose a comparison form with an explicit regional kernel. For one-dimensional systems, we use the Pauli Dirichlet form of~\cite[Lemma~III.4]{BC2025}
\begin{equation}\label{eq:commutator-form}
 \mathcal D_s(X)=\sum_{x,u}\|[S_{x,u},X]\|_{2,\sigma_s}^2,
\end{equation}
where $S_{x,u}$ are the on-site Pauli operators. The weighted Pauli operator representing this form on Hilbert--Schmidt space is
\begin{equation}\label{eq:pauli-form}
 \mathsf{A}_s=\sum_{x,u}\mathsf{B}_{s,x,u}^*\mathsf{B}_{s,x,u},\qquad
 \mathsf{B}_{s,x,u}=\mathsf{B}_{s,S_{x,u}}=\sGamma_{\sigma_s}^{1/2}\Ad_{S_{x,u}}\sGamma_{\sigma_s}^{-1/2},
\end{equation}
so that $\mathcal D_s(X)=\langle X|\sGamma_{\sigma_s}^{1/2}\mathsf{A}_s\sGamma_{\sigma_s}^{1/2}|X\rangle$.
The commutator maps $\ad_{S_{x,u}}$ are strictly local: the zero-energy condition on a region requires $X$ to act as the identity there, while allowing arbitrary action outside; see Appendix~\ref{app:one-dimensional}, Step~2. This explicit regional kernel makes local spectral gap estimates accessible. For one-dimensional systems, the comparison \eqref{eq:local-gap-auxiliary} holds at every finite temperature~\cite[Theorem~VI.1]{BC2025}. At high temperature, we instead choose $\mathsf{A}_s=\mathsf{U}_s(-\mathsf{L}_{\sigma_0})\mathsf{U}_s^*$, the transported product form in \eqref{eq:ht-transported-form}. Unitary conjugation transfers the local spectral gaps of $-\mathsf{L}_{\sigma_0}$ directly to $\mathsf{A}_s$. The proof of the comparison $\mathsf{A}_s\leq4(-\mathsf{L}_{\sigma_s})$ uses the high-temperature assumption (Lemma~\ref{lem:ht-comparison}).

\begin{proposition}[local action-energy inequality from local spectral gaps]\label{prop:local-gap-criterion}
Suppose Condition~\ref{ass:local-gap-state-domination} holds and $\min\{p_0,b\}>d+a$. Then Condition~\ref{cond:local-action} holds at order $p_0$.
\end{proposition}
The proof is given in Appendix~\ref{app:local-action-criterion}.

\begin{figure}[htbp]
\centering
\begin{tikzpicture}[
 font=\footnotesize,
 box/.style={draw,rounded corners=3pt,align=center,inner sep=4pt,text width=5.45cm},
 conditionbox/.style={box,draw=blue!75!black,fill=blue!5,minimum height=1.35cm},
 edgelabel/.style={font=\footnotesize,fill=white,inner sep=2pt,align=center},
 arr/.style={-{Latex[length=2mm]},thick}]
\node[conditionbox] (regional) at (0,0) {
 \textbf{Local spectral gaps and local state domination}\\[4pt]
 Condition~\ref{ass:local-gap-state-domination}};
\node[conditionbox] (mixing) at (8,0) {
 \textbf{Local rapid mixing}\\[4pt]
 Condition~\ref{cond:transport-mixing}};
\node[conditionbox] (action) at (0,-2.6) {
 \textbf{local action-energy inequality}\\[4pt]
 Condition~\ref{cond:local-action}};
\node[box,fill=green!5,minimum height=1.35cm] (locality) at (8,-2.6) {
 \textbf{Quasi-local derivative terms}\\[4pt]
 Local strengths $M_E$ and $M_S$};
\coordinate (inputs) at (4,-3.8);
\coordinate (consequences) at (4,-5.0);
\node[box,fill=blue!7,minimum height=2.1cm] (energy) at (0,-6.6) {
 \textbf{Energy variation}\\[5pt]
 $\left|\frac{d}{ds}E_s\right|\leq C_EE_s$\\[5pt]
 $C_E=2\sqrt{C_{\mathrm{loc}}}\,M_E$\\[4pt]
 Any initial vector $|Y_0\rangle$};
\node[box,fill=orange!9,minimum height=2.1cm] (entropy) at (8,-6.6) {
 \textbf{Entropy variation}\\[5pt]
 $\left|\frac{d}{ds}D(Y_s^2\Vert\sigma_s)\right|\leq C_SE_s$\\[5pt]
 $C_S=\tfrac12C_{\mathrm{loc}}M_S^2$\\[4pt]
 Hermitian $Y_0$, $\|Y_0\|_2=1$};
\draw[arr] (regional.south) -- node[edgelabel,midway]
 {Proposition~\ref{prop:local-gap-criterion}} (action.north);
\draw[arr] (mixing.south) -- node[edgelabel,midway]
 {Lemma~\ref{lem:comparison-families}} (locality.north);
\draw[thick] (action.south) -- (0,-3.8) -- (inputs);
\draw[thick] (locality.south) -- (8,-3.8) -- (inputs);
\draw[arr] (inputs) -- node[edgelabel,midway]
 {Lemma~\ref{lem:derivative-reduction}} (consequences);
\draw[thick] (0,-5.0) -- (8,-5.0);
\draw[arr] (0,-5.0) -- (energy.north);
\draw[arr] (8,-5.0) -- (entropy.north);
\end{tikzpicture}
\caption{Conditions for controlled variation of energy and entropy along the flow of Proposition~\ref{prop:exact-continuation}, under the hypotheses of Lemma~\ref{lem:comparison-families} and Proposition~\ref{prop:local-gap-criterion}, with the same exponent $p_0$ in the quasi-local norms.}
\label{fig:condition-implications}
\end{figure}
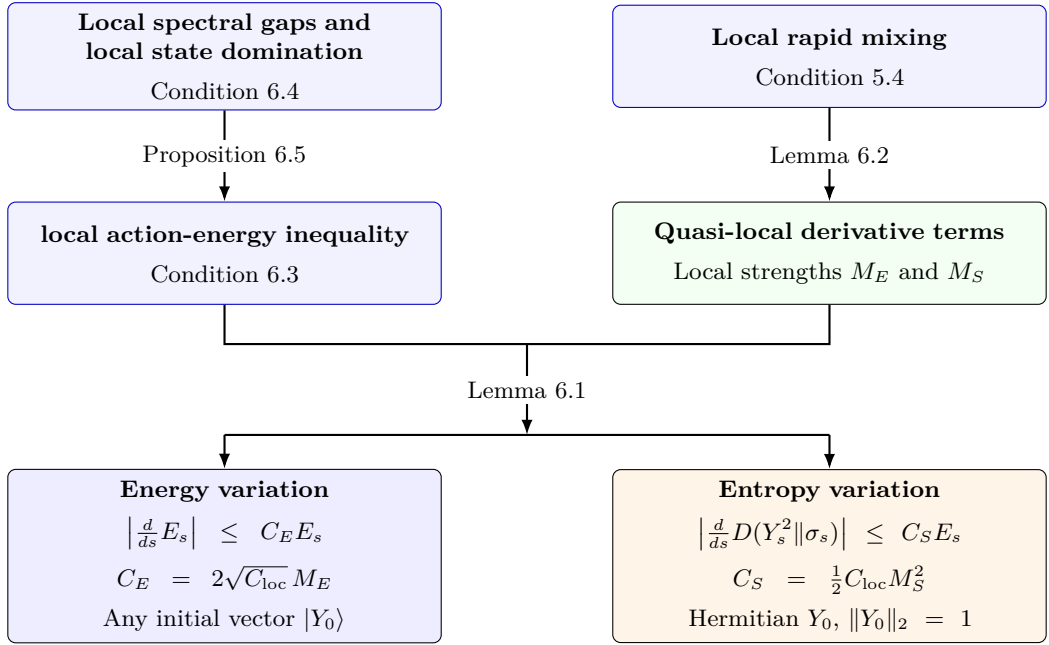
\FloatBarrier

We summarize the conditions required for the energy and entropy variation bounds \eqref{eq:variation-targets} in Figure~\ref{fig:condition-implications}.

\subsection{Criteria for Decay of R\'enyi-2 Divergence}\label{sec:renyi-criterion}
Since these bounds suffice to transport the $L^2$ LSI by Proposition~\ref{prop:transport}, we obtain the following theorem.

\begin{theorem}[A general criterion for size-independent R\'enyi-2 decay]\label{prop:general-criterion}\label{thm:local-criterion}
For the Gaussian KMS Gibbs sampler \eqref{eq:CKG} on the finite-range spin systems in \eqref{eq:growth}, with local dimension $2^q$, fix $p_0>0$. Assume:
\begin{itemize}
\item Local rapid mixing (Condition~\ref{cond:transport-mixing}) holds with $\nu>p_0+9d+12$. This condition reduces to regional rapid mixing by Proposition~\ref{prop:regional-mixing}.
\item The local action-energy inequality (Condition~\ref{cond:local-action}) holds at order $p_0$. Verifying this condition reduces to verifying Condition~\ref{ass:local-gap-state-domination} with $\min\{p_0,b\}>d+a$.
\end{itemize}
Then the $L^2$ LSI constant satisfies
\begin{equation}\label{eq:criterion-entropy}
 \alpha_2(\mathcal L_\beta)\geq\frac{2}{C_\beta},
 \qquad C_\beta=(C_0+\beta C_S)e^{\beta C_E},
\end{equation}
where $C_0$ is the constant in \eqref{eq:product-state-reference}, $C_E=2\sqrt{C_{\mathrm{loc}}}\,M_E$, and $C_S=\tfrac12C_{\mathrm{loc}}M_S^2$, with $M_E,M_S$ from Lemma~\ref{lem:comparison-families} and $C_{\mathrm{loc}}$ from Condition~\ref{cond:local-action}. Consequently, with $\kappa=2/C_\beta$, every state $\rho$ and every $t\geq0$ satisfy
\begin{equation}\label{eq:criterion-decay}
 \widetilde D_2(e^{t\mathcal L_{\beta,*}}\rho\Vert\sigma_\beta)
 \leq e^{-\kappa t}\widetilde D_2(\rho\Vert\sigma_\beta).
\end{equation}
\end{theorem}
\begin{proof}
Lemma~\ref{lem:comparison-families} supplies the local strengths $M_E,M_S$ of the derivative families, which annihilate the Gibbs vector by \eqref{eq:derivative-gibbs-cancellation}. Applying Condition~\ref{cond:local-action} and Lemma~\ref{lem:derivative-reduction} gives \eqref{eq:variation-targets} with the constants $C_E$ and $C_S$ above. Proposition~\ref{prop:transport} therefore gives \eqref{eq:criterion-entropy}, and Lemma~\ref{lem:l2-renyi} gives \eqref{eq:criterion-decay} with rate $\kappa=2/C_\beta$.
\end{proof}

We next verify these inputs in the two regimes before completing the MLSI argument in Section~\ref{sec:renyi-to-mlsi}.

\subsection{Verification in the two regimes}\label{sec:regime-verification}
We verify Conditions~\ref{cond:transport-mixing} and~\ref{ass:local-gap-state-domination} for one-dimensional systems at every finite temperature and for higher-dimensional systems at sufficiently high temperature.

\subsubsection{Local rapid mixing}\label{sec:local-rapid-mixing}
Local rapid mixing can be verified from regional rapid mixing. For a ball $\Ball_x(R)$, let $\mathcal L_s^{[R]}$ be the Gaussian KMS Gibbs sampler \eqref{eq:CKG} constructed from $H_{\Ball_x(R)}=\sum_{Z\subseteq\Ball_x(R)}h_Z$, with Gibbs state $\sigma_s^{[R]}$. The notation $[R]$ suppresses the center $x$. Its trace-norm mixing time is
\begin{equation}\label{eq:regional-mixing-time}
 t_{\rm mix}^{[R],s}(\varepsilon)
 =\inf\left\{t\geq0:\sup_\rho
       \|e^{t\mathcal L_{s,*}^{[R]}}\rho-\sigma_s^{[R]}\|_1
       \leq\varepsilon\right\},
\end{equation}
where the supremum is over states on $\Ball_x(R)$.

\begin{proposition}[Regional rapid mixing implies local rapid mixing]\label{prop:regional-mixing}
For the Gaussian KMS Gibbs sampler \eqref{eq:CKG} and the interactions in \eqref{eq:growth}, suppose there are constants $C_{\mathrm{reg}}>0$ and $0\leq a<1$ such that
\[
 t_{\rm mix}^{[R],s}(1/4)\leq C_{\mathrm{reg}}(1+R)^a
\]
for all balls $\Ball_x(R)$ and $s\in[0,\beta]$. Then local rapid mixing (Condition~\ref{cond:transport-mixing}) holds for every fixed $\nu>0$.
\end{proposition}
Regional rapid mixing controls the evolution inside a large ball, while quasi-locality bounds the effect of the exterior. See Appendix~\ref{app:regional-mixing} for details.

Global rapid mixing is known in both regimes. The cited results apply to each regional Hamiltonian in the same class and give the following bounds:
\begin{itemize}[leftmargin=*,itemsep=6pt]
\item For the one-dimensional Hamiltonians in \eqref{eq:chain-setting} at any fixed $0<\beta<\infty$, the rapid-mixing theorem~\cite[Theorem~2]{CD2026} gives
\begin{equation}\label{eq:chain-mixing-input}
 t_{\rm mix}^{[R],s}(1/4)\leq C_{\beta,q}\log^4(2+R)
 \qquad(0\leq s\leq\beta).
\end{equation}
\item At high temperature, for the finite-range interactions in \eqref{eq:growth}, choose $0<\theta_*\leq1$, depending only on $d,R_0,q$, so that the rapid-mixing theorem~\cite[Theorem~1]{RFA} applies. It gives
\begin{equation}\label{eq:ht-mixing-input}
 t_{\rm mix}^{[R],s}(1/4)\leq C\log(2+R)
 \qquad(0\leq s\leq\beta,\ \beta J\leq\theta_*).
\end{equation}
The constant is uniform over this temperature range.
\end{itemize}
Both bounds satisfy the hypothesis of Proposition~\ref{prop:regional-mixing}. Thus Condition~\ref{cond:transport-mixing} holds for every fixed $\nu>0$, with constants independent of system size and, in the high-temperature case, uniform over $\beta J\leq\theta_*$.

\subsubsection{The choice of comparison form and local spectral gap}\label{sec:comparison-choice}
To verify Condition~\ref{ass:local-gap-state-domination}, we choose comparison forms whose local spectral gaps can be controlled uniformly in system size.

\paragraph*{Proof idea for one-dimensional systems}\label{sec:local-gap-chain-assumptions}
We select the comparison form $\mathsf{A}_s$ to be the weighted Pauli operator in \eqref{eq:pauli-form}. The comparison theorem~\cite[Theorem~VI.1]{BC2025} gives $\mathsf{A}_s\leq c_{\beta,q}(-\mathsf{L}_{\sigma_s})$.
To establish the local spectral gap \eqref{eq:local-gap-assumption}, we cut the bonds across the boundary of $I=\Ball_x(r)$. Specifically, the Gibbs state after this cut is
\begin{equation}\label{eq:chain-cut-reference-overview}
 \tau=\tau_I\otimes\tau_{I^c},
\end{equation}
where $\tau_I\propto e^{-sH_I}$ and $\tau_{I^c}\propto e^{-sH_{I^c}^{\mathrm{cut}}}$, with $H_I$ and $H_{I^c}^{\mathrm{cut}}$ containing the interaction terms supported entirely inside and outside $I$, respectively. This product structure completely isolates $I$ from its exterior. We define the weighted Pauli operator for the cut state $\tau$ on $I$ by
\[
 \mathsf{A}_{s,I}^{\mathrm{cut}}=\sum_{x\in I,u}(\mathsf{B}_{s,x,u}^{\mathrm{cut}})^*\mathsf{B}_{s,x,u}^{\mathrm{cut}},\qquad
 \mathsf{B}_{s,x,u}^{\mathrm{cut}}=\sGamma_\tau^{1/2}\Ad_{S_{x,u}}\sGamma_\tau^{-1/2}.
\]
Let $\mathsf{P}_I^{\mathrm{cut}}$ be the projection onto its kernel. Since the Pauli operators act only on $I$ and $\tau$ factorizes across its boundary, the gap inequality reduces to the subsystem $I$, for which~\cite[Theorem~V.1]{BC2025} gives
\[
 \mathsf{A}_{s,I}^{\mathrm{cut}}\geq\lambda_{\mathrm{cut}}(\one-\mathsf{P}_I^{\mathrm{cut}}),
\]
with $\lambda_{\mathrm{cut}}>0$ independent of the interval length. Restoring the cut boundary bonds changes the weighting by $\mathsf{W}=\sGamma_{\sigma_s}^{1/2}\sGamma_\tau^{-1/2}$. Lemma~\ref{lem:chain-cut-locality} bounds both $\|\mathsf{W}\|$ and $\|\mathsf{W}^{-1}\|$ by a constant $K$ depending only on $\beta,q$. Thus restoring the boundary costs only a constant factor in the gap:
\[
 \mathsf{A}_{s,I}\geq\lambda_{\beta,q}(\one-\mathsf{P}_{s,I}),\qquad
 \lambda_{\beta,q}=K^{-4}\lambda_{\mathrm{cut}}>0.
\]
Here $\mathsf{A}_{s,I}=\sum_{x\in I,u}\mathsf{B}_{s,x,u}^*\mathsf{B}_{s,x,u}$ and $\mathsf{P}_{s,I}$ is its kernel projection. The comparison of the two gaps is proved in Lemma~\ref{lem:chain-local-poincare}.

\paragraph*{Proof idea for high-temperature systems}\label{sec:local-gap-ht-assumptions}
For the high-temperature setting in \eqref{eq:growth}, we select the comparison form by transporting the product form at $s=0$:
\begin{equation}\label{eq:ht-transported-form}
 \mathsf{A}_s=\mathsf{U}_s(-\mathsf{L}_{\sigma_0})\mathsf{U}_s^*.
\end{equation}
In Lemma~\ref{lem:ht-comparison}, we show that the comparison $\mathsf{A}_s\leq4(-\mathsf{L}_{\sigma_s})$ holds for $\beta J\leq\beta_0$, where $0<\beta_0\leq\theta_*$ is independent of system size. We now verify the local spectral gap \eqref{eq:local-gap-assumption}.
 At $s=0$, the form $\mathsf{A}_0$ admits the decomposition
\begin{equation}\label{eq:ht-product-projections}
 \mathsf{A}_0=-\mathsf{L}_{\sigma_0}=\nu_q\sum_{z\in\Lambda}\mathsf{Q}_z,\qquad
 \mathsf{A}_{0,I}:=\nu_q\sum_{z\in I}\mathsf{Q}_z,
\end{equation}
see Lemma~\ref{lem:product-entropy} for details. With $\mathsf{P}_I^0$ denoting the projection onto $\ker\mathsf{A}_{0,I}$, the local spectral gap satisfies $\mathsf{A}_{0,I}\geq\nu_q(\one-\mathsf{P}_I^0)$.
Unitary conjugation transports the decomposition of $\mathsf{A}_0$ to $\mathsf{A}_s$:
\[
 \mathsf{A}_s=\sum_{z\in\Lambda}\mathsf{a}_{s,z},\qquad
 \mathsf{a}_{s,z}=\nu_q\mathsf{U}_s\mathsf{Q}_z\mathsf{U}_s^*.
\]
The local spectral gap is preserved under unitary conjugation:
\begin{equation}\label{eq:ht-regional-kernel}
 \mathsf{A}_{s,I}=\sum_{z\in I}\mathsf{a}_{s,z}=\mathsf{U}_s\mathsf{A}_{0,I}\mathsf{U}_s^*\geq\nu_q(\one-\mathsf{P}_{s,I}),\qquad \mathsf{P}_{s,I}=\mathsf{U}_s\mathsf{P}_I^0\mathsf{U}_s^*.
\end{equation}
The detailed calculation and proof of the local spectral gap are given in Appendix~\ref{app:high-temperature}, Step~1. Thus Condition~\ref{ass:local-gap-state-domination}\textup{(i)} holds in both settings. It remains to verify local state domination for these two comparison forms.

\subsubsection{Local state domination}\label{sec:local-domination}
Recall that our goal is the local state domination bound \eqref{eq:local-gap-kernel-assumption}. For $I=\Ball_x(r)$ and $J=\Ball_x(\lfloor r/2\rfloor)$, we need to show that every unit vector $|\psi_s\rangle\in\ker\mathsf{A}_{s,I}$ and every observable $0\leq\mathsf{O}_J\leq\one$ supported on $J$ satisfy
\[
 \langle\psi_s|\mathsf{O}_J|\psi_s\rangle
 \leq C_{\mathrm{dom}}\langle\Omega_{\sigma_s}|\mathsf{O}_J|\Omega_{\sigma_s}\rangle
       +\delta_*(1+r)^{-2b}.
\]

\paragraph*{Proof idea for one-dimensional systems}
Here, we give a brief overview of the proof idea. The detailed derivation is in Appendix~\ref{app:domination-chain}, with the proof in Step~4. We use the same boundary cut as in \eqref{eq:chain-cut-reference-overview}. Once $I$ is disconnected from its exterior, the Pauli constraints inside $I$ fix its Gibbs factor, while leaving the exterior unrestricted:
\[
 \ker\mathsf{A}_{s,I}^{\mathrm{cut}}
 =\operatorname{span}\{|\Omega_{\tau_I}\rangle\}\otimes\mathcal H_{I^c}^{\mathrm{dbl}}.
\]
Consequently, for every unit vector $|\phi\rangle\in\ker\mathsf{A}_{s,I}^{\mathrm{cut}}$ and every observable $\mathsf{O}_I$ supported on $I$,
\begin{equation}\label{eq:cut-interior-expectations}
 \langle\phi|\mathsf{O}_I|\phi\rangle
 =\langle\Omega_\tau|\mathsf{O}_I|\Omega_\tau\rangle.
\end{equation}
For an operator $0\leq\mathsf{O}_J\leq\one$ supported on $J=\Ball_x(\lfloor r/2\rfloor)$ and a unit vector $|\psi_s\rangle\in\ker\mathsf{A}_{s,I}$, we use the weighting map $\mathsf{W}=\sGamma_{\sigma_s}^{1/2}\sGamma_\tau^{-1/2}$ to transfer \eqref{eq:cut-interior-expectations} to the desired local state domination bound:
\begin{equation}\label{eq:chain-domination-reduction-overview}
 \langle\psi_s|\mathsf{O}_J|\psi_s\rangle\leq2K^6\langle\Omega_{\sigma_s}|\mathsf{O}_J|\Omega_{\sigma_s}\rangle+8K^6\|[\sqrt{\mathsf{O}_J},\mathsf{W}]\|^2.
\end{equation}
It remains to show that $\|[\sqrt{\mathsf{O}_J},\mathsf{W}]\|$ decays with the distance from $J$ to the cut bonds.

In one dimension, $\mathsf{W}$ is quasi-local near those bonds. For $r\geq4$, we truncate it to $\widetilde{\mathsf W}_L$ supported on a boundary strip $N_L$, with $L=\lfloor r/4\rfloor$ and $\dist(J,N_L)\geq L$, such that
\begin{equation}\label{eq:chain-boundary-truncation}
 \|\mathsf{W}-\widetilde{\mathsf W}_L\|\leq Ce^{-cL\log(2+L)}.
\end{equation}
See Lemma~\ref{lem:chain-cut-locality} for details. The separation $\dist(J,N_L)\geq L>0$ gives $[\sqrt{\mathsf{O}_J},\widetilde{\mathsf W}_L]=0$, so the remaining term in \eqref{eq:chain-domination-reduction-overview} is controlled by
\[
 \|[\sqrt{\mathsf{O}_J},\mathsf{W}]\|=\|[\sqrt{\mathsf{O}_J},\mathsf{W}-\widetilde{\mathsf W}_L]\|\leq2\|\mathsf{W}-\widetilde{\mathsf W}_L\|\leq C^{\prime}e^{-c^{\prime}r\log(2+r)}.
\]
Substituting this bound into \eqref{eq:chain-domination-reduction-overview} gives \eqref{eq:local-gap-kernel-assumption}.

\paragraph*{Proof idea for high-temperature systems}
The proof of \eqref{eq:local-gap-kernel-assumption} for the high-temperature case is in Appendix~\ref{app:high-temperature}, Step~3. In the proof, we use the product state $\sigma_0=\bigotimes_{z\in\Lambda}(2^{-q}I_{2^q})$ as the reference. For $I=\Ball_x(r)$, we have
\[
 \ker\mathsf{A}_{0,I}=\operatorname{span}\{|\Omega_{\sigma_0,I}\rangle\}\otimes\mathcal H_{I^c}^{\mathrm{dbl}}.
\]
For a unit vector $|\psi_s\rangle\in\ker\mathsf{A}_{s,I}$, the identity $\mathsf{A}_{s,I}=\mathsf{U}_s\mathsf{A}_{0,I}\mathsf{U}_s^*$ in \eqref{eq:ht-regional-kernel} implies $\mathsf{U}_s^*|\psi_s\rangle\in\ker\mathsf{A}_{0,I}$. Hence
\[
 \mathsf{U}_s^*|\psi_s\rangle=|\Omega_{\sigma_0,I}\rangle\otimes|v\rangle,\qquad \||v\rangle\|_2=1.
\]
Thus the transformed vector has the same factor on $I$ as $|\Omega_{\sigma_0}\rangle$. For every Hermitian operator $\mathsf{O}_I$ supported on $I$,
\begin{equation}\label{eq:ht-reference-expectations-overview}
 \langle\psi_s|\mathsf{U}_s\mathsf{O}_I\mathsf{U}_s^*|\psi_s\rangle=\langle\Omega_{\sigma_0}|\mathsf{O}_I|\Omega_{\sigma_0}\rangle=\langle\Omega_{\sigma_s}|\mathsf{U}_s\mathsf{O}_I\mathsf{U}_s^*|\Omega_{\sigma_s}\rangle,
\end{equation}
where the last equality uses $\mathsf{U}_s|\Omega_{\sigma_0}\rangle=|\Omega_{\sigma_s}\rangle$.
Thus, for $0\leq\mathsf{O}_J\leq\one$ supported on $J=\Ball_x(\lfloor r/2\rfloor)$, we bound its expectation by
\begin{equation}\label{eq:ht-domination-transfer-overview}
 \begin{alignedat}{2}
 &\langle\psi_s|\mathsf{O}_J|\psi_s\rangle\leq\langle\Omega_{\sigma_s}|\mathsf{U}_s\mathsf{O}_I\mathsf{U}_s^*|\Omega_{\sigma_s}\rangle+\|\mathsf{U}_s^*\mathsf{O}_J\mathsf{U}_s-\mathsf{O}_I\|&\qquad(\text{using \eqref{eq:ht-reference-expectations-overview}})&\\
 &\leq\langle\Omega_{\sigma_s}|\mathsf{O}_J|\Omega_{\sigma_s}\rangle+2\|\mathsf{U}_s^*\mathsf{O}_J\mathsf{U}_s-\mathsf{O}_I\|.
 \end{alignedat}
\end{equation}
It remains to approximate $\mathsf{U}_s^*\mathsf{O}_J\mathsf{U}_s$ by an operator supported on $I$. The unitary flow is quasi-local, and $\dist(J,I^c)\geq r/2$. Thus there is a Hermitian operator $\mathsf{O}_I$ supported on $I$ such that
\[
 \|\mathsf{U}_s^*\mathsf{O}_J\mathsf{U}_s-\mathsf{O}_I\|\leq C_b(1+r)^{-2b}.
\]
Substituting this bound into \eqref{eq:ht-domination-transfer-overview} proves \eqref{eq:local-gap-kernel-assumption}.

\section{From R\'enyi-2 decay to MLSI}\label{sec:renyi-to-mlsi}
We now pass from the R\'enyi-2 decay of Theorem~\ref{prop:general-criterion} to relative-entropy decay and MLSI, completing the proof of Theorem~\ref{thm:main-mlsi}. Section~\ref{sec:transfer} proves Theorem~\ref{thm:transfer}. Section~\ref{sec:independent-copies} obtains its tensor-power hypothesis from R\'enyi-2 decay uniform over a Hamiltonian class closed under independent copies (Corollary~\ref{cor:class-renyi-mlsi}) and checks the required structure for the Gaussian KMS Gibbs sampler. Section~\ref{sec:main-proof} combines these results with Section~\ref{sec:regime-verification} to prove Theorem~\ref{thm:main-mlsi}.

\subsection{Relative-entropy contraction from tensor-power R\'enyi-2 contraction}\label{sec:transfer}
We now prove Theorem~\ref{thm:transfer}. The idea of the proof is that, by the quantum asymptotic equipartition property~\cite[Theorem~14]{DMHB}, $\rho^{\otimes m}$ is close in trace distance to states whose max-relative entropy per copy is asymptotically at most $D(\rho\Vert\sigma)$; since $\widetilde D_2\leq D_{\max}$, the same bound holds for their R\'enyi-2 divergence per copy. Applying the contraction hypothesis to these states, using $D\leq\widetilde D_2$ on the output, and letting $m\to\infty$ turns R\'enyi-2 contraction on $m$ copies into relative-entropy contraction on a single copy.

\begin{proof}[Proof of Theorem~\ref{thm:transfer}]
For any state $\nu$ and faithful state $\tau$, define the max-relative entropy as
\[
D_{\max}(\nu\Vert\tau)
:=\inf\{r\in\mathbb R:\nu\leq e^r\tau\}.
\]
We claim that
\begin{equation}\label{eq:renyi-max-inequality}
\Dtwo(\nu\Vert\tau)\leq D_{\max}(\nu\Vert\tau).
\end{equation}
Indeed, for any $r$ satisfying $\nu\leq e^r\tau$, let
$Z=\Gamma_\tau^{-1/2}(\nu)=\tau^{-1/4}\nu\tau^{-1/4}$, so that $\Dtwo(\nu\Vert\tau)=\log\Tr(Z^2)$. Conjugating $\nu\leq e^r\tau$ by $\tau^{-1/4}$ gives
$0\leq Z\leq e^r\tau^{1/2}$, and hence, since $Z\geq0$,
\[
\Tr(Z^2)
\leq e^r\Tr(Z\tau^{1/2})
=e^r\Tr(\nu)
=e^r.
\]
Taking logarithms and then the infimum over all admissible $r$ yields the desired inequality.

Fix $\varepsilon\in(0,1)$. Applying the quantum asymptotic equipartition property for the smooth max-relative entropy~\cite[Definition~2 and Theorem~14]{DMHB}, we obtain states $\nu_m$ on $m$ copies of the system and a number $\delta(\varepsilon)\geq0$, with $\delta(\varepsilon)\to0$ as $\varepsilon\to0$, such that
\[
 \norm{\nu_m-\rho^{\otimes m}}_1\leq\varepsilon,
 \qquad
 \limsup_{m\to\infty}\frac1m
       D_{\max}(\nu_m\Vert\sigma^{\otimes m})
 \leq D(\rho\Vert\sigma)+\delta(\varepsilon).
\]

Define
$\omega_m=\Phi^{\otimes m}(\nu_m)$ and let
$\overline\omega_m=\frac{1}{m}\sum_{r=1}^m\omega_{m,r}$, where
\[
\omega_{m,r}
:=\Tr_{\{1,\ldots,m\}\setminus\{r\}}(\omega_m),
\qquad r=1,\ldots,m.
\]
Because $\Phi^{\otimes m}(\rho^{\otimes m})=\Phi(\rho)^{\otimes m}$, and both the partial trace and the channel $\Phi^{\otimes m}$ are trace-distance contractive,
\[
\|\omega_{m,r}-\Phi(\rho)\|_1
\le \|\omega_m-\Phi(\rho)^{\otimes m}\|_1=\|\Phi^{\otimes m}(\nu_m)-\Phi^{\otimes m}(\rho^{\otimes m})\|_1\le \|\nu_m-\rho^{\otimes m}\|_1\leq\varepsilon
\]
for every $r$. By the triangle inequality, $\|\overline\omega_m-\Phi(\rho)\|_1\leq\varepsilon$. Furthermore,
\[
 \begin{aligned}
 D(\overline\omega_m\Vert\sigma)
 &\leq\frac1m\sum_{r=1}^mD(\omega_{m,r}\Vert\sigma)
 \leq\frac1mD(\omega_m\Vert\sigma^{\otimes m})\\
 &\leq\frac1m\Dtwo(\omega_m\Vert\sigma^{\otimes m})
 \leq\frac a m\Dtwo(\nu_m\Vert\sigma^{\otimes m})\leq \frac a m D_{\max}(\nu_m\Vert\sigma^{\otimes m}),
 \end{aligned}
\]
where we use convexity of relative entropy in the first inequality, and superadditivity of relative entropy with respect to the product state $\sigma^{\otimes m}$ (a consequence of subadditivity of the von Neumann entropy) in the second inequality. The third inequality uses $D\leq\Dtwo$ from Section~\ref{sec:notation-entropies}. The second-to-last inequality uses the contraction hypothesis on $\Phi^{\otimes m}$. The last inequality uses \eqref{eq:renyi-max-inequality}.

It remains to compare $\overline\omega_m$ with $\Phi(\rho)$. Because $\sigma$ is faithful, $D(\cdot\Vert\sigma)$ is continuous on the compact set of one-copy states, hence uniformly continuous in trace distance: the modulus of continuity
\[
\theta(\varepsilon):=\sup\bigl\{|D(\omega\Vert\sigma)-D(\omega'\Vert\sigma)|:\ \|\omega-\omega'\|_1\leq\varepsilon\bigr\}
\]
satisfies $\theta(\varepsilon)\to0$ as $\varepsilon\to0$. Since $\|\overline\omega_m-\Phi(\rho)\|_1\leq\varepsilon$, the previous display gives
\[
D(\Phi(\rho)\Vert\sigma)
\leq D(\overline\omega_m\Vert\sigma)+\theta(\varepsilon)
\leq\frac a m D_{\max}(\nu_m\Vert\sigma^{\otimes m})+\theta(\varepsilon).
\]
Taking the $\limsup_{m\to\infty}$ and using the bound on $D_{\max}(\nu_m\Vert\sigma^{\otimes m})$ above gives
\[
D(\Phi(\rho)\Vert\sigma)\leq a\bigl(D(\rho\Vert\sigma)+\delta(\varepsilon)\bigr)+\theta(\varepsilon).
\]
Letting $\varepsilon\to0$ proves the claim.
\end{proof}

For a primitive KMS-symmetric generator $\mathcal L$ with invariant state $\sigma$, if the hypothesis of Theorem~\ref{thm:transfer} holds for $\Phi=e^{t\mathcal L_*}$ and $a=e^{-\kappa t}$ for every $t\geq0$, Theorem~\ref{thm:transfer} and Lemma~\ref{lem:mlsi-decay} give
\[
 D(e^{t\mathcal L_*}\rho\Vert\sigma)\leq e^{-\kappa t}D(\rho\Vert\sigma),
 \qquad \alpha_1(\mathcal L)\geq\kappa/2.
\]

\subsection{Hamiltonian classes closed under independent copies}\label{sec:independent-copies}
R\'enyi-2 decay at a size-independent rate also supplies the tensor-power contraction required by Theorem~\ref{thm:transfer}, provided the Hamiltonian class contains independent copies and the sampler on them is the tensor sum of the original sampler. We state this implication for a general class $\mathfrak H$ of Hamiltonians at fixed $\beta$. For each $H\in\mathfrak H$, let $\mathcal L_H$ be a primitive KMS-symmetric Lindblad generator with invariant state $\sigma_H=e^{-\beta H}/\Tr e^{-\beta H}$.
For $m$ independent copies, write $H^{(m)}=\sum_{j=1}^mH^{(j)}$, where $H^{(j)}$ acts as $H$ on copy $j$ and as the identity on the other factors. The structural requirements are
\begin{equation}\label{eq:criterion-copy-structure}
 H^{(m)}\in\mathfrak H,\qquad
 \sigma_{H^{(m)}}=\sigma_H^{\otimes m},\qquad
 \mathcal L_{H^{(m)}}=\sum_{j=1}^m\mathcal L_H^{(j)}.
\end{equation}
Here $\mathcal L_H^{(j)}$ acts only on copy $j$, with the same time normalization as $\mathcal L_H$.

\begin{corollary}[Uniform R\'enyi-2 decay implies MLSI]\label{cor:class-renyi-mlsi}
Suppose \eqref{eq:criterion-copy-structure} holds for every $H\in\mathfrak H$ and $m\geq1$. If a common $\kappa>0$ satisfies
\[
 \widetilde D_2(e^{t\mathcal L_{H,*}}\rho\Vert\sigma_H)
 \leq e^{-\kappa t}\widetilde D_2(\rho\Vert\sigma_H)
 \qquad(H\in\mathfrak H,\ t\geq0,\ \rho\text{ any state}),
\]
then, for every $H\in\mathfrak H$,
\[
 D(e^{t\mathcal L_{H,*}}\rho\Vert\sigma_H)
 \leq e^{-\kappa t}D(\rho\Vert\sigma_H),\qquad
 \alpha_1(\mathcal L_H)\geq\kappa/2.
\]
\end{corollary}
\begin{proof}
Fix $H$. We need the R\'enyi-2 contraction hypothesis of Theorem~\ref{thm:transfer} on every tensor power of $e^{t\mathcal L_{H,*}}$. The tensor-sum identity in \eqref{eq:criterion-copy-structure} implies
\[
 e^{t\mathcal L_{H^{(m)},*}}
 =(e^{t\mathcal L_{H,*}})^{\otimes m}.
\]
Since $H^{(m)}$ belongs to the same class, the assumed contraction bound applies to every state $\eta$ on its Hilbert space, including states correlated between copies:
\[
 \widetilde D_2\!\left((e^{t\mathcal L_{H,*}})^{\otimes m}(\eta)
                      \Vert\sigma_H^{\otimes m}\right)
 \leq e^{-\kappa t}\widetilde D_2(\eta\Vert\sigma_H^{\otimes m}).
\]
The rate is the same for all $m$. Theorem~\ref{thm:transfer}, with $\Phi=e^{t\mathcal L_{H,*}}$ and $a=e^{-\kappa t}$, therefore gives relative-entropy contraction at rate $\kappa$, and Lemma~\ref{lem:mlsi-decay} gives $\alpha_1(\mathcal L_H)\geq\kappa/2$.
\end{proof}

For the Gaussian KMS Gibbs sampler, a coupling $S^{(j)}$ on copy $j$ evolves within that copy, so its filtered jumps and coherent correction in \eqref{eq:CKG} do as well:
\[
 e^{itH^{(m)}}S^{(j)}e^{-itH^{(m)}}
 =(e^{itH}Se^{-itH})^{(j)}
 \quad\Longrightarrow\quad
 \mathcal L_{H^{(m)}}=\sum_{j=1}^m\mathcal L_H^{(j)}.
\]
The one-dimensional and high-temperature Hamiltonian classes are closed under independent copies with unchanged local parameters, so \eqref{eq:criterion-copy-structure} holds in both settings.

\subsection{Proof of the main theorem}\label{sec:main-proof}
\begin{proof}[Proof of Theorem~\ref{thm:main-mlsi}]
By Corollary~\ref{cor:class-renyi-mlsi} and \eqref{eq:criterion-copy-structure}, it suffices to obtain a uniform R\'enyi-2 contraction rate in each class. The verifications in Section~\ref{sec:regime-verification} give the hypotheses of Theorem~\ref{prop:general-criterion}, with $a=1$ and $p_0=b=10$ in one dimension, and $a=1$ and $p_0=b=2d+2$ at high temperature. That theorem gives the rate
\[
 \kappa=\frac{2}{C_\beta},\qquad
 C_\beta=(C_0+\beta C_S)e^{\beta C_E},\qquad C_0=\frac{2q}{\nu_q}.
\]
In one dimension, $C_E,C_S$ depend only on $\beta,q$, so $\kappa$ has the same dependence. At high temperature, Lemmas~\ref{lem:comparison-families} and~\ref{lem:derivative-reduction} give $C_E,C_S\leq CJ$, with $C$ depending only on $d,R_0,q$. Hence a common rate can be chosen as
\[
 \kappa=\frac{2e^{-C\beta_0}}{C_0+C\beta_0}>0.
\]
Since $\beta C_E,\beta C_S\leq C\beta J\leq C\beta_0$, this rate satisfies $\kappa\leq2/C_\beta$. In both settings, \eqref{eq:criterion-entropy} therefore gives $\alpha_2(\mathcal L_\beta)\geq2/C_\beta\geq\kappa$. Corollary~\ref{cor:class-renyi-mlsi} now gives relative-entropy contraction and MLSI in both settings.
\end{proof}

\appendix
\section{The reference inequality at infinite temperature}\label{app:product}
In Section~\ref{sec:vectorization}~\eqref{eq:product-state-reference}, we claimed that an entropy inequality at the maximally mixed product state holds with a precise constant. We formally state and prove that inequality in the following.

Recall the Hilbert space in~\eqref{eqn:hilbert-space} and the Gaussian KMS Gibbs sampler in~\eqref{eq:CKG}. At $s=0$, defined by the limit in Section~\ref{sec:temperature-path}, the Gaussian KMS Gibbs sampler becomes a sum of independent single-site depolarizing generators, where each site is a $q$-qubit system. The entropy inequality is summarized in the following lemma.
\begin{lemma}[The product-state entropy inequality]\label{lem:product-entropy}
For $\sigma_0=I/2^{q|\Lambda|}$, the vectorization of the Gibbs sampler in~\eqref{eq:CKG} satisfies
\begin{equation}\label{eq:native-product-generator}
 -\mathsf{L}_{\sigma_0}=\nu_q\sum_x\mathsf{Q}_x,\qquad
 \nu_q=\frac{4^qe^{-1/4}}{\sqrt2},\qquad
 \mathsf{Q}_x=\one-\left|I_{2^q}/2^{q/2}\right\rangle\left\langle I_{2^q}/2^{q/2}\right|_x.
\end{equation}
For every Hermitian $Y$ with $\|Y\|_2=1$,
\begin{equation}\label{eq:product-entropy-inequality}
 D(Y^2\Vert\sigma_0)\leq\frac{2q}{\nu_q}\langle Y|(-\mathsf{L}_{\sigma_0})|Y\rangle.
\end{equation}
\end{lemma}
\begin{proof}[Proof of Lemma~\ref{lem:product-entropy}]

We first derive~\eqref{eq:native-product-generator}. For $0<s\leq\beta$, substituting $\omega=z/s$ in \eqref{eq:CKG} and \eqref{eq:CKG-correction} and using \eqref{eq:native-rescaled-jump} gives
\begin{equation}\label{eq:CKG-rescaled}
 \begin{aligned}
 \mathcal L_{s,*}(\rho)
 &=-i[H_s^{\mathrm{coh}},\rho]+\sum_{x,u}\int_\R e^{-(z+1)^2/2}
   \Bigl(\widehat J^s_{x,u}(z)\rho\widehat J^s_{x,u}(z)^*
   -\tfrac12\bigl\{\widehat J^s_{x,u}(z)^*\widehat J^s_{x,u}(z),\rho\bigr\}\Bigr)dz,\\
 Q_s&=\sum_{x,u}\int_\R e^{-(z+1)^2/2}\,\widehat J^s_{x,u}(z)^*\widehat J^s_{x,u}(z)\,dz.
 \end{aligned}
\end{equation}
The integral in \eqref{eq:native-rescaled-jump} extends to $s=0$, and the coherent correction vanishes there because $\tanh(0)=0$. Thus \eqref{eq:CKG-rescaled} at $s=0$ gives the limiting generator $\mathcal L_{0,*}$.

The Pauli identity
\[
 \sum_{u=1}^{4^q-1}[S_{x,u},[S_{x,u},Y]]
 =2\cdot4^q\bigl(Y-I_x\otimes\Tr_xY/2^q\bigr)
\]
identifies the site depolarizing generator. Setting $s=0$ in \eqref{eq:native-rescaled-jump} gives
\[
 \widehat J^0_{x,u}(z)=(2\pi)^{-1/4}e^{-z^2/4}S_{x,u},\qquad
 \int_\R\frac{e^{-z^2/2-(z+1)^2/2}}{\sqrt{2\pi}}\,dz
 =\frac{e^{-1/4}}{\sqrt2}.
\]
Since $H_0^{\mathrm{coh}}=0$, \eqref{eq:CKG-rescaled} at $s=0$ gives $\mathcal L_0(X)=(e^{-1/4}/\sqrt2)\sum_{x,u}(S_{x,u}XS_{x,u}-X)$, proving \eqref{eq:native-product-generator}.

Next, we prove~\eqref{eq:product-entropy-inequality}. We first consider the case $Y>0$, so that the matrix $X$ below is positive definite, as required in~\eqref{eqn:l2-entropy}. Let $d_\Lambda=2^{q|\Lambda|}=\dim\mathcal H$, so that $\sigma_0=I/d_\Lambda$, and set $X=\sqrt{d_\Lambda}\,Y$. The weighting map satisfies $\Gamma_{\sigma_0}^{1/2}(X)=d_\Lambda^{-1/2}X=Y$, so $\|X\|_{2,\sigma_0}=\|Y\|_2=1$ and $\eta_X=Y^2$. By~\eqref{eqn:l2-entropy},
\[
 2\Ent_{2,\sigma_0}(X)
 =D(Y^2\Vert\sigma_0).
\]
Let $\operatorname{tr}_0=\Tr/d_\Lambda$ be the normalized trace, and let $\mathcal Q_i(X)=X-I_i\otimes\Tr_i(X)/2$ remove the identity component at the $i$th individual qubit, with vectorized representative $\mathsf{Q}_i$, defined by $\mathsf{Q}_i|X\rangle=|\mathcal Q_i(X)\rangle$. Product qubit hypercontractivity~\cite{King} states
\begin{equation}\label{eq:product-hypercontractivity}
 \|e^{-t\sum_i\mathcal Q_i}X\|_{p(t),\operatorname{tr}_0}
 \leq\|X\|_{2,\operatorname{tr}_0},\qquad p(t)=1+e^{2t}\,.
\end{equation}
Here, $\|X\|_{p,\operatorname{tr}_0}=\bigl(\operatorname{tr}_0(|X|^p)\bigr)^{1/p}$ is the normalized Schatten $p$-norm. Since the Gibbs state at $s=0$ is $\sigma_0=I/d_\Lambda$, we have $\operatorname{tr}_0(A)=\Tr(\sigma_0 A)$. In particular, the normalized Hilbert--Schmidt norm agrees with the KMS norm defined in~\eqref{eqn:KMS_norm}:
\[
 \|X\|_{2,\operatorname{tr}_0}
 =\|\sigma_0^{1/4}X\sigma_0^{1/4}\|_2
 =\|X\|_{2,\sigma_0}.
\]
Set $\mathcal A=\sum_i\mathcal Q_i$. At $\sigma_0=I/d_\Lambda$, the $L^2$ relative entropy takes the form
\[
 2\Ent_{2,\sigma_0}(X)
 =\operatorname{tr}_0(X^2\log X^2)
 -\operatorname{tr}_0(X^2)\log\operatorname{tr}_0(X^2).
\]
For $X>0$, the hypercontractive bound~\eqref{eq:product-hypercontractivity} is an equality at $t=0$, so the right derivative of its logarithm is nonpositive there. Using $p(0)=2$ and $p'(0)=2$ gives
\[
 \left.\frac{d}{dt}\log\|e^{-t\mathcal A}X\|_{p(t),\operatorname{tr}_0}\right|_{t=0^+}
 =\frac{\Ent_{2,\sigma_0}(X)
 -\operatorname{tr}_0\bigl(X\mathcal A(X)\bigr)}
 {\operatorname{tr}_0(X^2)}\leq0.
\]
Consequently,
\begin{equation}\label{eq:qubit-l2-entropy}
 \frac{1}{2}D(Y^2\Vert\sigma_0)=\Ent_{2,\sigma_0}(X)
 \leq\sum_i\operatorname{tr}_0\bigl(X\mathcal Q_i(X)\bigr).
\end{equation}
The right-hand side of~\eqref{eq:qubit-l2-entropy} can be rewritten:
\[
 \operatorname{tr}_0\bigl(X\mathcal Q_i(X)\bigr)
 =\frac1{d_\Lambda}\Tr\bigl[(\sqrt{d_\Lambda}\,Y)\mathcal Q_i(\sqrt{d_\Lambda}\,Y)\bigr]
 =\Tr\bigl[Y\mathcal Q_i(Y)\bigr]
 =\langle Y|\mathsf{Q}_i|Y\rangle.
\]
Because $\sum_{i\text{ at }x}\mathsf{Q}_i\leq q\mathsf{Q}_x$, we have
\[
 D(Y^2\Vert\sigma_0)
 \leq2\sum_i\langle Y|\mathsf{Q}_i|Y\rangle
 \leq2q\sum_x\langle Y|\mathsf{Q}_x|Y\rangle.
\]
This proves~\eqref{eq:product-entropy-inequality} for $Y>0$.

Finally, we extend the inequality to all Hermitian $Y$ with $\|Y\|_2=1$. For Hermitian $Y$ with eigenvalues $y_i$, each matrix $S$ obeys
\[
 \|[S,|Y|]\|_2^2=\sum_{i,j}|S_{ij}|^2(|y_i|-|y_j|)^2
 \leq\sum_{i,j}|S_{ij}|^2(y_i-y_j)^2=\|[S,Y]\|_2^2.
\]
By the Pauli identity above and~\eqref{eq:native-product-generator}, $\langle Y|(-\mathsf L_{\sigma_0})|Y\rangle=\frac{\nu_q}{2\cdot4^q}\sum_{x,u}\|[S_{x,u},Y]\|_2^2$. Let $Y$ be Hermitian with $\|Y\|_2=1$. For $\varepsilon>0$, apply the case already proved to the positive definite matrix $Y_\varepsilon=(|Y|+\varepsilon I)/\||Y|+\varepsilon I\|_2$. Since $[S_{x,u},|Y|+\varepsilon I]=[S_{x,u},|Y|]$, the commutator bound above gives
\[
 D(Y_\varepsilon^2\Vert\sigma_0)
 \leq\frac{2q}{\nu_q}\,\frac{\langle Y|(-\mathsf L_{\sigma_0})|Y\rangle}{\||Y|+\varepsilon I\|_2^2}.
\]
Letting $\varepsilon\downarrow0$ and using $|Y|^2=Y^2$, $\||Y|\|_2=1$, and continuity of $D(\cdot\Vert\sigma_0)$ for faithful $\sigma_0$ proves~\eqref{eq:product-entropy-inequality} for $Y$.
\end{proof}

\section{Quasi-locality of the unitary flow}\label{app:transport-details}

In this section, we first establish Proposition~\ref{prop:regional-mixing}. We then turn to the proof of Proposition~\ref{prop:flow-locality}. Both proofs use the spatial truncation bound in Lemma~\ref{lem:physical-localization}.

\subsection{From regional rapid mixing to local rapid mixing}\label{app:regional-mixing}
We first compare the full evolution with the sampler constructed from a regional Hamiltonian. This estimate will also be used to localize the Poisson inverse in Proposition~\ref{prop:profile}.

\begin{lemma}[Spatial truncation of the Gaussian KMS Gibbs sampler]\label{lem:physical-localization}
For the Gaussian KMS Gibbs sampler \eqref{eq:CKG} and interactions satisfying \eqref{eq:growth}, let $h$ be supported in $\Ball_x(R_0)$. For $R\geq R_0$, let $\mathcal L_s^{[R]}$ be the sampler constructed from $H_{\Ball_x(R)}$. There are constants $C,v,\mu>0$, independent of system size, such that
\begin{equation}\label{eq:physical-native}
 \|(e^{t\mathcal L_s}-e^{t\mathcal L_s^{[R]}})h\|_\infty\leq Ce^{vt-\mu R}\|h\|_\infty
 \qquad(t\geq0,\ 0\leq s\leq\beta).
\end{equation}
\end{lemma}
\begin{proof}
We prove \eqref{eq:physical-native} by cutting the interactions across the boundary of the ball. We estimate the change of the generator along the cut evolution, then integrate this estimate using Duhamel's formula.

Write $B=\Ball_x(R)$ and $B^c=\Lambda\setminus B$. We cut the Hamiltonian by deleting the terms whose supports meet both regions:
\begin{equation}\label{eq:cut-hamiltonian}
 V_{\partial B}:=\sum_{\substack{Z\cap B\ne\varnothing\\Z\cap B^c\ne\varnothing}}h_Z,
 \qquad H^{\rm cut}:=H-V_{\partial B}=H_B+H_{B^c}.
\end{equation}
Define the cut sampler $\mathcal L_s^{\rm cut}$ from $H^{\rm cut}$ in \eqref{eq:cut-hamiltonian}. Because the two Hamiltonians $H_B$ and $H_{B^c}$ act on separate tensor factors, we have
$
 \mathcal L_s^{\rm cut}
 =\mathcal L_s^B\otimes\operatorname{id}_{B^c}
  +\operatorname{id}_B\otimes\mathcal L_s^{B^c}.
$
Since $h$ is supported in $B$ and $\mathcal L_s^{B^c}(I_{B^c})=0$, we have:
\begin{equation}\label{eq:cut-regional-evolution}
 e^{u\mathcal L_s^{\rm cut}}h=e^{u\mathcal L_s^{[R]}}h,
 \qquad u\geq0.
\end{equation}
Thus, it suffices to prove \eqref{eq:physical-native} with $\mathcal L_s^{[R]}$ replaced by $\mathcal L_s^{\rm cut}$.

We first justify the comparison between the full and cut generators. Let $\mathcal L_{s,y}[F]$ denote the contribution centered at $y$ to the sampler constructed from a Hamiltonian $F$, and let $\mathcal L_{s,y}^{[r]}[F]$ be that contribution constructed instead from $F_{\Ball_y(r)}$. Set
\[
 \mathcal D_{s,y,r}[F]
 :=\mathcal L_{s,y}^{[r]}[F]-\mathcal L_{s,y}^{[r-1]}[F],
 \qquad \mathcal L_{s,y}^{[-1]}[F]:=0.
\]
According to~\cite[Eq.~(28) and Appendix~C.1, Eqs.~(54), (62), and (68)--(70)]{RFA}, the sampler can be decomposed into local superoperators supported on balls of increasing radius, with exponentially decaying norms:
\begin{equation}\label{eq:physical-generator-shells}
 \begin{gathered}
 \mathcal L_s[F]=\sum_{y\in\Lambda}\sum_{r\geq0}\mathcal D_{s,y,r}[F],
 \qquad \supp\mathcal D_{s,y,r}[F]\subseteq\Ball_y(r),\\
 \mathcal D_{s,y,r}[F](I)=0,\qquad
 \|\mathcal D_{s,y,r}[F]\|_{\infty\to\infty}
 \leq b_r:=C_{\mathcal D}e^{-\xi r}.
 \end{gathered}
\end{equation}
This holds uniformly for $F=H$ and $F=H^{\rm cut}$ and $0\leq s\leq\beta$.

Next, we bound the spread of the cut evolution, so that the effect of the deleted interactions can be estimated. Define
\[
 \delta_z:=\operatorname{id}-\frac{I_z}{2^q}\otimes \mathrm{Tr}_z[\cdot],\qquad
 A(u):=e^{u\mathcal L_s^{\rm cut}}h,\qquad
 d_z(u):=\|\delta_z A(u)\|_\infty.
\]
A map $\Phi$ supported on $U$ with $\Phi(I)=0$ satisfies
\begin{equation}\label{eq:local-map-oscillation}
 \|\Phi(A)\|_\infty
 \leq\|\Phi\|_{\infty\to\infty}
       \sum_{w\in U}\|\delta_w A\|_\infty.
\end{equation}
The contractive-semigroup argument in~\cite[Appendix~B.1, Eqs.~(21)--(24)]{RFA}, with the full contribution $\mathcal L_{s,z}[H^{\rm cut}]$ kept as the source term, gives the following integral inequality. Indeed, $\sum_{y\ne z}\mathcal L_{s,y}[H^{\rm cut}]$ is a Lindblad generator, each local term commutes with $\delta_z$ unless its support contains $z$, and $\|\delta_z\|_{\infty\to\infty}\leq2$. Applying \eqref{eq:local-map-oscillation} to the terms in \eqref{eq:physical-generator-shells} yields
\[
 d_z(u)\leq d_z(0)+\int_0^u\sum_w M_{zw}d_w(v)\,dv,
 \qquad
 M_{zw}:=4\sum_{\substack{y\in\Lambda,\ r\geq0\\z,w\in\Ball_y(r)}}b_r.
\]
Here the sum runs over $y\in\Lambda$, integer $r\geq0$, and $b_r$ is defined in \eqref{eq:physical-generator-shells}.

For $0<\eta<\xi/2$, define
\[
 v_\eta:=\sup_z\sum_w e^{\eta\dist(z,w)}M_{zw}
 \leq C\sum_{r\geq0}(1+r)^{2d}e^{-(\xi-2\eta)r}<\infty.
\]
We note that
\[
(M^n)_{zw}\leq v_\eta^n e^{-\eta\dist(z,w)}\,.
\]
Iteration then gives $d(u)\leq e^{uM}d(0)$ and $(e^{uM})_{zw}\leq e^{v_\eta u-\eta\dist(z,w)}$. Because $d_w(0)=0$ outside $\supp h\subseteq\Ball_x(R_0)$, $
\dist(z,w)\geq\dist(z,x)-R_0$, and $\|\delta_z\|_{\infty\to\infty}\leq2$, we obtain
\begin{equation}\label{eq:cut-observable-propagation}
 \|\delta_z e^{u\mathcal L_s^{\rm cut}}h\|_\infty
 \leq C e^{v_\eta u-\eta\dist(z,x)}\|h\|_\infty.
\end{equation}

It remains to locate the terms changed by cutting. Define
\[
 \mathcal T_{y,r}:=\mathcal D_{s,y,r}[H]
                         -\mathcal D_{s,y,r}[H^{\rm cut}].
\]
These maps have support in $\Ball_y(r)$, annihilate $I$, and satisfy $\|\mathcal T_{y,r}\|_{\infty\to\infty}\leq2b_r$. Moreover, $\mathcal T_{y,r}=0$ whenever $\Ball_y(r)\subseteq B$, since the two Hamiltonians have identical interactions inside $B$. Thus $\mathcal T_{y,r}\ne0$ requires $\dist(x,y)+r>R$.

To sum the changed terms, apply \eqref{eq:local-map-oscillation} and \eqref{eq:cut-observable-propagation}. Since $\dist(x,z)\geq\dist(x,y)-r$ for $z\in\Ball_y(r)$ and $|\Ball_y(r)|\leq C(1+r)^d$, we obtain
\[
 \|(\mathcal L_s-\mathcal L_s^{\rm cut})A(u)\|_\infty
 \leq C e^{v_\eta u}\|h\|_\infty
   \sum_{\substack{y\in\Lambda,\ r\geq0\\\dist(x,y)+r>R}}
        (1+r)^d e^{-(\xi-\eta)r-\eta\dist(x,y)}.
\]
Grouping sites by $\rho=\dist(x,y)$ bounds this sum by the tail, over $r+\rho>R$, of the convolution of $(1+r)^d e^{-(\xi-\eta)r}$ and $C(1+\rho)^{d-1}e^{-\eta\rho}$. The scalar exponential-convolution estimate in~\cite[proof of Lemma~5.5, the exponential case following Eq.~(25)]{CLMPG} bounds this tail by $Ce^{-\mu R}$ for $0<\mu<\eta<\xi/2$. Hence
\[
 \|(\mathcal L_s-\mathcal L_s^{\rm cut})A(u)\|_\infty
 \leq C e^{v_\eta u-\mu R}\|h\|_\infty.
\]
Using \eqref{eq:cut-regional-evolution}, Duhamel's formula and operator-norm contractivity yield
\[
 \begin{aligned}
 \|(e^{t\mathcal L_s}-e^{t\mathcal L_s^{[R]}})h\|_\infty=\left\|\int_0^t e^{(t-u)\mathcal L_s}
       (\mathcal L_s-\mathcal L_s^{\rm cut})
       e^{u\mathcal L_s^{\rm cut}}h\,du\right\|_\infty\leq C e^{(v_\eta+1)t-\mu R}\|h\|_\infty.
 \end{aligned}
\]
This proves \eqref{eq:physical-native}.
\end{proof}

\begin{proof}[Proof of Proposition~\ref{prop:regional-mixing}]
Fix $\nu>0$, $s\in[0,\beta]$, and an interaction term $h=h_Z$ supported in $\Ball_x(R_0)$. Our goal is the bound
\[
 \|e^{t\mathcal L_s}(h-\Tr(\sigma_sh)I)\|_\infty
 \leq C_\nu(1+t)^{-\nu}\|h\|_\infty,\qquad t\geq0,
\]
with $C_\nu$ independent of $s$, $Z$, and the volume. Set $\tau(R)=\max\{1,C_{\mathrm{reg}}\}(1+R)^a$, so the hypothesis gives $t_{\rm mix}^{[R],s}(1/4)\leq\tau(R)$. Lemma~\ref{lem:physical-localization} supplies the spatial truncation bound. We combine it with regional rapid mixing on a ball whose radius grows linearly with $t$, and then convert the resulting stretched-exponential decay to the required polynomial bound.

\proofstep{Step 1: regional rapid mixing for all $t$}
Fix $R\geq R_0$. For the evolution inside the ball, duality gives
\[
 \begin{aligned}
 \|e^{t\mathcal L_s^{[R]}}h-\Tr(\sigma_s^{[R]}h)I\|_\infty
 &=\sup_\rho\left|\Tr\!\left[\bigl(e^{t\mathcal L_{s,*}^{[R]}}\rho-\sigma_s^{[R]}\bigr)h\right]\right|\\
 &\leq\|h\|_\infty\sup_\rho\|e^{t\mathcal L_{s,*}^{[R]}}\rho-\sigma_s^{[R]}\|_1
 =:\|h\|_\infty m_R(t).
 \end{aligned}
\]
The hypothesis gives $m_R(\tau(R))\leq1/4$. We extend this bound to all $t$ by iterating over intervals of length $\tau(R)$. For a nonzero traceless Hermitian $W$, its positive and negative parts give $W=\tfrac12\|W\|_1(\rho_+-\rho_-)$ with states $\rho_\pm$. Hence
\[
 \begin{aligned}
 \|e^{\tau(R)\mathcal L_{s,*}^{[R]}}W\|_1\leq\frac{\|W\|_1}{2}\sum_{\epsilon=\pm}
       \|e^{\tau(R)\mathcal L_{s,*}^{[R]}}\rho_\epsilon-\sigma_s^{[R]}\|_1\leq\tfrac14\|W\|_1.
 \end{aligned}
\]
Applying this to $W=e^{t\mathcal L_{s,*}^{[R]}}\rho-\sigma_s^{[R]}$ gives
\[
 m_R(t+\tau(R))\leq\tfrac14m_R(t),\qquad m_R(t)\leq2,
\]
and therefore
\begin{equation}\label{eq:regional-mixing-iteration}
 m_R(t)\leq2(1/4)^{\lfloor t/\tau(R)\rfloor}
       \leq8e^{-(\log4)t/\tau(R)}.
\end{equation}

\proofstep{Step 2: transfer to the full system and choice of radius}
Fix $R\geq R_0$. We first replace the Gibbs state in the target by the regional Gibbs state. Stationarity gives
\[
 \begin{aligned}
 |\Tr(\sigma_sh)-\Tr(\sigma_s^{[R]}h)|=\left|\Tr\!\left[\sigma_s\bigl(e^{t\mathcal L_s}h-\Tr(\sigma_s^{[R]}h)I\bigr)\right]\right|\leq\|e^{t\mathcal L_s}h-\Tr(\sigma_s^{[R]}h)I\|_\infty.
 \end{aligned}
\]
The triangle inequality and Lemma~\ref{lem:physical-localization} therefore give
\begin{equation}\label{eq:regional-decay-reduction}
 \begin{aligned}
 \|e^{t\mathcal L_s}(h-\Tr(\sigma_sh)I)\|_\infty
 &\leq2\|e^{t\mathcal L_s}h-\Tr(\sigma_s^{[R]}h)I\|_\infty\\
 &\leq2Ce^{vt-\mu R}\|h\|_\infty
       +2\|e^{t\mathcal L_s^{[R]}}h-\Tr(\sigma_s^{[R]}h)I\|_\infty.
 \end{aligned}
\end{equation}

Using~\eqref{eq:regional-mixing-iteration} gives
\[
 \|e^{t\mathcal L_s}(h-\Tr(\sigma_sh)I)\|_\infty
 \leq2\left[Ce^{vt-\mu R}+8e^{-(\log4)t/\tau(R)}\right]\|h\|_\infty.
\]
We now choose $R=R_0+\lfloor Kt\rfloor$ with $\mu K>v+1$. Since $\tau$ is nondecreasing,
\[
 Ce^{vt-\mu R}\leq C'e^{-t},\qquad
 e^{-(\log4)t/\tau(R)}\leq e^{-(\log4)t/\tau(R_0+Kt)}.
\]
Substitution gives the intermediate bound
\begin{equation}\label{eq:profile-from-mixing}
 \|e^{t\mathcal L_s}(h-\Tr(\sigma_sh)I)\|_\infty
 \leq C\left[e^{-ct}
       +\exp\!\left(-\frac{ct}{\tau(R_0+Kt)}\right)\right]\|h\|_\infty,
\end{equation}
where $C,c,K>0$ are independent of $s$, $Z$, and the volume.

\proofstep{Step 3: polynomial decay}
Finally, we compare the scalar decay function in \eqref{eq:profile-from-mixing} with $(1+t)^{-\nu}$. By the definition of $\tau$,
\[
 \tau(R_0+Kt)\leq A(1+t)^a,
 \qquad A=\max\{1,C_{\mathrm{reg}}\}(1+R_0+K)^a.
\]
Consequently, for $t\geq1$,
\[
 \frac{t}{\tau(R_0+Kt)}
 \geq\frac{t}{A(1+t)^a}
 \geq\frac{2^{-a}}{A}t^{1-a}.
\]
Also $t\geq t^{1-a}$ for $t\geq1$, since $a\geq0$. Hence \eqref{eq:profile-from-mixing} implies
\[
 \|e^{t\mathcal L_s}(h-\Tr(\sigma_sh)I)\|_\infty
 \leq C_1e^{-c_1t^{1-a}}\|h\|_\infty,
 \qquad t\geq1.
\]
Here $a<1$ is essential: it gives $t^{1-a}/\log(1+t)\to\infty$, and therefore
\[
 M_\nu:=\sup_{t\geq1}(1+t)^\nu e^{-c_1t^{1-a}}<\infty.
\]
For $0\leq t\leq1$, operator-norm contractivity bounds the centered evolution by $2\|h\|_\infty$. Thus, with $C_\nu=\max\{C_1M_\nu,2^{\nu+1}\}$, we obtain
\[
 \|e^{t\mathcal L_s}(h-\Tr(\sigma_sh)I)\|_\infty
 \leq C_\nu(1+t)^{-\nu}\|h\|_\infty,
 \qquad t\geq0.
\]
This is Condition~\ref{cond:transport-mixing}. Since $\nu>0$ was arbitrary, the proposition follows.
\end{proof}

\subsection{Locality of the transport generator}\label{app:transport-coefficients}

We prove Proposition~\ref{prop:flow-locality}. Fix $p>0$ and assume Condition~\ref{cond:transport-mixing} with $\nu>p+6d+10$. Our goal is to prove the volume-independent bound in~\eqref{eq:flow-locality-target}. Specifically, we will prove
\begin{equation}\label{eq:flow-locality-target-2}
 \sup_{0\leq s\leq\beta}\|\mathsf G_s\|_{p,\mathrm{loc}}\leq CJ,
 \qquad \mathsf G_s=[\mathsf M_{f_s},\mathsf L_{\sigma_s}].
\end{equation}
For the orders $P$ and $\ell$ chosen below, we first use Condition~\ref{cond:transport-mixing} and Lemma~\ref{lem:physical-localization} to prove $\|X_{s,Z}\|_{P,x_Z}\leq C\|h_Z\|_\infty$, where $X_{s,Z}$ is defined in \eqref{eq:local-poisson-integral}. We then estimate $\mathsf M_{f_s}$, where $f_s=\sum_Z\sech(s\ad_H/4)X_{s,Z}$, and the terms $\mathsf L_{\sigma_s,x}$ in $\mathsf L_{\sigma_s}=\sum_x\mathsf L_{\sigma_s,x}$ separately, obtaining
\[
 \|\mathsf M_{f_s}\|_{\ell,\mathrm{loc}}\leq CJ,
 \qquad \sup_x\|\mathsf L_{\sigma_s,x}\|_{\ell,x}\leq C.
\]
Finally, applying the commutator estimate to $\mathsf G_s=\sum_x[\mathsf M_{f_s},\mathsf L_{\sigma_s,x}]$ gives
\[
 \|\mathsf G_s\|_{p,\mathrm{loc}}
 \leq C\|\mathsf M_{f_s}\|_{\ell,\mathrm{loc}}
       \sup_x\|\mathsf L_{\sigma_s,x}\|_{\ell,x}
 \leq CJ.
\]

\begin{proof}[Proof of Proposition~\ref{prop:flow-locality}]

Set
\[
 \ell=p+d+2,\qquad P=p+6d+8.
\]
The hypothesis gives $\nu>P+2$, while $\ell+d+2=p+2d+4\leq P$. All constants below are uniform in $s\in[0,\beta]$ and the volume.

\proofstep{Step 1: local Poisson inverse}
For each interaction term $h_Z$, define
\begin{equation}\label{eq:local-poisson-integral}
 X_{s,Z}=-\frac12\int_0^\infty e^{t\mathcal L_s}
       \bigl(h_Z-\Tr(\sigma_sh_Z)I\bigr)\,dt.
\end{equation}
Choose $x_Z\in Z$ as the center for the locality estimates below. Our first task is to bound $\|X_{s,Z}\|_{P,x_Z}$ by $\|h_Z\|_\infty$. We use local rapid mixing for the long-time tail and Lemma~\ref{lem:physical-localization} for short times. This gives the following proposition:

\begin{proposition}[A local Poisson inverse from power-law mixing]\label{prop:profile}
For the Gaussian KMS Gibbs sampler \eqref{eq:CKG}, assume Condition~\ref{cond:transport-mixing} with $\nu>P+2$, where $P\geq0$. Choose $x_Z\in Z$. The integral \eqref{eq:local-poisson-integral} satisfies
\begin{equation}\label{eq:local-poisson}
 \begin{gathered}
 \mathcal L_sX_{s,Z}=\tfrac12\bigl(h_Z-\Tr(\sigma_sh_Z)I\bigr),
 \qquad \Tr(\sigma_sX_{s,Z})=0,\\
 \|X_{s,Z}\|_{P,x_Z}\leq C_P\|h_Z\|_\infty.
 \end{gathered}
\end{equation}
\end{proposition}
\begin{proof}[Proof of Proposition~\ref{prop:profile}]

First, Condition~\ref{cond:transport-mixing} makes \eqref{eq:local-poisson-integral} converge and gives
\[
 \|X_{s,Z}\|_\infty
 \leq\frac{C_{\mathrm{mix}}}{2}\|h_Z\|_\infty
           \int_0^\infty(1+t)^{-\nu}\,dt
 =\frac{C_{\mathrm{mix}}}{2(\nu-1)}\|h_Z\|_\infty.
\]
The integrand tends to zero, so applying the generator yields
\[
 \mathcal L_sX_{s,Z}
 =-\frac12\left[e^{t\mathcal L_s}
       \bigl(h_Z-\Tr(\sigma_sh_Z)I\bigr)\right]_{t=0}^{t=\infty}
 =\frac12\bigl(h_Z-\Tr(\sigma_sh_Z)I\bigr).
\]
Since $\mathcal L_{s,*}\sigma_s=0$, we have $e^{t\mathcal L_{s,*}}\sigma_s=\sigma_s$ for every $t\geq0$. By duality,
\[
 \begin{aligned}
 \Tr\!\left[\sigma_s e^{t\mathcal L_s}
       \bigl(h_Z-\Tr(\sigma_sh_Z)I\bigr)\right]
 &=\Tr\!\left[(e^{t\mathcal L_{s,*}}\sigma_s)
       \bigl(h_Z-\Tr(\sigma_sh_Z)I\bigr)\right]\\
 &=\Tr(\sigma_sh_Z)-\Tr(\sigma_sh_Z)\Tr(\sigma_s)=0.
 \end{aligned}
\]
Thus the integrand in \eqref{eq:local-poisson-integral} has zero expectation under $\sigma_s$ at every time, and taking the trace under the norm-convergent integral gives $\Tr(\sigma_sX_{s,Z})=0$.

We next approximate $X_{s,Z}$ inside $\Ball_{x_Z}(R)$. Let $\mathcal L_s^{[R]}$ be the local sampler on this ball, generated by $H_{\Ball_{x_Z}(R)}$ and extended trivially outside. Choose $\theta>0$ with $v\theta<\mu/2$, where $v,\mu$ are from Lemma~\ref{lem:physical-localization}. Set
\[
 X_{s,Z}^{[R]}=-\frac12\int_0^{\theta R}
       \bigl(e^{t\mathcal L_s^{[R]}}h_Z-\Tr(\sigma_sh_Z)I\bigr)\,dt,
 \qquad R\geq R_0.
\]
This matrix is supported on $\Ball_{x_Z}(R)$. Using Lemma~\ref{lem:physical-localization} for the first and \eqref{eq:transport-mixing} for the second gives
\begin{equation}\label{eq:poisson-truncation}
 \begin{aligned}
 \|X_{s,Z}-X_{s,Z}^{[R]}\|_\infty
 &\leq\frac{\|h_Z\|_\infty}{2}
   \left(C\int_0^{\theta R}e^{vt-\mu R}\,dt
       +C_{\mathrm{mix}}\int_{\theta R}^\infty(1+t)^{-\nu}\,dt\right)\\
 &\leq C\|h_Z\|_\infty
       \bigl(e^{-cR}+(1+R)^{1-\nu}\bigr).
 \end{aligned}
\end{equation}
Write the decomposition
\[
 X_{s,Z}=X_{s,Z}^{[R_0]}
       +\sum_{R>R_0}\bigl(X_{s,Z}^{[R]}-X_{s,Z}^{[R-1]}\bigr).
\]
The initial term has norm at most $C\|h_Z\|_\infty$. The two successive errors in \eqref{eq:poisson-truncation} bound each difference by
\[
 \|X_{s,Z}^{[R]}-X_{s,Z}^{[R-1]}\|_\infty
 \leq C\|h_Z\|_\infty
       \bigl(e^{-cR}+(1+R)^{1-\nu}\bigr).
\]
Each difference is supported on $\Ball_{x_Z}(R)$. Thus the definition \eqref{eq:polynorm} gives
\[
 \|X_{s,Z}\|_{P,x_Z}
 \leq C\|h_Z\|_\infty
 \left[1+\sum_{R>R_0}(1+R)^Pe^{-cR}
          +\sum_{R>R_0}(1+R)^{P+1-\nu}\right]
 \leq C_P\|h_Z\|_\infty,
\]
since $\nu>P+2$.
\end{proof}

\proofstep{Step 2: estimates in the quasi-local and local strength norms}

For the real-time representation of the potential filter, write
\[
 K_s=\tfrac12\log\sigma_s,\qquad
 \tau_t(A)=e^{itK_s}Ae^{-itK_s}=e^{-istH/2}Ae^{istH/2}.
\]
Thus bounds for Hamiltonian evolution also control $\tau_t$. Recall the quasi-local norm \eqref{eq:polynorm} and the local strength norm \eqref{eq:local-strength-norm}. The following lemma collects the basic estimates needed to control the quasi-local norms of products, commutators, and time evolution.

\begin{lemma}[Locality under evolution, products, and commutators]\label{lem:locality-calculus}
Under \eqref{eq:growth}, for $p\geq0$ and $\eta\geq0$,
\begin{align}
 \|e^{itH}Ae^{-itH}\|_{p,x}
 &\leq C_p(1+J|t|)^{p+d+2}\|A\|_{p+d+2,x},\label{eq:hamiltonian-locality}\\
 \|\tau_t(A)\|_{p,x}
 &\leq C_p(1+|t|)^{p+d+2}\|A\|_{p+d+2,x},\label{eq:real-polynomial}\\
 \|[C_x,D_y]\|_{p,x}
 &\leq C_{p,\eta}(1+\dist(x,y))^{-\eta}
       \|C_x\|_{p+\eta,x}\|D_y\|_{p+\eta,y}.\label{eq:overlap}
\end{align}
At the same center $x$, products and adjoints obey
\begin{equation}\label{eq:local-norm-product}
 \|AB\|_{p,x}\leq\|A\|_{p,x}\|B\|_{p,x},\qquad
 \|A^*\|_{p,x}=\|A\|_{p,x}.
\end{equation}
For a generator measured in the local strength norm \eqref{eq:local-strength-norm}, summation of the commutator estimate gives
\begin{equation}\label{eq:summed-local-commutator}
 \|[A,G]\|_{p,x}
 \leq C_p\|A\|_{p+d+2,x}\|G\|_{p+d+2,\mathrm{loc}}.
\end{equation}
For a family $D_\alpha$ centered at $x_\alpha$, summing the individual commutator norms gives
\begin{equation}\label{eq:family-local-commutator}
 \sum_\alpha\|[A,D_\alpha]\|_{p,x}
 \leq C_p\|A\|_{p+d+2,x}
       \sup_y\sum_{\alpha:x_\alpha=y}\|D_\alpha\|_{p+d+2,y}.
\end{equation}
For a scalar function $f$ with finite weighted integral below, real-time locality also gives
\begin{equation}\label{eq:local-filter-integral}
 \left\|\int_{\mathbb R}f(t)\tau_t(A)\,dt\right\|_{p,x}
 \leq C_p\|A\|_{p+d+2,x}
       \int_{\mathbb R}|f(t)|(1+|t|)^{p+d+2}\,dt.
\end{equation}
\end{lemma}
\begin{proof}
\emph{Evolution and filtering.}
For \eqref{eq:hamiltonian-locality}, write $A=\sum_rA_r$ with $\supp A_r\subseteq\Ball_x(r)$. Fix $r$ and set $B(t)=e^{itH}A_re^{-itH}$. We first approximate $B(t)$ by operators supported on larger balls. For a region $I$, let
\[
 \mathbb E_I(B)
 :=\frac{\Tr_{I^c}B}{\dim\mathcal H_{I^c}}\otimes I_{I^c},
 \qquad I^c=\Lambda\setminus I.
\]
Define
\[
 B_R(t):=\mathbb E_{\Ball_x(R)}B(t),\qquad R\geq r.
\]
Then $\supp B_R(t)\subseteq\Ball_x(R)$, and contractivity of $\mathbb E_{\Ball_x(R)}$ gives $\|B_R(t)\|\leq\|B(t)\|=\|A_r\|$. The Hamiltonian Lieb--Robinson bound~\cite[Theorem~3.1]{NSY2019}, applied to these local approximations as in~\cite[Eq.~(2)]{BHV2006}, gives constants $C,\mu,v_H>0$, independent of the volume, such that
\begin{equation}\label{eq:hamiltonian-local-approximation}
 \|B(t)-B_R(t)\|
 \leq C(1+r)^d e^{-\mu(R-r-v_HJ|t|)}\|A_r\|.
\end{equation}
Here, the factor $(1+r)^d$ bounds the size of the initial support. Thus the approximation error decays exponentially once $R$ exceeds $r+v_HJ|t|$.

To turn these approximations into a decomposition for the quasi-local norm, choose the integer radius
\[
 R_*=r+\lceil v_HJ|t|\rceil
\]
and set $D_n(t)=B_{R_*+n}(t)-B_{R_*+n-1}(t)$ for $n\geq1$. Telescoping gives
\[
 B(t)=B_{R_*}(t)+\sum_{n\geq1}D_n(t),
 \qquad \supp D_n(t)\subseteq\Ball_x(R_*+n).
\]
Applying \eqref{eq:hamiltonian-local-approximation} to both terms in each difference yields
\[
 \begin{aligned}
 \|D_n(t)\|
 &\leq\|B_{R_*+n}(t)-B(t)\|
       +\|B(t)-B_{R_*+n-1}(t)\|\\
 &\leq C(1+r)^d e^{-\mu(n-1)}\|A_r\|.
 \end{aligned}
\]
Using this decomposition in \eqref{eq:polynorm}, we obtain
\[
 \begin{aligned}
 \|B(t)\|_{p,x}
 &\leq(1+R_*)^p\|B_{R_*}(t)\|
       +\sum_{n\geq1}(1+R_*+n)^p\|D_n(t)\|\\
 &\leq(1+R_*)^p\|A_r\|
       \left[1+C(1+r)^d\sum_{n\geq1}(1+n)^p e^{-\mu(n-1)}\right]\\
 &\leq C_p(1+r)^d(1+R_*)^p\|A_r\|.
 \end{aligned}
\]
Here we used $1+R_*+n\leq(1+R_*)(1+n)$ and the convergence of the exponentially weighted sum. Since $1+R_*\leq C(1+r)(1+J|t|)$, this proves
\[
 \|e^{itH}A_re^{-itH}\|_{p,x}
 \leq C_p(1+r)^{p+d}(1+J|t|)^p\|A_r\|.
\]
Increasing both exponents to $p+d+2$ gives:
\[
 \|e^{itH}A_re^{-itH}\|_{p,x}
 \leq C_p(1+r)^{p+d+2}(1+J|t|)^{p+d+2}\|A_r\|.
\]
Summing over $r$ and taking the infimum proves \eqref{eq:hamiltonian-locality}. Substituting $-st/2$ for $t$ gives \eqref{eq:real-polynomial}.
Integrating \eqref{eq:real-polynomial} against $|f(t)|$ proves \eqref{eq:local-filter-integral}.

\smallskip\noindent
\emph{Products and adjoints.}
For products centered at $x$, the support radius is at most $\max(r,\ell)$, and
\[
 \|AB\|_{p,x}
 \leq\sum_{r,\ell}(1+\max(r,\ell))^p\|A_r\|\|B_\ell\|
 \leq\left(\sum_r(1+r)^p\|A_r\|\right)
       \left(\sum_\ell(1+\ell)^p\|B_\ell\|\right).
\]
Taking infima proves the product bound in \eqref{eq:local-norm-product}. Taking adjoints preserves the supports and norms of every term, proving the remaining identity~\eqref{eq:local-norm-product}.

\smallskip\noindent
\emph{Commutators.}
For \eqref{eq:overlap}, terms of radii $r,\ell$ commute unless $\dist(x,y)\leq r+\ell$; their union then lies in $\Ball_x(2r+2\ell)$. The inequality
\[
 (1+2r+2\ell)^p(1+\dist(x,y))^\eta
 \leq C_{p,\eta}(1+r)^{p+\eta}(1+\ell)^{p+\eta}
\]
proves the claim after summing their commutators.

For \eqref{eq:family-local-commutator}, apply \eqref{eq:overlap} with $\eta=d+2$ and group the terms by their centers:
\[
 \sum_\alpha\|[A,D_\alpha]\|_{p,x}
 \leq C_p\|A\|_{p+d+2,x}
       \left(\sup_y\sum_{\alpha:x_\alpha=y}\|D_\alpha\|_{p+d+2,y}\right)
       \sum_y(1+\dist(x,y))^{-d-2}.
\]
The distance sum is bounded by $\sum_{z\in\mathbb Z^d}(1+\|z\|_1)^{-d-2}<\infty$. Applying this estimate to a decomposition $G=\sum_yG_y$ and taking the infimum gives \eqref{eq:summed-local-commutator}.
\end{proof}

\proofstep{Step 3: locality of the Poisson potential}
We now prove the first bound needed for the generator:
\begin{equation}\label{eq:flow-potential-target}
 \sup_{0\leq s\leq\beta}\|\mathsf M_{f_s}\|_{\ell,\mathrm{loc}}\leq CJ.
\end{equation}
By \eqref{eq:continuation-sylvester} and linearity, $f_s=\sum_Z f_{s,Z}$, where each $f_{s,Z}$ is obtained by applying the $\sech$ filter to $X_{s,Z}$. The following estimate controls the locality lost under this filter.

\begin{lemma}[Locality of the Poisson potential]\label{lem:structure}
For $f_{s,Z}=\sech(s\ad_H/4)X_{s,Z}$ and every $p\geq0$,
\[
 \|f_{s,Z}\|_{p,x_Z}\leq C_p\|X_{s,Z}\|_{p+d+2,x_Z}.
\]
Here $x_Z\in Z$ is the center for the quasi-local norm, and $C_p$ is independent of $s$, $Z$, and the volume.
\end{lemma}
\begin{proof}
To bound the filtered potential, use the real-time representation
\begin{equation}\label{eq:sech}
 \sech(\ad_{K_s}/2)X=\int_\R\frac{\tau_t(X)}{\cosh(\pi t)}\,dt.
\end{equation}
Since $K_s=-sH/2$ up to a scalar multiple of $I$ and $\sech$ is even, the left-hand side of \eqref{eq:sech} is $\sech(s\ad_H/4)X$. Applying Lemma~\ref{lem:locality-calculus} \eqref{eq:local-filter-integral} gives
\[
 \|f_{s,Z}\|_{p,x_Z}
 \leq C_p\|X_{s,Z}\|_{p+d+2,x_Z}
       \int_{\mathbb R}\frac{(1+|t|)^{p+d+2}}{\cosh(\pi t)}\,dt
 \leq C_p'\|X_{s,Z}\|_{p+d+2,x_Z}.
\]
This proves the claimed locality bound.
\end{proof}

Apply Lemma~\ref{lem:structure} at order $\ell$. Since the quasi-local norms increase with their order and $\ell+d+2\leq P$, Proposition~\ref{prop:profile} \eqref{eq:local-poisson} gives
\[
 \|f_{s,Z}\|_{\ell,x_Z}
 \leq C\|X_{s,Z}\|_{\ell+d+2,x_Z}
 \leq C\|X_{s,Z}\|_{P,x_Z}
 \leq C\|h_Z\|_\infty.
\]
Symmetrized multiplication preserves support and satisfies \[
\|\mathsf M_f\|
=\|\mathcal M_f\|_{2\to2}
=\sup_{Y\ne0}
\frac{\left\|\frac12(fY+Yf)\right\|_2}{\|Y\|_2}\leq \|f\|_{\infty}\,.
\]
The same bound thus holds for $\|\mathsf M_{f_{s,Z}}\|_{\ell,x_Z}$. Grouping the terms by $x_Z$ therefore yields
\begin{equation}\label{eq:potential-local-strength}
 \begin{aligned}
 \sup_s\|\mathsf M_{f_s}\|_{\ell,\mathrm{loc}}\leq\sup_{s,x}\sum_{Z:x_Z=x}\|f_{s,Z}\|_{\ell,x}\leq C\sup_x\sum_{Z:x_Z=x}\|h_Z\|_\infty
 \leq CJ.
 \end{aligned}
\end{equation}
The last inequality uses $x_Z\in Z$ and the interaction-strength bound \eqref{eq:growth}. This proves~\eqref{eq:flow-potential-target}.

\proofstep{Step 4: locality of the weighted generator}
The second bound needed for the generator is
\[
 \sup_{s,x}\|\mathsf L_{\sigma_s,x}\|_{\ell,x}\leq C.
\]
We obtain it by controlling the gradients in the sampler factorization and then applying the product estimate from Step 2.

The factorization \eqref{eq:native-dirichlet-form} gives
\begin{equation}\label{eq:local-sampler-decomposition}
 \mathsf L_{\sigma_s,x}=-\sum_u\int_{\mathbb R^2}
 \widehat{\mathsf B}_{s,x,u}(v,z)^*\widehat{\mathsf B}_{s,x,u}(v,z)\,d\mu(v,z),\qquad
 \mathsf L_{\sigma_s}=\sum_x\mathsf L_{\sigma_s,x}.
\end{equation}
By the product bound Lemma~\ref{lem:locality-calculus} \eqref{eq:local-norm-product}, it therefore suffices to bound the integrated squared quasi-local norms of the gradients.
\begin{lemma}[Locality of the sampler gradients]\label{lem:native-gradient-locality}
For the Gaussian KMS Gibbs sampler \eqref{eq:CKG} of a finite-range Hamiltonian satisfying \eqref{eq:growth}, and every $p>0$,
\begin{equation}\label{eq:native-gradient-locality}
 \sup_{0\leq s\leq\beta,\,x}\sum_u\int_{\mathbb R^2}
 \|\widehat{\mathsf B}_{s,x,u}(v,z)\|_{p,x}^2\,d\mu(v,z)\leq C_B^2\,,
\end{equation}
where $C_B$ is independent of $s$, $x$, and the volume.
\end{lemma}
\begin{proof}
The Gibbs conjugation can be removed by shifting the Gaussian integration variable. For the rescaled jumps in \eqref{eq:native-rescaled-jump} and $a\in\mathbb R$,
\begin{equation}\label{eq:gaussian-modular-shift}
 \sigma_s^a\widehat J^s_{x,u}(z)\sigma_s^{-a}
 =e^{a^2-az}\widehat J^s_{x,u}(z-2a).
\end{equation}
In \eqref{eq:native-rescaled-jump}, Gibbs conjugation replaces the evolution time $sr$ by $s(r+ia)$. Substituting $t=r+ia$ gives
\[
 \begin{aligned}
 \sigma_s^a\widehat J^s_{x,u}(z)\sigma_s^{-a}
 =\frac{(2/\pi)^{1/4}}{\sqrt{2\pi}}e^{a^2-az}
   \int_{\mathbb R+ia}e^{-t^2-i(z-2a)t}
       e^{istH}S_{x,u}e^{-istH}\,dt=e^{a^2-az}\widehat J^s_{x,u}(z-2a).
 \end{aligned}
\]
The last equality shifts the contour to $\mathbb R$.

The Gibbs conjugation in \eqref{eq:native-gradient} is the case $a=-1/2$ of \eqref{eq:gaussian-modular-shift}, so
\begin{equation}\label{eq:native-gaussian-shift}
 \sigma_s^{-1/2}\widehat J^s_{x,u}(z)\sigma_s^{1/2}
 =e^{1/4+z/2}\widehat J^s_{x,u}(z+1).
\end{equation}
Thus the two multiplication matrices in \eqref{eq:native-gradient} are obtained from
\begin{equation}\label{eq:native-jump-evolution}
 e^{isvH}\widehat J^s_{x,u}(z)e^{-isvH}
 =\frac{(2/\pi)^{1/4}}{\sqrt{2\pi}}
   \int_{\mathbb R}e^{-r^2-izr}
       e^{is(v+r)H}S_{x,u}e^{-is(v+r)H}\,dr
\end{equation}
by using $z$ and $z+1$, respectively, and multiplying the second by $e^{1/4+z/2}$. For the integrand of \eqref{eq:native-jump-evolution}, \eqref{eq:hamiltonian-locality} applied to $S_{x,u}$ gives
\[
 \|e^{itH}S_{x,u}e^{-itH}\|_{p,x}
 \leq C_p(1+J|t|)^{p+d+2}.
\]
Applying this bound to both terms of~\eqref{eq:native-gradient} and using \eqref{eq:native-gaussian-shift}, we obtain
\begin{align*}
 \|\widehat{\mathsf{B}}_{s,x,u}(v,z)\|_{p,x}\leq C_p(1+e^{z/2})\int_{\mathbb R}
       e^{-r^2}(1+sJ|v+r|)^{p+d+2}\,dr\leq C_{p,\beta J}(1+|v|)^{p+d+2}(1+e^{z/2}).
\end{align*}
Squaring and integrating against $d\mu$ gives the bound:
\[
 \begin{aligned}
 \sup_{s,x}\sum_u\int_{\mathbb R^2}
       \|\widehat{\mathsf{B}}_{s,x,u}(v,z)\|_{p,x}^2\,d\mu(v,z)\leq C_{p,\beta J}\int_{\mathbb R^2}
       (1+|v|)^{2p+2d+4}(1+e^{z/2})^2\,d\mu(v,z)\leq C_B^2.
 \end{aligned}
\]
Here $d\mu(v,z)=e^{-(z+1)^2/2}\,dv\,dz/\cosh(2\pi v)$, so the polynomial factor in $v$ and the factor $(1+e^{z/2})^2$ are integrable. The constants are uniform in $s$, $x$, and the volume.
\end{proof}

We can now bound the local sampler terms. Using \eqref{eq:local-sampler-decomposition}, the product and adjoint bounds \eqref{eq:local-norm-product}, and Lemma~\ref{lem:native-gradient-locality} at order $\ell$, we obtain
\begin{equation}\label{eq:weighted-sampler-local-strength}
 \begin{aligned}
 \sup_{s,x}\|\mathsf L_{\sigma_s,x}\|_{\ell,x}
 \leq\sup_{s,x}\sum_u\int_{\mathbb R^2}
       \|\widehat{\mathsf B}_{s,x,u}(v,z)\|_{\ell,x}^2\,d\mu(v,z)\leq C.
 \end{aligned}
\end{equation}
This proves the required sampler bound.

\proofstep{Step 5: the commutator bound}
It remains to combine the potential bound \eqref{eq:potential-local-strength} and the sampler bound \eqref{eq:weighted-sampler-local-strength}. Decompose the generator according to the sampler sites:
\[
 \mathsf G_s=\sum_x\mathsf G_{s,x},\qquad
 \mathsf G_{s,x}:=[\mathsf M_{f_s},\mathsf L_{\sigma_s,x}].
\]
Apply \eqref{eq:summed-local-commutator} with the operator centered at $x$ equal to $\mathsf L_{\sigma_s,x}$ and the other factor equal to $\mathsf M_{f_s}$. Since $\ell=p+d+2$, this gives
\begin{equation}\label{eq:flow-locality-reduction}
 \begin{aligned}
 \|\mathsf G_s\|_{p,\mathrm{loc}}\leq\sup_x\|\mathsf G_{s,x}\|_{p,x}\leq C_p\|\mathsf M_{f_s}\|_{\ell,\mathrm{loc}}
       \sup_x\|\mathsf L_{\sigma_s,x}\|_{\ell,x}\leq CJ.
 \end{aligned}
\end{equation}
Both input bounds are uniform in $s$ and the volume. Taking the supremum over $0\leq s\leq\beta$ proves \eqref{eq:flow-locality-target} with $C_G\leq CJ$.
\end{proof}

\section{Technical estimates for energy and entropy variation}\label{app:kernel}
We prove Lemma~\ref{lem:comparison-families}, the entropy representation \eqref{eq:entropy-integral-representation}, and the entropy bound \eqref{eq:entropy-derivative-bound} in Lemma~\ref{lem:derivative-reduction}; the energy bound \eqref{eq:energy-derivative-bound} and the cancellations \eqref{eq:derivative-gibbs-cancellation} were derived in Section~\ref{sec:derivative-comparison}. Appendix~\ref{sec:energy-reduction} bounds the quasi-local norms of $\mathsf F_{s,x,u}(v,z)$ from \eqref{eq:moving-frame-derivative}. Appendix~\ref{sec:entropy-reduction} decomposes the transport generator into pairs of matrix products, computes the entropy contribution of each pair, constructs the filtered commutators $\mathsf B_{s,P_\gamma(\omega)}$ and $\mathsf B_{s,Q_\gamma(\omega)}$, and sums their quasi-local norms. Together with Condition~\ref{cond:local-action}, verified in Appendix~\ref{sec:local-gap-criterion-proof}, these estimates give \eqref{eq:variation-targets} as in the proof of Theorem~\ref{prop:general-criterion}.

\subsection{Energy variation}\label{sec:energy-reduction}
We prove \eqref{eq:energy-family-locality} in Lemma~\ref{lem:comparison-families}.


\proofstep{Step 1: differentiate the Gaussian gradients}
The measure $d\mu$ in \eqref{eq:native-gradient} is independent of $s$. We therefore differentiate the rescaled gradients at fixed $(v,z)$.
\begin{lemma}[Locality of the temperature derivatives]\label{lem:native-gradient-derivatives}
For the Gaussian KMS Gibbs sampler \eqref{eq:CKG} of a finite-range Hamiltonian satisfying \eqref{eq:growth}, and every $p>0$,
\begin{equation}\label{eq:dirichlet-locality-target}
 \sup_{0\leq s\leq\beta,\,x}\sum_u\int_{\mathbb R^2}
 \|\partial_s\widehat{\mathsf B}_{s,x,u}(v,z)\|_{p,x}^2\,d\mu(v,z)\leq c_p^2J^2=:C_{\dot B}^2,
\end{equation}
where $C_{\dot B}$ is independent of $s$, $x$, and the volume.
\end{lemma}
\begin{proof}
We follow the proof of Lemma~\ref{lem:native-gradient-locality}. Differentiate \eqref{eq:native-jump-evolution}:
\[
 \partial_s\!\left(e^{isvH}\widehat J^s_{x,u}(z)e^{-isvH}\right)
 =\frac{i(2/\pi)^{1/4}}{\sqrt{2\pi}}
   \int_{\mathbb R}(v+r)e^{-r^2-izr}
       e^{is(v+r)H}[H,S_{x,u}]e^{-is(v+r)H}\,dr.
\]
Since $H$ has finite range,
\[
 \supp[H,S_{x,u}]\subseteq\Ball_x(R_0),\qquad
 \|[H,S_{x,u}]\|\leq2J.
\]
Applying \eqref{eq:hamiltonian-locality} to $[H,S_{x,u}]$ therefore gives
\[
 \|e^{itH}[H,S_{x,u}]e^{-itH}\|_{p,x}
 \leq C_pJ(1+J|t|)^{p+d+2}.
\]
Since the factor $e^{1/4+z/2}$ in \eqref{eq:native-gaussian-shift} is independent of $s$, applying this bound inside the differentiated integral gives
\begin{align*}
 \|\partial_s\widehat{\mathsf{B}}_{s,x,u}(v,z)\|_{p,x}
 &\leq C_pJ(1+e^{z/2})\int_{\mathbb R}
       |v+r|e^{-r^2}(1+sJ|v+r|)^{p+d+2}\,dr\\
 &\leq C_{p,\beta J}J(1+|v|)^{p+d+3}(1+e^{z/2}).
\end{align*}
Squaring and integrating against $d\mu$ gives
\[
 \begin{aligned}
 \sup_{s,x}\sum_u\int_{\mathbb R^2}
       \|\partial_s\widehat{\mathsf{B}}_{s,x,u}(v,z)\|_{p,x}^2\,d\mu(v,z)
 &\leq C_{p,\beta J}J^2\int_{\mathbb R^2}
       (1+|v|)^{2p+2d+6}(1+e^{z/2})^2\,d\mu(v,z)\\
       &\leq c_p^2J^2=C_{\dot B}^2.
 \end{aligned}
\]
\end{proof}

\proofstep{Step 2: bound the quasi-local norms of $\mathsf F_{s,x,u}$}
\begin{proof}[Proof of \eqref{eq:energy-family-locality} in Lemma~\ref{lem:comparison-families}]
To prove \eqref{eq:energy-family-locality}, set $p_1=p_0+d+2$ and estimate the two terms in \eqref{eq:moving-frame-derivative}. Since Condition~\ref{cond:transport-mixing} holds with
$\nu>p_0+9d+12>p_1+6d+10$, Proposition~\ref{prop:flow-locality} bounds $\mathsf G_s$ at order $p_1$ by $C_G$. Applying \eqref{eq:summed-local-commutator} to the commutator term gives
\[
 \|\mathsf F_{s,x,u}(v,z)\|_{p_0,x}
 \leq\|\partial_s\widehat{\mathsf B}_{s,x,u}(v,z)\|_{p_1,x}
       +cC_G\|\widehat{\mathsf B}_{s,x,u}(v,z)\|_{p_1,x}.
\]
Applying Lemmas~\ref{lem:native-gradient-locality} and~\ref{lem:native-gradient-derivatives} gives
\begin{equation}\label{eq:differentiated-gradient-locality}
 \sup_{s,x}\sum_u\int_{\mathbb R^2}\|\mathsf F_{s,x,u}(v,z)\|_{p_0,x}^2\,d\mu(v,z)
 \leq c(C_{\dot B}^2+C_G^2C_B^2)=:M_E^2.
\end{equation}
This proves \eqref{eq:energy-family-locality}.
\end{proof}

\subsection{Entropy derivative and gradient locality}\label{sec:entropy-reduction}
We prove \eqref{eq:entropy-derivative-bound} in Lemma~\ref{lem:derivative-reduction} and \eqref{eq:entropy-family-locality} in Lemma~\ref{lem:comparison-families}. The first task is to construct the representation \eqref{eq:entropy-integral-representation} of the derivative, with $\mathcal B_{s,V}$ and $\mathcal K_{h,Y}$ as in Section~\ref{sec:derivative-comparison}.

\proofstep{Step 1: write $\mathcal G_s(Y)$ as a sum of matrix products}
The entropy variation formula in Section~\ref{sec:derivative-comparison} requires a decomposition of $\mathsf{G}_s$ into multiplication and commutator terms. We obtain it from the Gaussian KMS Gibbs sampler and the Poisson potentials constructed in Section~\ref{sec:poisson-construction}.

\begin{lemma}[Decomposition of the transport]\label{lem:transport-decomposition}
Let $H=\sum_Zh_Z$ with $h_Z=h_Z^*$, and let $\mathcal L_s$ be the trace adjoint of the Gaussian KMS Gibbs sampler \eqref{eq:CKG} for $\sigma_s=e^{-sH}/\Tr(e^{-sH})$, with $s>0$. Write $k=(x,u,r)$, where $r\in\{1,2\}$ labels the real and imaginary parts of a weighted jump, and set $x_k=x$. The weighted generator $\mathcal{L}_{\sigma_s}$ has the continuous decomposition

\begin{equation}\label{eq:aux-parent}
\mathcal{L}_{\sigma_s}=-\sum_k\int_{\mathbb R}\mathrm{ad}_{V_k(\omega)}^2\,d\omega-\mathcal{M}_F,
 \qquad \mathcal M_F(Y)=\tfrac12\{F,Y\},
\end{equation}
where the matrices are defined directly from the filtered jumps by
\begin{equation}\label{eq:weighted-jump-coefficients}
 \begin{aligned}
 V_{x,u,1}(\omega)
 =\sqrt{\frac{\gamma_s(\omega)}2}\,
   \operatorname{Re}\!\left(\sigma_s^{1/4}J_{x,u}^{s}(\omega)^*\sigma_s^{-1/4}\right),\quad V_{x,u,2}(\omega)=\sqrt{\frac{\gamma_s(\omega)}2}\,
   \operatorname{Im}\!\left(\sigma_s^{1/4}J_{x,u}^{s}(\omega)^*\sigma_s^{-1/4}\right).
 \end{aligned}
\end{equation}
Here $\operatorname{Re}A=(A+A^*)/2$, $\operatorname{Im}A=(A-A^*)/(2i)$, and
\begin{equation}\label{eq:transport-multiplication-coefficient}
 F=2\operatorname{Re}\!\left[
       \sigma_s^{1/4}\left(\tfrac12Q_s-iH_s^{\mathrm{coh}}\right)
       \sigma_s^{-1/4}\right]
       -2\sum_k\int_{\mathbb R}V_k(\omega)^2\,d\omega.
\end{equation}
The matrices $Q_s$ and $H_s^{\mathrm{coh}}$ are those in \eqref{eq:CKG-correction}.

Moreover, the transport $\mathcal{G}_s=[\mathcal{M}_{f_s},\mathcal{L}_{\sigma_s}]$ acts by
\begin{equation}\label{eq:G-multiplication}
 \mathcal G_s(Y)=[W_s,Y]+\sum_{k,Z}\int_{\mathbb R}
  \bigl(V_k(\omega)YC_{k,Z}(\omega)
            -C_{k,Z}(\omega)YV_k(\omega)\bigr)\,d\omega,
\end{equation}
where
\[
 C_{k,Z}(\omega)=[V_k(\omega),f_{s,Z}],\qquad
 W_s=\tfrac14\left[F+2\sum_k\int_{\mathbb R}V_k(\omega)^2\,d\omega,f_s\right].
\]

The potential is $f_s=\sum_Zf_{s,Z}$, with $f_{s,Z}=\sech(s\ad_H/4)X_{s,Z}$ and $X_{s,Z}$ from \eqref{eq:local-poisson-integral}, as in Appendix~\ref{app:transport-coefficients}, Step~3. These matrices satisfy $F=F^*$, $W_s^*=-W_s$, and $C_{k,Z}(\omega)^*=-C_{k,Z}(\omega)$.
\end{lemma}
\begin{proof}[Proof of Lemma~\ref{lem:transport-decomposition}]


A direct calculation using \eqref{eq:CKG} gives
\begin{equation}\label{eq:decomposition-cited-form}
 \begin{aligned}
 \Gamma_{\sigma_s}^{1/2}\mathcal L_s\Gamma_{\sigma_s}^{-1/2}(Y)
 &=\sum_{x,u}\int_{\mathbb R}\!d\omega\,\gamma_s(\omega)
 (\sigma_s^{1/4}J_{x,u}^{s}(\omega)^*\sigma_s^{-1/4})Y
 (\sigma_s^{-1/4}J_{x,u}^{s}(\omega)\sigma_s^{1/4})\\
 &\quad-\sigma_s^{1/4}\left(\tfrac12Q_s-iH_s^{\mathrm{coh}}\right)\sigma_s^{-1/4}Y
 -Y\sigma_s^{-1/4}\left(\tfrac12Q_s+iH_s^{\mathrm{coh}}\right)\sigma_s^{1/4}.
 \end{aligned}
\end{equation}
KMS detailed balance gives $\Gamma_{\sigma_s}^{1/2}\mathcal L_s\Gamma_{\sigma_s}^{-1/2}=\Gamma_{\sigma_s}^{-1/2}\mathcal L_{s,*}\Gamma_{\sigma_s}^{1/2}$. This implies
\begin{equation}\label{eq:decomposition-symmetrized-form}
 \begin{aligned}
 \Gamma_{\sigma_s}^{1/2}\mathcal L_s\Gamma_{\sigma_s}^{-1/2}(Y)
 &=\tfrac12\left(\Gamma_{\sigma_s}^{1/2}\mathcal L_s\Gamma_{\sigma_s}^{-1/2}
 +\Gamma_{\sigma_s}^{-1/2}\mathcal L_{s,*}\Gamma_{\sigma_s}^{1/2}\right)(Y)\\
 &=\tfrac12\sum_{x,u}\int_{\mathbb R}\!d\omega\,\gamma_s(\omega)\Bigl[
 (\sigma_s^{1/4}J_{x,u}^{s}(\omega)^*\sigma_s^{-1/4})Y
 (\sigma_s^{-1/4}J_{x,u}^{s}(\omega)\sigma_s^{1/4})\\
 &\qquad+(\sigma_s^{-1/4}J_{x,u}^{s}(\omega)\sigma_s^{1/4})Y
 (\sigma_s^{1/4}J_{x,u}^{s}(\omega)^*\sigma_s^{-1/4})\Bigr]\\
 &\quad-\left\{\operatorname{Re}\!\left[\sigma_s^{1/4}\left(\tfrac12Q_s-iH_s^{\mathrm{coh}}\right)\sigma_s^{-1/4}\right],Y\right\}.
 \end{aligned}
\end{equation}
We reformulate \eqref{eq:decomposition-symmetrized-form} using the notation in \eqref{eq:weighted-jump-coefficients} and \eqref{eq:transport-multiplication-coefficient}:
\begin{equation}\label{eq:decomposition-weighted-final}
 \begin{aligned}
 \Gamma_{\sigma_s}^{1/2}\mathcal L_s\Gamma_{\sigma_s}^{-1/2}(Y)
 &=2\sum_k\int_{\mathbb R}V_k(\omega)YV_k(\omega)\,d\omega
 -\tfrac12\left\{F+2\sum_k\int_{\mathbb R}V_k(\omega)^2\,d\omega,Y\right\}\\
 &=-\sum_k\int_{\mathbb R}[V_k(\omega),[V_k(\omega),Y]]\,d\omega-\tfrac12\{F,Y\}.
 \end{aligned}
\end{equation}
This proves \eqref{eq:aux-parent}.

Define $A=F+2\sum_k\int V_k(\omega)^2\,d\omega$ in this calculation. Then,
\[
 \mathcal L_{\sigma_s}(Y)=2\sum_k\int V_k(\omega)YV_k(\omega)\,d\omega-\mathcal M_A(Y).
\]
The commutator with $\mathcal M_{f_s}$ follows from two algebraic identities:
\[
 \begin{aligned}
 [\mathcal M_f,\,2V(\,\cdot\,)V](Y)
 &=VY[V,f]-[V,f]YV,\\
 [\mathcal M_f,-\mathcal M_A](Y)&=\tfrac14[[A,f],Y].
 \end{aligned}
\]
Summing the first identity and using $f_s=\sum_Zf_{s,Z}$ gives the pair terms in \eqref{eq:G-multiplication}; the second gives $[W_s,Y]$ with the stated $W_s$.
\end{proof}

\proofstep{Step 2: compute the entropy contribution of one pair}
We now compute the derivative produced by the two types of terms in \eqref{eq:G-multiplication}.

Recall that $\|Y_s\|_2=1$, so $Y_s^2$ is a state and
\[
 D(Y_s^2\Vert\sigma_s)=\Tr[Y_s^2(\log Y_s^2-\log\sigma_s)].
\]
The unitary preserves $\|Y_s\|_2$, so $\Tr[(\partial_sY_s)Y_s]=0$. Differentiating the two trace terms gives
\[
 \begin{aligned}
 \partial_s\Tr[Y_s^2\log Y_s^2]
 &=2\Tr[(\partial_sY_s)Y_s\log Y_s^2],\\
 \partial_s\Tr[Y_s^2\log\sigma_s]
 &=\Tr[(\partial_sY_s)\{Y_s,\log\sigma_s\}]
   +\Tr[Y_s^2\partial_s\log\sigma_s].
 \end{aligned}
\]
Their difference gives \eqref{eq:entropy-simultaneous-derivative}. Equation~\eqref{eq:entropy-simultaneous-derivative} also holds at zero eigenvalues: the function $x\mapsto x^2\log x^2$, with value zero at $x=0$, is continuously differentiable and has derivative zero there. Thus $Y_s\log Y_s^2$ is interpreted by continuous functional calculus.
Both matrices move by the transport:
\begin{equation}\label{eq:simultaneous-matrix-motion}
 \partial_sY_s=\mathcal G_s(Y_s),\qquad
 \partial_s\sqrt{\sigma_s}=\mathcal G_s(\sqrt{\sigma_s}),
\end{equation}
with $\mathcal G_s$ as in \eqref{eq:G-multiplication}. The entropy derivative contains $\partial_s\log\sigma_s$, whereas \eqref{eq:simultaneous-matrix-motion} gives $\partial_s\sqrt{\sigma_s}$. The resolvent formula connects them:
\begin{equation}\label{eq:entropy-log-derivative}
 \begin{aligned}
 \log\sigma_s
 &=2\int_0^\infty\left(\frac{I}{1+t}-(\sqrt{\sigma_s}+tI)^{-1}\right)dt,\\
 \partial_s(\sqrt{\sigma_s}+tI)^{-1}
 &=-(\sqrt{\sigma_s}+tI)^{-1}
       (\partial_s\sqrt{\sigma_s})(\sqrt{\sigma_s}+tI)^{-1},\\
 \partial_s\log\sigma_s
 &=-2\int_0^\infty\partial_s(\sqrt{\sigma_s}+tI)^{-1}\,dt =2\int_0^\infty(\sqrt{\sigma_s}+tI)^{-1}
       (\partial_s\sqrt{\sigma_s})
       (\sqrt{\sigma_s}+tI)^{-1}\,dt.
 \end{aligned}
\end{equation}
The commutator terms $[W_s,\cdot]$ in \eqref{eq:G-multiplication} contribute zero. Indeed, simultaneous unitary conjugation preserves the relative entropy:
\[
 D(UY^2U^*\Vert U\sigma U^*)=D(Y^2\Vert\sigma)\qquad(U^*U=I).
\]
Since $W_s^*=-W_s$, taking $U=e^{\theta W_s}$ and differentiating at $\theta=0$ cancels all terms arising from $[W_s,Y_s]$ and $[W_s,\sqrt{\sigma_s}]$. The remaining derivative is therefore
\begin{equation}\label{eq:entropy-pair-derivative}
 \begin{aligned}
 \frac d{ds}D(Y_s^2\Vert\sigma_s)
 &=\sum_{k,Z}\int_{\mathbb R}\!d\omega\,
       \delta_{V_k(\omega),C_{k,Z}(\omega)}(Y_s),\\
 \delta_{V,C}(Y)
 &=\underbrace{2\Tr[(VYC-CYV)Y\log Y^2]}_{\mathrm I}\\
 &\quad+\underbrace{\bigl(-\Tr[(VYC-CYV)\{Y,\log\sigma_s\}]\bigr)}_{\mathrm{II}}\\
 &\quad+\underbrace{\left(-2\int_0^\infty\Tr\!\left[
       Y^2(\sqrt{\sigma_s}+tI)^{-1}
       (V\sqrt{\sigma_s}C-C\sqrt{\sigma_s}V)
       (\sqrt{\sigma_s}+tI)^{-1}\right]dt\right)}_{\mathrm{III}}.
 \end{aligned}
\end{equation}

Here $\delta_{V,C}(Y)$ is defined by this formula for every Hermitian $Y$; it depends on $s$ through $\sigma_s$. Its three terms combine into weighted commutators. The next step gives the representation and its local strength.

\proofstep{Step 3: construct the filtered commutators}
For a single $\delta_{V,C}$, we need a representation by weighted commutators and a bound on the sum of their squared quasi-local norms, as summarized in the following lemma:

\begin{lemma}[Filtered representation and locality of one pair]\label{lem:entropy-filtered-representation}
Fix $p_0>0$, $p_1=p_0+d+2$, and a center $x$. For $V=V^*$ and $C^*=-C$, there are smooth scalar filters $p_\alpha,q_\alpha$, independent of $Y$, such that
\begin{equation}\label{eq:entropy-filtered-matrices}
 P_\alpha=p_\alpha(-s\ad_H/2)V,\qquad
 Q_\alpha=q_\alpha(-s\ad_H/2)C
\end{equation}
satisfy the following properties.
\begin{enumerate}
\item For every Hermitian $Y$, the operators in \eqref{eq:entropy-kernel-contraction} give
\begin{equation}\label{eq:entropy-filtered-identity}
 \delta_{V,C}(Y)
 =\sum_\alpha\Tr\!\left[\mathcal B_{s,Q_\alpha}(Y)
       \mathcal K_{h_\alpha,Y}\bigl(\mathcal B_{s,P_\alpha}(Y)\bigr)\right],
 \qquad \|h_\alpha\|_\infty\leq1.
\end{equation}
\item The filtered gradients obey
\begin{equation}\label{eq:entropy-pair-local-strength}
 \sum_\alpha\bigl(\|\mathsf{B}_{s,P_\alpha}\|_{p_0,x}^2
                   +\|\mathsf{B}_{s,Q_\alpha}\|_{p_0,x}^2\bigr)
 \leq2C_{\log}\|\mathsf{B}_{s,V}\|_{p_1+2d,x}\|\mathsf{B}_{s,C}\|_{p_1+2d,x}.
\end{equation}
\end{enumerate}
Here $\sum_\alpha$ denotes integration over $(t,u)\in\mathbb R^2$ with the positive measure $d\nu_N(t,u)$ constructed in the proof; it is independent of $Y$. The constant depends only on $p_0$, an upper bound on $\beta J$, the interaction range, and the geometry.
\end{lemma}

We first derive and factor the scalar kernel of $\delta_{V,C}(Y)$, treating the three terms $\mathrm I$, $\mathrm{II}$, $\mathrm{III}$ of \eqref{eq:entropy-pair-derivative} separately.
For term $\mathrm I$, take an orthonormal decomposition $Y=\sum_i x_i|e_i\rangle\langle e_i|$ and write $V_{ij}=\langle e_i|V|e_j\rangle$, $C_{ij}=\langle e_i|C|e_j\rangle$ in this basis. Expanding the trace gives
\[
 \begin{aligned}
 \mathrm I
 =2\sum_{i,j}\bigl[x_j(x_i\log x_i^2)-x_i(x_j\log x_j^2)\bigr]V_{ij}C_{ji}=4\sum_{i,j}x_ix_j\log|x_i/x_j|\,V_{ij}C_{ji}.\qquad(0\log0=0)
 \end{aligned}
\]
For term $\mathrm{II}$, cyclicity of the trace gives
\begin{equation}\label{eq:entropy-second-term-commutators}
 \begin{aligned}
 \mathrm{II}
 &=-\Tr([\log\sigma_s,V]YCY)
   +\Tr([\log\sigma_s,C]YVY).
 \end{aligned}
\end{equation}
Since $\log\sigma_s=-sH-(\log Z_s)I$, both can be evaluated by separating $V$ and $C$ according to energy differences of $H$. Writing $H=\sum_EE\Pi_E$, set
\[
 V=\sum_\nu V_\nu,\quad V_\nu=\sum_{E-E'=\nu}\Pi_EV\Pi_{E'},
 \qquad
 C=\sum_\eta C_\eta,\quad C_\eta=\sum_{E-E'=\eta}\Pi_EC\Pi_{E'}.
\]
The commutators and Gibbs factors now take the explicit form
\begin{equation}\label{eq:entropy-bohr-identities}
 \begin{aligned}
 [\log\sigma_s,V_\nu]&=-s\nu V_\nu,&
 V_\nu\sqrt{\sigma_s}&=e^{s\nu/2}\sqrt{\sigma_s}V_\nu,\\
 [\log\sigma_s,C_\eta]&=-s\eta C_\eta,&
 C_\eta\sqrt{\sigma_s}&=e^{s\eta/2}\sqrt{\sigma_s}C_\eta.
 \end{aligned}
\end{equation}
Substituting the commutator identities in \eqref{eq:entropy-bohr-identities} into \eqref{eq:entropy-second-term-commutators} gives
\[
 \mathrm{II}
 =s\sum_{\nu,\eta}\bigl(\nu\Tr[V_\nu YC_\eta Y]
                         -\eta\Tr[C_\eta YV_\nu Y]\bigr)
 =s\sum_{\nu,\eta,i,j}(\nu-\eta)x_ix_j(V_\nu)_{ij}(C_\eta)_{ji}.
\]
For term $\mathrm{III}$, the identities in \eqref{eq:entropy-bohr-identities} give
\begin{equation}\label{eq:entropy-pair-gibbs-product}
 V\sqrt{\sigma_s}C-C\sqrt{\sigma_s}V
 =\sqrt{\sigma_s}\sum_{\nu,\eta}
       (e^{s\nu/2}V_\nu C_\eta-e^{s\eta/2}C_\eta V_\nu).
\end{equation}
Substituting \eqref{eq:entropy-pair-gibbs-product} into the integral in term $\mathrm{III}$ gives
\begin{equation}\label{eq:entropy-pair-resolvent}
 \begin{aligned}
 &\int_0^\infty(\sqrt{\sigma_s}+tI)^{-1}
       (V\sqrt{\sigma_s}C-C\sqrt{\sigma_s}V)
       (\sqrt{\sigma_s}+tI)^{-1}\,dt\\
 &=\sum_{\nu,\eta}\int_0^\infty
       (\sqrt{\sigma_s}+tI)^{-1}\sqrt{\sigma_s}
       (e^{s\nu/2}V_\nu C_\eta-e^{s\eta/2}C_\eta V_\nu)
       (\sqrt{\sigma_s}+tI)^{-1}\,dt\\
 &\overset{\mathrm{(i)}}{=}\sum_{\nu,\eta}
       \left[\int_0^\infty\sqrt{\sigma_s}
       (\sqrt{\sigma_s}+tI)^{-1}
       (e^{s(\nu+\eta)/2}\sqrt{\sigma_s}+tI)^{-1}\,dt\right]
       (e^{s\nu/2}V_\nu C_\eta-e^{s\eta/2}C_\eta V_\nu)\\
 &\overset{\mathrm{(ii)}}{=}\sum_{\nu,\eta}r\!\left(-\frac{s(\nu+\eta)}2\right)
       (e^{s\nu/2}V_\nu C_\eta-e^{s\eta/2}C_\eta V_\nu),
 \end{aligned}
\end{equation}
where $r(u)=u/(1-e^{-u})$, with $r(0)=1$. Equality \textup{(i)} follows from the Gibbs identities in \eqref{eq:entropy-bohr-identities}:
\[
 \begin{aligned}
 V_\nu C_\eta(\sqrt{\sigma_s}+tI)^{-1}
 &=(e^{s(\nu+\eta)/2}\sqrt{\sigma_s}+tI)^{-1}V_\nu C_\eta,\\
 C_\eta V_\nu(\sqrt{\sigma_s}+tI)^{-1}
 &=(e^{s(\nu+\eta)/2}\sqrt{\sigma_s}+tI)^{-1}C_\eta V_\nu.
 \end{aligned}
\]
Equality \textup{(ii)} uses the scalar identity
\begin{equation}\label{eq:entropy-scalar-resolvent}
 \int_0^\infty\frac{\lambda\,dt}{(\lambda+t)(e^{-u}\lambda+t)}
 =r(u),\qquad \lambda>0.
\end{equation}
Its value is independent of $\lambda$, so applying it with $u=-s(\nu+\eta)/2$ to the eigenvalues of $\sqrt{\sigma_s}$ proves equality \textup{(ii)}.
Multiplying \eqref{eq:entropy-pair-resolvent} by $-2Y^2$ and taking the trace evaluates term $\mathrm{III}$:
\[
 \mathrm{III}
 =-2\sum_{\nu,\eta,i,j}r\!\left(-\frac{s(\nu+\eta)}2\right)
       (e^{s\nu/2}x_i^2-e^{s\eta/2}x_j^2)(V_\nu)_{ij}(C_\eta)_{ji}.
\]
Adding terms $\mathrm I$, $\mathrm{II}$, and $\mathrm{III}$ gives
\begin{equation}\label{eq:entropy-bohr-kernel}
 \begin{aligned}
 \delta_{V,C}(Y)
 &=\sum_{\nu,\eta,i,j}
 \Bigl[4x_ix_j\log|x_i/x_j|+s(\nu-\eta)x_ix_j\\
 &\hspace{16mm}-2r\!\left(-\frac{s(\nu+\eta)}2\right)
       (e^{s\nu/2}x_i^2-e^{s\eta/2}x_j^2)\Bigr]
       (V_\nu)_{ij}(C_\eta)_{ji}\\
 &=\sum_{\nu,\eta,i,j}k_{-s\nu/2,-s\eta/2}(x_i,x_j)(V_\nu)_{ij}(C_\eta)_{ji},
 \end{aligned}
\end{equation}
where the scalar kernel is
\begin{equation}\label{eq:entropy-scalar-kernel}
 k_{a,c}(x,y):=4xy\log|x/y|-2(a-c)xy-2\frac{a+c}{1-e^{-(a+c)}}(e^{-a}x^2-e^{-c}y^2).
\end{equation}

Direct substitution into \eqref{eq:entropy-scalar-kernel} gives
\[
 k_{a,c}(e^a y,y)=0,\qquad k_{a,c}(e^{-c}y,y)=0.
\]
Lemma~\ref{lem:entropy-scalar-kernel} below writes $k_{a,c}(x,y)$ as a bounded kernel $H_{a,c}(x,y)$ times the two linear factors vanishing at these zeros. For $a=-s\nu/2$, $c=-s\eta/2$, $x=x_i$, and $y=x_j$, these two factors are exactly the weighted commutator factors:
\[
 \begin{aligned}
 (e^{-s\nu/4}x_j-e^{s\nu/4}x_i)(V_\nu)_{ij}
 &=(\mathcal B_{s,V_\nu}(Y))_{ij},\\
 (e^{-s\eta/4}x_i-e^{s\eta/4}x_j)(C_\eta)_{ji}
 &=(\mathcal B_{s,C_\eta}(Y))_{ji},
 \end{aligned}
\]
with $\mathcal B_{s,V}$ as in Section~\ref{sec:derivative-comparison}. Consequently,
\begin{equation}\label{eq:entropy-bohr-products}
 \delta_{V,C}(Y)
 =\sum_{\nu,\eta,i,j}H_{-s\nu/2,-s\eta/2}(x_i,x_j)
       (\mathcal B_{s,V_\nu}(Y))_{ij}(\mathcal B_{s,C_\eta}(Y))_{ji}.
\end{equation}

We next give the bound on the scalar kernel $H_{a,c}(x,y)$ and its derivatives.
\begin{lemma}[Factorization of the logarithmic kernel]\label{lem:entropy-scalar-kernel}
For the scalar kernel \eqref{eq:entropy-scalar-kernel} and $x,y\in\mathbb R$, there is a function $H_{a,c}(x,y)$, smooth in $a,c\in\mathbb R$, such that
\[
 k_{a,c}(x,y)=H_{a,c}(x,y)
 (e^{a/2}y-e^{-a/2}x)(e^{c/2}x-e^{-c/2}y),
\]
and, for all nonnegative integers $m,n$,
\[
 \sup_{a,c,x,y\in\mathbb R}
 |\partial_a^m\partial_c^nH_{a,c}(x,y)|\leq C_{m,n}.
\]
The constants are independent of $x,y,a,c$.
\end{lemma}

\begin{proof}[Proof of Lemma~\ref{lem:entropy-scalar-kernel}]
For $xy\ne0$, the factorization and the bounds follow from the explicit formula
\begin{equation}\label{eq:kernel-factorization}
 H_{a,c}(x,y)
 =2\vartheta(a+c)
   \int_0^1 F_{\operatorname{sign}(xy)}'
       (\log|x/y|-(1-t)a+tc)\,dt,
\end{equation}
with
\[
 \vartheta(u)=\frac{u}{2\sinh(u/2)},\qquad
 F_+(u)=u\coth(u/2),\qquad F_-(u)=u\tanh(u/2).
\]
The values at removable singularities are understood by continuity.
To check the factorization, put
$z=\log|x/y|-(a-c)/2$ and $t_0=(a+c)/2$.
Dividing $k_{a,c}$ by the two unbalanced factors
$(y-e^{-a}x)(x-e^{-c}y)$ gives, respectively,
\[
 2r(2t_0)
 \frac{\sinh z\mp(z/t_0)\sinh t_0}
      {\cosh z\mp\cosh t_0}
 =2r(2t_0)
   \frac{F_\pm(z+t_0)-F_\pm(z-t_0)}{2t_0}.
\]
Multiplying by $e^{-t_0}$ gives
\eqref{eq:kernel-factorization}.
All derivatives of $F_\pm$ of positive order are bounded, as are
all derivatives of $\vartheta$. Differentiating under the integral
in \eqref{eq:kernel-factorization} therefore gives the derivative bounds of the lemma for $xy\ne0$.
At zero arguments, the same factorization and derivative bounds hold with
\[
 H_{a,c}(0,y)=-2\vartheta(a+c),\quad H_{a,c}(x,0)=2\vartheta(a+c)\quad(x,y\ne0),\qquad H_{a,c}(0,0)=0.
\]

\end{proof}

Let $K_s$ be as in Appendix~\ref{app:transport-coefficients}, Step~2. For any matrix $A$, we write $A=\sum_r A_r$ with $\supp A_r\subseteq\Ball_x(r)$. Finite range gives
\[
 \supp[H,A_r]\subseteq\Ball_x(r+R_0),\qquad
 \|[H,A_r]\|\leq2J|\Ball_x(r)|\,\|A_r\|.
\]
Since $\ad_{K_s}=-(s/2)\ad_H$ and $|\Ball_x(r)|\leq C(1+r)^d$,
\[
 \begin{aligned}
 \|[K_s,A]\|_{p,x}\leq sJ\sum_r(1+r+R_0)^p|\Ball_x(r)|\,\|A_r\|\leq C_p\sum_r(1+r)^{p+d}\|A_r\|.
 \end{aligned}
\]
Taking the infimum and iterating proves
\begin{equation}\label{eq:modular-derivatives}
 \|\ad_{K_s}^m A\|_{p,x}\leq C_{p,m}\|A\|_{p+dm,x},
\end{equation}
where $C_{p,m}$ depends only on $p$, $m$, an upper bound on $\beta J$, the interaction range, and the geometry.

Now, we are ready to prove Lemma~\ref{lem:entropy-filtered-representation}.

\begin{proof}[Proof of Lemma~\ref{lem:entropy-filtered-representation}]
Recall the representation \eqref{eq:entropy-bohr-products}, where $x_i$ are the eigenvalues of $Y$ and $V_\nu,C_\eta$ are the Bohr components used in \eqref{eq:entropy-bohr-identities}. Lemma~\ref{lem:entropy-scalar-kernel} bounds $H_{a,c}(x,y)$ and its derivatives. Thus, with $w(a)=1+a^2$, the function
\[
 F_{x,y}(a,c)=\frac{H_{a,c}(x,y)}{w(a)w(c)}
\]
is integrable in $a,c$. Using the Fourier convention
\[
 \widehat F_{x,y}(t,u)=\frac1{(2\pi)^2}
       \int_{\mathbb R^2}F_{x,y}(a,c)e^{-i(ta+uc)}\,da\,dc,
\]
we can rewrite the scalar factor as
\begin{equation}\label{eq:entropy-separated-kernel}
 H_{a,c}(x,y)=\int_{\mathbb R^2}
 \widehat F_{x,y}(t,u)w(a)e^{ita}w(c)e^{iuc}\,dt\,du.
\end{equation}
Inserting \eqref{eq:entropy-separated-kernel} into \eqref{eq:entropy-bohr-products}, the factors $w(a)e^{ita}$ and $w(c)e^{iuc}$ multiply $V_\nu$ at $a=-s\nu/2$ and $C_\eta$ at $c=-s\eta/2$. With $K_s$ and $\tau_t$ as in Appendix~\ref{app:transport-coefficients}, Step~2, these sums are
\begin{equation}\label{eq:entropy-unbalanced-filters}
 \begin{aligned}
 \widehat P_t&=\sum_\nu w(-s\nu/2)e^{-ist\nu/2}V_\nu
 =\tau_t\bigl(V+[K_s,[K_s,V]]\bigr),\\
 \widehat Q_u&=\sum_\eta w(-s\eta/2)e^{-isu\eta/2}C_\eta
 =\tau_u\bigl(C+[K_s,[K_s,C]]\bigr).
 \end{aligned}
\end{equation}
Consequently, collecting the spectral blocks of $Y$ using \eqref{eq:entropy-kernel-contraction} gives
\[
 \delta_{V,C}(Y)
 =\int_{\mathbb R^2}\Tr\!\left[\mathcal B_{s,\widehat Q_u}(Y)
       \mathcal K_{\widehat F_{\cdot,\cdot}(t,u),Y}
       \bigl(\mathcal B_{s,\widehat P_t}(Y)\bigr)\right]dt\,du.
\]

\emph{Proof of assertion (i).}
This is a representation of the form \eqref{eq:entropy-filtered-identity} with filters $\widehat P_t,\widehat Q_u$, the measure $dt\,du$, and kernel $\widehat F_{x,y}(t,u)$; we now normalize this kernel. Fix an integer $N>p_1+1$. We now derive the uniform bound
\begin{equation}\label{eq:entropy-fourier-bound}
 \sup_{x,y\in\mathbb R}|\widehat F_{x,y}(t,u)|
 \leq C_N(1+t^2)^{-N}(1+u^2)^{-N}.
\end{equation}
Each derivative of $w^{-1}$ is bounded by a constant times $w^{-1}$. The bounds in Lemma~\ref{lem:entropy-scalar-kernel} and the product rule therefore give
\[
 \sup_{x,y}\left\|(1-\partial_a^2)^N(1-\partial_c^2)^N
                         F_{x,y}\right\|_{L^1(\mathbb R^2)}\leq C_N.
\]
Integration by parts in the Fourier integral yields
\[
 (1+t^2)^N(1+u^2)^N|\widehat F_{x,y}(t,u)|
 \leq\frac1{(2\pi)^2}
   \left\|(1-\partial_a^2)^N(1-\partial_c^2)^N F_{x,y}\right\|_1,
\]
which proves \eqref{eq:entropy-fourier-bound}.
This decay also justifies Fourier inversion in \eqref{eq:entropy-separated-kernel} and interchange of the integral with the finite Bohr sums. Writing $\rho_N(t)=(1+t^2)^{-N}$, we absorb the positive bound into the measure by setting
\[
  d\nu_N(t,u)=C_N\rho_N(t)\rho_N(u)\,dt\,du,\qquad
 h_{t,u}(x,y)=\frac{\widehat F_{x,y}(t,u)}{C_N\rho_N(t)\rho_N(u)}.
\]
Then $|h_{t,u}(x,y)|\leq1$, and $\nu_N$ is positive and independent of $Y$. With these choices, the representation reads
\begin{equation}\label{eq:entropy-normalized-representation}
 \delta_{V,C}(Y)
 =\int_{\mathbb R^2}\Tr\!\left[\mathcal B_{s,\widehat Q_u}(Y)
       \mathcal K_{h_{t,u},Y}\bigl(\mathcal B_{s,\widehat P_t}(Y)\bigr)\right]d\nu_N(t,u).
\end{equation}
The scalar filters $w(a)e^{ita}$ and $w(c)e^{iuc}$ are smooth and independent of $Y$.

\emph{Proof of assertion (ii).}
Since the self-adjoint operator $\Ad_{K_s}$ commutes with $\sGamma_{\sigma_s}^{1/2}$, the weighted commutators satisfy
\begin{equation}\label{eq:weighted-modular-covariance}
 \mathsf B_{s,[K_s,V]}=[\Ad_{K_s},\mathsf B_{s,V}],\qquad
 \mathsf B_{s,\tau_t(V)}=e^{it\Ad_{K_s}}\mathsf B_{s,V}e^{-it\Ad_{K_s}}.
\end{equation}
On the site tensor factors of Hilbert--Schmidt space, $\Ad_{K_s}=K_s\otimes I-I\otimes K_s^{\mathrm T}$ has interaction range $R_0$ and interaction strength at most $sJ$, so it obeys the same evolution and commutator estimates as \eqref{eq:real-polynomial} and \eqref{eq:modular-derivatives}. Applying these estimates to \eqref{eq:entropy-unbalanced-filters} gives
\begin{equation}\label{eq:entropy-filter-locality}
 \begin{aligned}
 \|\mathsf B_{s,\widehat P_t}\|_{p_0,x}
 &\leq C(1+|t|)^{p_1}\|\mathsf B_{s,V}\|_{p_1+2d,x},\\
 \|\mathsf B_{s,\widehat Q_u}\|_{p_0,x}
 &\leq C(1+|u|)^{p_1}\|\mathsf B_{s,C}\|_{p_1+2d,x}.
 \end{aligned}
\end{equation}
Here the evolution requires $p_1=p_0+d+2$, and the double commutator increases the required exponent to $p_1+2d$. Write
\[
 a_V=\|\mathsf B_{s,V}\|_{p_1+2d,x},\qquad
 a_C=\|\mathsf B_{s,C}\|_{p_1+2d,x}.
\]
Since $N>p_1+1$, the squared bounds are integrable against $\nu_N$, giving
\[
 \int\|\mathsf B_{s,\widehat P_t}\|_{p_0,x}^2\,d\nu_N\leq C_{\log}a_V^2,
 \qquad
 \int\|\mathsf B_{s,\widehat Q_u}\|_{p_0,x}^2\,d\nu_N\leq C_{\log}a_C^2.
\]
If $a_Va_C=0$, one of the maps $\mathsf B_{s,V}$ or $\mathsf B_{s,C}$ is zero. By \eqref{eq:weighted-modular-covariance}, all its filtered versions are also zero, so \eqref{eq:entropy-normalized-representation} vanishes. Taking both filters to be zero proves (ii) in this case.

For $a_Va_C>0$, obtain the product required in \eqref{eq:entropy-pair-local-strength} by rescaling the two filters reciprocally:
\[
 \chi=(a_C/a_V)^{1/2},\qquad
 P_{t,u}=\chi\widehat P_t,\qquad Q_{t,u}=\chi^{-1}\widehat Q_u.
\]

The factors cancel in \eqref{eq:entropy-normalized-representation}, leaving the representation unchanged. The squared quasi-local norms now satisfy
\[
 \int\bigl(\|\mathsf B_{s,P_{t,u}}\|_{p_0,x}^2+
            \|\mathsf B_{s,Q_{t,u}}\|_{p_0,x}^2\bigr)\,d\nu_N
 \leq C_{\log}(\chi^2a_V^2+\chi^{-2}a_C^2)
 =2C_{\log}a_Va_C.
\]
Here $C_{\log}$ is independent of $V$ and $C$. The rescaled scalar filters $p_{t,u}(a)=\chi w(a)e^{ita}$ and $q_{t,u}(c)=\chi^{-1}w(c)e^{iuc}$ in~\eqref{eq:entropy-filtered-matrices} remain smooth and independent of $Y$.

\end{proof}

\proofstep{Step 4: bound the entropy derivative}
The representation has now been established. We are ready to prove the entropy estimate in Lemma~\ref{lem:derivative-reduction}.

\begin{proof}[Proof of~\eqref{eq:entropy-derivative-bound} in Lemma~\ref{lem:derivative-reduction}]
Apply Lemma~\ref{lem:entropy-filtered-representation}\textup{(i)} to each pair in \eqref{eq:entropy-pair-derivative}. This gives \eqref{eq:entropy-integral-representation}. Write $\gamma=(k,Z,\alpha)$, with center $x_\gamma=x_k$; the sum over $\alpha$ includes $d\nu_N(t,u)$.

By \eqref{eq:entropy-kernel-contraction} and $\|h_\gamma\|_\infty\leq1$, $\mathcal K_{h_\gamma,Y_s}$ is a contraction. Thus $|\Tr[A\mathcal K_{h_\gamma,Y_s}(B)]|\leq\|A\|_2\|B\|_2$ and $ab\leq(a^2+b^2)/2$ give
\begin{equation}\label{eq:entropy-square-reduction}
 \begin{aligned}
 |\partial_sD(Y_s^2\Vert\sigma_s)|
 &\leq\sum_\gamma\int
       \|\mathsf{B}_{s,P_\gamma(\omega)}|Y_s\rangle\|_2
       \|\mathsf{B}_{s,Q_\gamma(\omega)}|Y_s\rangle\|_2\,d\omega\\
 &\leq\tfrac12\sum_\gamma\int
       \bigl(\|\mathsf{B}_{s,P_\gamma(\omega)}|Y_s\rangle\|_2^2
             +\|\mathsf{B}_{s,Q_\gamma(\omega)}|Y_s\rangle\|_2^2\bigr)\,d\omega.
 \end{aligned}
\end{equation}
Equation~\eqref{eq:entropy-square-reduction} proves~\eqref{eq:entropy-derivative-bound} of Lemma~\ref{lem:derivative-reduction}.
\end{proof}

\proofstep{Step 5: sum the quasi-local norms of the filtered commutators}
We now prove the entropy estimate \eqref{eq:entropy-family-locality} in Lemma~\ref{lem:comparison-families}. Since Gibbs weighting introduces conjugations by $\sigma_s^{\pm1/4}$, we first bound $\sigma_s^{\pm1/4}V_k(\omega)\sigma_s^{\mp1/4}$ in Lemma~\ref{lem:filtered-modular} and $\sigma_s^{\pm1/4}f_{s,Z}\sigma_s^{\mp1/4}$ in Lemma~\ref{lem:potential-modular}.


\begin{lemma}[Modular bounds for Gaussian-filtered jumps]\label{lem:filtered-modular}
For the finite-range Hamiltonians in \eqref{eq:growth} and every fixed $\beta,p$, the Hermitian matrices in \eqref{eq:weighted-jump-coefficients} satisfy
\[
 \sup_{0<s\leq\beta}\sup_x\sum_{k:x_k=x}\int_{\mathbb R}\!d\omega\,
 \left(\sum_{\epsilon=\pm1}
 \|\sigma_s^{\epsilon/4}V_k(\omega)\sigma_s^{-\epsilon/4}\|_{p,x}\right)^2
 \leq C_{p,\beta}.
\]
In particular,
\[
 \sup_{s,x}\sum_{k:x_k=x}\int_{\mathbb R}\|\mathsf{B}_{s,V_k(\omega)}\|_{p,x}^2\,d\omega
 \leq C_{p,\beta}.
\]
\end{lemma}
\begin{proof}[Proof of Lemma~\ref{lem:filtered-modular}]
Recall $V_k(\omega)$ from \eqref{eq:weighted-jump-coefficients} and the rescaled jumps $\widehat J^s_{x,u}(z)$ from \eqref{eq:native-rescaled-jump}. Applying \eqref{eq:hamiltonian-locality} to $S_{x,u}$ gives
\[
 \|\widehat J^s_{x,u}(z)\|_{p,x}
 \leq C_p\int_{\mathbb R}e^{-r^2}(1+sJ|r|)^{p+d+2}\,dr
 \leq C_{p,\beta J}.
\]
The bound is independent of $z$. For $\omega=z/s$, the two Gibbs conjugations of the weighted jump inside \eqref{eq:weighted-jump-coefficients} correspond to $a=0$ and $a=1/2$ below: Taking adjoints in \eqref{eq:gaussian-modular-shift} gives
\[
 e^{-(z+1)^2/4}\sigma_s^a\widehat J^s_{x,u}(z)^*\sigma_s^{-a}
 =\begin{cases}
 e^{-(z+1)^2/4}\widehat J^s_{x,u}(z)^*,&a=0,\\
 e^{-z^2/4}\widehat J^s_{x,u}(z+1)^*,&a=1/2.
 \end{cases}
\]
Since $J^s_{x,u}(z/s)=\sqrt{s}\,\widehat J^s_{x,u}(z)$, for $r=1,2$, the above calculations give
\[
 \sum_{\epsilon=\pm1}
 \|\sigma_s^{\epsilon/4}V_{x,u,r}(z/s)\sigma_s^{-\epsilon/4}\|_{p,x}
 \leq C_{p,\beta J}\sqrt{s}\,
       \bigl(e^{-(z+1)^2/4}+e^{-z^2/4}\bigr).
\]
Finally, $\|\mathsf{B}_{s,V}\|_{p,x}\leq\sum_{\epsilon=\pm1}\|\sigma_s^{\epsilon/4}V\sigma_s^{-\epsilon/4}\|_{p,x}$ gives the second bound.
\end{proof}

The following estimate controls the Gibbs conjugations of $f_{s,Z}$.
\begin{lemma}[Modular bounds for the Poisson potential]\label{lem:potential-modular}
For $f_{s,Z}=\sech(s\ad_H/4)X_{s,Z}$ and every $p\geq0$,
\[
 \sum_{\epsilon=\pm1}\|\sigma_s^{\epsilon/4}f_{s,Z}\sigma_s^{-\epsilon/4}\|_{p,x_Z}
 \leq C_p\|X_{s,Z}\|_{p+2d+2,x_Z}.
\]
\end{lemma}
\begin{proof}[Proof of Lemma~\ref{lem:potential-modular}]
With $K_s$ and $\tau_t$ as in Appendix~\ref{app:transport-coefficients}, Step~2, use the identities
\[
 \sigma_s^{\pm1/4}f_{s,Z}\sigma_s^{\mp1/4}
 =X_{s,Z}\pm\tanh(\ad_{K_s}/2)X_{s,Z},\qquad
 \tanh(\ad_{K_s}/2)X
 =-i\int_0^\infty\frac{\tau_t(X)-\tau_{-t}(X)}{\sinh(\pi t)}\,dt.
\]
The first identity follows from $e^{\pm a/2}\sech(a/2)=1\pm\tanh(a/2)$ with $a=\ad_{K_s}$. To estimate the integral, split it at $t=1$. For $0<t\leq1$,
\[
 \tau_t(X)-\tau_{-t}(X)=i\int_{-t}^t\tau_u([K_s,X])\,du,
 \qquad
 \|\tau_t(X)-\tau_{-t}(X)\|_{p,x}
 \leq C_pt\|[K_s,X]\|_{p+d+2,x}.
\]
For $t\geq1$, \eqref{eq:real-polynomial} bounds the difference by
$C_p(1+t)^{p+d+2}\|X\|_{p+d+2,x}$. Consequently,
\[
 \begin{aligned}
 \|\tanh(\ad_{K_s}/2)X\|_{p,x}
 &\leq C_p\|[K_s,X]\|_{p+d+2,x}
                \int_0^1\frac{t\,dt}{\sinh(\pi t)}+C_p\|X\|_{p+d+2,x}
                \int_1^\infty\frac{(1+t)^{p+d+2}}{\sinh(\pi t)}\,dt\\
 &\leq C_p\|X\|_{p+2d+2,x}.
 \end{aligned}
\]
Both scalar integrals are finite, and the last inequality uses \eqref{eq:modular-derivatives}.
Substituting $X=X_{s,Z}$ proves the bound.
\end{proof}

Now, we are ready to prove the entropy estimate \eqref{eq:entropy-family-locality} in Lemma~\ref{lem:comparison-families}.

\begin{proof}[Proof of~\eqref{eq:entropy-family-locality} in Lemma~\ref{lem:comparison-families}]
Set
\[
 p_1=p_0+d+2,\qquad p=p_1+2d=p_0+3d+2,\qquad \ell=p+d+2.
\]
The pairwise estimate \eqref{eq:entropy-pair-local-strength} reduces the desired bound to summing
$\|\mathsf B_{s,V_k(\omega)}\|_{p,x}\|\mathsf B_{s,C_{k,Z}(\omega)}\|_{p,x}$ over $k,Z$ with $x_k=x$ and integrating over $\omega$. We first control this sum using the two previous lemmas.

Since $C_{k,Z}(\omega)=[V_k(\omega),f_{s,Z}]$ (Lemma~\ref{lem:transport-decomposition}), we have
\[
 \mathsf B_{s,C_{k,Z}(\omega)}
 =[\mathsf B_{s,V_k(\omega)},\mathsf B_{s,f_{s,Z}}].
\]
For each $k,\omega$, apply the summed commutator bound \eqref{eq:family-local-commutator} to the family $\mathsf B_{s,f_{s,Z}}$, centered at $x_Z$. This gives
\begin{equation}\label{eq:modular-product-reduction}
 \begin{aligned}
 &\sup_{s,x}\sum_{k,Z:x_k=x}\int_{\mathbb R}\!d\omega\,
 \|\mathsf{B}_{s,V_k(\omega)}\|_{p,x}\|\mathsf{B}_{s,C_{k,Z}(\omega)}\|_{p,x}\\
 &\quad\leq C\left(\sup_{s,x}
 \sum_{k:x_k=x}\int_{\mathbb R}\!d\omega\,\|\mathsf{B}_{s,V_k(\omega)}\|_{\ell,x}^2\right)
 \left(\sup_{s,z}\sum_{Z:x_Z=z}\|\mathsf{B}_{s,f_{s,Z}}\|_{\ell,z}\right).
 \end{aligned}
\end{equation}
For the first factor, the Gaussian bound in Lemma~\ref{lem:filtered-modular} gives
\[
 \sup_{s,x}\sum_{k:x_k=x}\int_{\mathbb R}\!d\omega\,
 \|\mathsf{B}_{s,V_k(\omega)}\|_{\ell,x}^2\leq C.
\]
For the second factor, Lemma~\ref{lem:potential-modular} reduces the estimate to the quasi-local norm of $X_{s,Z}$ at order $\ell+2d+2$. Proposition~\ref{prop:profile} applies at this order because
\[
 \nu>p_0+9d+12>\ell+2d+4=p_0+6d+8.
\]
Consequently, Lemma~\ref{lem:potential-modular} and \eqref{eq:local-poisson} give
\[
 \|\mathsf{B}_{s,f_{s,Z}}\|_{\ell,x_Z}
 \leq\sum_{\epsilon=\pm1}
       \|\sigma_s^{\epsilon/4}f_{s,Z}\sigma_s^{-\epsilon/4}\|_{\ell,x_Z}
 \leq C\|X_{s,Z}\|_{\ell+2d+2,x_Z}
 \leq C\|h_Z\|_\infty.
\]
Since $x_Z\in Z$, the interaction-strength bound \eqref{eq:growth} implies
\[
 \sup_{s,z}\sum_{Z:x_Z=z}\|\mathsf B_{s,f_{s,Z}}\|_{\ell,z}
 \leq C\sup_z\sum_{Z:x_Z=z}\|h_Z\|_\infty\leq CJ.
\]
Substituting the two factor bounds into \eqref{eq:modular-product-reduction} yields
\begin{equation}\label{eq:flow-modular-target}
 \sup_{s,x}\sum_{k,Z:x_k=x}\int_{\mathbb R}\!d\omega\,
 \|\mathsf B_{s,V_k(\omega)}\|_{p,x}
 \|\mathsf B_{s,C_{k,Z}(\omega)}\|_{p,x}\leq CJ.
\end{equation}

Summing \eqref{eq:entropy-pair-local-strength} over $k,Z$ and integrating over $\omega$, then applying \eqref{eq:flow-modular-target} with $p=p_1+2d$, gives
\begin{equation}\label{eq:entropy-gradient-local-strength}
 \begin{aligned}
 &\sup_{s,x}\sum_{\gamma:x_\gamma=x}\int_{\mathbb R}\!d\omega\,
       \bigl(\|\mathsf{B}_{s,P_\gamma(\omega)}\|_{p_0,x}^2
             +\|\mathsf{B}_{s,Q_\gamma(\omega)}\|_{p_0,x}^2\bigr)\\
 &\quad\leq2C_{\log}\sup_{s,x}\sum_{k,Z:x_k=x}\int_{\mathbb R}\!d\omega\,
       \|\mathsf{B}_{s,V_k(\omega)}\|_{p_1+2d,x}
       \|\mathsf{B}_{s,C_{k,Z}(\omega)}\|_{p_1+2d,x}\\
 &\quad\leq2C_{\log}CJ=:M_S^2.
 \end{aligned}
\end{equation}
This proves \eqref{eq:entropy-family-locality}. Together with the energy estimate proved in Appendix~\ref{sec:energy-reduction}, it completes the proof of Lemma~\ref{lem:comparison-families}.
\end{proof}

\section{Verification of the local action-energy inequality}\label{sec:local-gap-criterion-proof}
Our goal is Condition~\ref{cond:local-action}: for every family with
\[
 \mathsf R_{s,\alpha}|\Omega_{\sigma_s}\rangle=0,
 \qquad M^2=\sup_{s,x}\sum_{\alpha:x_\alpha=x}
                    \|\mathsf R_{s,\alpha}\|_{p_0,x}^2<\infty,
\]
we must prove
\[
 \sum_\alpha\mathsf R_{s,\alpha}^*\mathsf R_{s,\alpha}
 \leq C_{\mathrm{loc}}M^2(-\mathsf L_{\sigma_s}).
\]
This is \eqref{eq:comparison-target}. We prove it in three parts.
\begin{enumerate}
\item In Appendix~\ref{app:local-action-criterion}, we prove Proposition~\ref{prop:local-gap-criterion}. The local state domination in \eqref{eq:local-gap-kernel-assumption} bounds the action of each $\mathsf R_{s,\alpha}$ on regional kernels. The local spectral gaps in \eqref{eq:local-gap-assumption} then give a bound by regional energies. Summing over centers and radii and using \eqref{eq:local-gap-auxiliary} proves \eqref{eq:comparison-target}.
\item In Appendix~\ref{app:one-dimensional}, we verify \eqref{eq:local-gap-auxiliary}, \eqref{eq:local-gap-assumption}, and \eqref{eq:local-gap-kernel-assumption} for the weighted Pauli operator \eqref{eq:pauli-form} to verify Condition~\ref{cond:local-action} for the one-dimensional sampler. The bounds on the change of Gibbs weighting in \eqref{eq:chain-cut-map} give both the local spectral gap and the comparison of interior expectations.
\item In Appendix~\ref{app:high-temperature}, we verify the same three equations for the transported product form \eqref{eq:ht-transported-form} to verify Condition~\ref{cond:local-action} for the high-temperature sampler. Unitary conjugation preserves the local spectral gaps, and \eqref{eq:ht-observable-truncation} gives local state domination. To compare the form with the sampler, we use \eqref{eq:ht-poisson-action}, \eqref{eq:ht-poisson-product-quadratic}, and the derivative estimate \eqref{eq:differentiated-gradient-locality} from Appendix~\ref{sec:energy-reduction}.
\end{enumerate}

\subsection{From Condition~\ref{ass:local-gap-state-domination} to Condition~\ref{cond:local-action}}\label{app:local-action-criterion}
The comparison in \eqref{eq:local-gap-auxiliary} reduces Proposition~\ref{prop:local-gap-criterion} to proving
\begin{equation}\label{eq:local-gap-proof-target}
 \sum_\alpha\|\mathsf R_{s,\alpha}|Y\rangle\|_2^2
 \leq CM^2\langle Y|\mathsf A_s|Y\rangle,
\end{equation}
with $C$ independent of $s$ and the volume.

\proofstep{Step 1: compare with complements of nested projections}

\begin{lemma}[A bound from nested kernels]\label{lem:nested-kernel-bound}
Let $\mathsf{P}_0\geq \mathsf{P}_1\geq\cdots\geq \mathsf{P}_N$ be orthogonal projections on a finite-dimensional Hilbert space. Suppose $\mathsf{R}\mathsf{P}_N=0$, $\|\mathsf{R}\|\leq m$, and $\|\mathsf{R}\mathsf{P}_k\|\leq m\varepsilon_k$. With $\mathsf{Q}_k=\one-\mathsf{P}_k$, one has
\begin{equation}\label{eq:nested-kernel-bound}
 \mathsf{R}^*\mathsf{R}\leq m^2\mathsf{Q}_0+
       2m^2\sum_{k=1}^N(\varepsilon_{k-1}+\varepsilon_k)\mathsf{Q}_k.
\end{equation}
\end{lemma}
\begin{proof}[Proof of Lemma~\ref{lem:nested-kernel-bound}]
To obtain \eqref{eq:nested-kernel-bound}, we decompose $\mathsf{R}^*\mathsf{R}$ into increments supported on the increasing projections $\mathsf{Q}_k$ and bound each increment. The compressions $\mathsf{F}_k=\mathsf{Q}_k\mathsf{R}^*\mathsf{R}\mathsf{Q}_k$ satisfy $\mathsf{F}_N=\mathsf{R}^*\mathsf{R}$ because $\mathsf{R}\mathsf{P}_N=0$, so
\[
 \mathsf{R}^*\mathsf{R}=\mathsf{F}_0+\sum_{k=1}^N(\mathsf{F}_k-\mathsf{F}_{k-1}).
\]
The first term obeys $\mathsf{F}_0\leq m^2\mathsf{Q}_0$. To control each remaining increment, compare both compressions with $\mathsf{R}^*\mathsf{R}$:
\[
 \|\mathsf{R}^*\mathsf{R}-\mathsf{F}_k\|\leq2\|\mathsf{R}\|\|\mathsf{R}\mathsf{P}_k\|\leq2m^2\varepsilon_k.
\]
Since $\mathsf{Q}_{k-1}\leq \mathsf{Q}_k$, the Hermitian difference $\mathsf{F}_k-\mathsf{F}_{k-1}$ is supported on $\mathsf{Q}_k$. Hence
\[
 \mathsf{F}_k-\mathsf{F}_{k-1}\leq\|\mathsf{F}_k-\mathsf{F}_{k-1}\|\mathsf{Q}_k\leq2m^2(\varepsilon_{k-1}+\varepsilon_k)\mathsf{Q}_k.
\]
Summing these bounds proves \eqref{eq:nested-kernel-bound}.
\end{proof}

\proofstep{Step 2: bound the action on regional kernels}
Recall that $\mathsf{P}_{s,I}$ is the Hilbert-Schmidt kernel projection of $\mathsf{A}_{s,I}=\sum_{z\in I}\mathsf{a}_{s,z}$. We will prove
\begin{equation}\label{eq:local-gap-derived-kernel-action}
 \|\mathsf{R}\mathsf{P}_{s,\Ball_x(r)}\|\leq K\epsilon(r)\|\mathsf{R}\|_{p_0,x}\,,
\end{equation}
where
\begin{equation}\label{eq:local-gap-derived-decay}
 \epsilon(r)=(1+\lfloor r/2\rfloor)^{-p_0}+\sqrt{\delta_*}(1+r)^{-b},\qquad K=2(1+\sqrt{C_{\mathrm{dom}}}).
\end{equation}
\begin{proof}[Proof of \eqref{eq:local-gap-derived-kernel-action}]
Fix $s$ and an operator $\mathsf R$ centered at $x$ with $\mathsf R|\Omega_{\sigma_s}\rangle=0$. Assume $0<\|\mathsf R\|_{p_0,x}<\infty$, the zero case being immediate. To apply \eqref{eq:local-gap-kernel-assumption} inside $\Ball_x(r)$, choose a decomposition in \eqref{eq:polynorm} with weighted sum at most $2\|\mathsf R\|_{p_0,x}$, and retain the terms of radius at most $\lfloor r/2\rfloor$. This gives a truncation $\widetilde{\mathsf R}_r$ supported on $\Ball_x(\lfloor r/2\rfloor)$ with
\[
 \|\widetilde{\mathsf{R}}_r\|\leq2\|\mathsf{R}\|_{p_0,x},\qquad
 \|\mathsf{R}-\widetilde{\mathsf{R}}_r\|\leq2\|\mathsf{R}\|_{p_0,x}(1+\lfloor r/2\rfloor)^{-p_0}.
\]
Because $\mathsf R|\Omega_{\sigma_s}\rangle=0$ and $\|\Omega_{\sigma_s}\|_2^2=\Tr\sigma_s=1$,
\[
 \|\widetilde{\mathsf R}_r|\Omega_{\sigma_s}\rangle\|_2
 =\|(\widetilde{\mathsf R}_r-\mathsf R)|\Omega_{\sigma_s}\rangle\|_2
 \leq\|\mathsf R-\widetilde{\mathsf R}_r\|.
\]
To transfer this bound to~\eqref{eq:local-gap-derived-kernel-action}, set $I=\Ball_x(r)$ and $J=\Ball_x(\lfloor r/2\rfloor)$. The case $\widetilde{\mathsf R}_r=0$ is immediate. Otherwise, the operator
\[
 \mathsf O_J=\frac{\widetilde{\mathsf R}_r^*\widetilde{\mathsf R}_r}
                  {\|\widetilde{\mathsf R}_r\|^2}
\]
is supported on $J$ and satisfies $0\leq\mathsf O_J\leq\one$, since $\widetilde{\mathsf R}_r^*\widetilde{\mathsf R}_r\leq\|\widetilde{\mathsf R}_r\|^2\one$. Thus \eqref{eq:local-gap-kernel-assumption} applies to every unit vector $|\psi_s\rangle\in\ker\mathsf A_{s,I}$. Multiplying that inequality by $\|\widetilde{\mathsf R}_r\|^2$ gives
\[
 \begin{aligned}
 \|\widetilde{\mathsf R}_r|\psi_s\rangle\|_2^2=\|\widetilde{\mathsf R}_r\|^2\langle\psi_s|\mathsf O_J|\psi_s\rangle\leq C_{\mathrm{dom}}\|\widetilde{\mathsf R}_r|\Omega_{\sigma_s}\rangle\|_2^2
       +\delta_*(1+r)^{-2b}\|\widetilde{\mathsf R}_r\|^2.
 \end{aligned}
\]
The right-hand side is independent of $|\psi_s\rangle$. Since $\mathsf P_{s,I}$ is the orthogonal projection onto $\ker\mathsf A_{s,I}$, it is contractive and fixes every vector in that kernel. Taking the supremum over its unit vectors, and then using the preceding bound on $\|\widetilde{\mathsf R}_r|\Omega_{\sigma_s}\rangle\|_2$, therefore yields
\[
 \begin{aligned}
 \|\widetilde{\mathsf R}_r\mathsf P_{s,I}\|^2=\sup_{\substack{|\psi_s\rangle\in\ker\mathsf A_{s,I}\\\|\psi_s\|_2=1}}
       \|\widetilde{\mathsf R}_r|\psi_s\rangle\|_2^2\leq C_{\mathrm{dom}}\|\mathsf R-\widetilde{\mathsf R}_r\|^2
       +\delta_*(1+r)^{-2b}\|\widetilde{\mathsf R}_r\|^2.
 \end{aligned}
\]
Adding back the truncation error gives
\[
 \|\mathsf{R}\mathsf{P}_{s,\Ball_x(r)}\|
 \leq2\|\mathsf{R}\|_{p_0,x}\left[(1+\sqrt{C_{\mathrm{dom}}})(1+\lfloor r/2\rfloor)^{-p_0}+\sqrt{\delta_*}(1+r)^{-b}\right].
\]
This concludes the proof.
\end{proof}

\proofstep{Step 3: compare with regional energies}

We next convert \eqref{eq:local-gap-derived-kernel-action} into a bound on $\mathsf{R}^*\mathsf{R}$. Fix $s,x$ and $\mathsf R$ as in Step~2. Choose $r_0=0$ and $r_k=2^{k-1}$ for $k\geq1$, and use the notation of Lemma~\ref{lem:nested-kernel-bound}:
\[
 \mathsf P_k=\mathsf P_{s,\Ball_x(r_k)},\qquad
 m=\|\mathsf R\|_{p_0,x},\qquad
 \varepsilon_k=K\epsilon(r_k).
\]
Because the balls are nested and the terms $\mathsf a_{s,z}$ are positive,
\[
 \mathsf A_{s,\Ball_x(r_{k+1})}-\mathsf A_{s,\Ball_x(r_k)}
 =\sum_{z\in\Ball_x(r_{k+1})\setminus\Ball_x(r_k)}\mathsf a_{s,z}\geq0.
\]
Consequently, $\ker\mathsf A_{s,\Ball_x(r_{k+1})}\subseteq\ker\mathsf A_{s,\Ball_x(r_k)}$, so $\mathsf P_{k+1}\leq\mathsf P_k$. At finite volume, choose $N$ such that $\Ball_x(r_N)=\Lambda$. By \eqref{eq:local-gap-auxiliary},
\[
 \mathsf P_N=|\Omega_{\sigma_s}\rangle\langle\Omega_{\sigma_s}|,
 \qquad \mathsf R\mathsf P_N=0.
\]
The definition \eqref{eq:polynorm} gives $\|\mathsf R\|\leq m$, and \eqref{eq:local-gap-derived-kernel-action} gives $\|\mathsf R\mathsf P_k\|\leq m\varepsilon_k$. Thus \eqref{eq:nested-kernel-bound} yields
\[
 \begin{aligned}
 \mathsf R^*\mathsf R
 &\leq m^2(\one-\mathsf P_0)
       +2m^2\sum_{k=1}^N(\varepsilon_{k-1}+\varepsilon_k)(\one-\mathsf P_k)\\
 &=\|\mathsf R\|_{p_0,x}^2\left[
       (\one-\mathsf P_0)
       +2K\sum_{k=1}^N\bigl(\epsilon(r_{k-1})+\epsilon(r_k)\bigr)
          (\one-\mathsf P_k)\right].
 \end{aligned}
\]
We now replace each projection complement by the corresponding regional energy. Rearranging \eqref{eq:local-gap-assumption} gives
\[
 \begin{aligned}
 \one-\mathsf P_k
 &\leq\frac{(1+r_k)^a}{\gamma_*}\mathsf A_{s,\Ball_x(r_k)},\qquad k\geq1,\\
 \one-\mathsf P_0
 &\leq\frac1{\gamma_*}\mathsf A_{s,\Ball_x(0)},
 \end{aligned}
\]
where the second line uses $r_0=0$. Factoring out $\|\mathsf R\|_{p_0,x}^2/\gamma_*$ and extending the finite sum to $k\geq1$ by adding positive terms gives
\begin{equation}\label{eq:local-gap-single-operator}
 \mathsf{R}^*\mathsf{R}\leq\frac{\|\mathsf{R}\|_{p_0,x}^2}{\gamma_*}\left[
 \mathsf{A}_{s,\Ball_x(0)}
 +2K\sum_{k\geq1}
 (1+r_k)^a\bigl(\epsilon(r_{k-1})+\epsilon(r_k)\bigr)
 \mathsf{A}_{s,\Ball_x(r_k)}\right].
\end{equation}

\proofstep{Step 4: sum over centers and radii}
\begin{proof}[Proof of Proposition~\ref{prop:local-gap-criterion}]
To sum \eqref{eq:local-gap-single-operator} over the family, we control how often each energy term is counted. At radius $r$, a site belongs to at most $(2r+1)^d\leq2^d(1+r)^d$ balls. The definition of local strength in Condition~\ref{cond:local-action} and positivity therefore give
\[
 \sum_\alpha\|\mathsf{R}_{s,\alpha}\|_{p_0,x_\alpha}^2\mathsf{A}_{s,\Ball_{x_\alpha}(r)}
 \leq M^2\sum_x\mathsf{A}_{s,\Ball_x(r)}\leq 2^dM^2(1+r)^d\mathsf{A}_s.
\]
Substituting this bound into the sum of \eqref{eq:local-gap-single-operator} gives
\begin{equation}\label{eq:local-gap-summability}
 \begin{aligned}
 \sum_\alpha\|\mathsf R_{s,\alpha}|Y\rangle\|_2^2
 &\leq\frac{2^dM^2}{\gamma_*}
 \left[1+2K\sum_{k\geq1}(1+r_k)^{d+a}
 \bigl(\epsilon(r_{k-1})+\epsilon(r_k)\bigr)\right]
 \langle Y|\mathsf A_s|Y\rangle.
 \end{aligned}
\end{equation}
It remains to bound the series in \eqref{eq:local-gap-summability}. Set $c=\min\{p_0,b\}>d+a$. Since $1+\lfloor r/2\rfloor\geq(1+r)/2$, the definition \eqref{eq:local-gap-derived-decay} gives
\[
 \epsilon(r)\leq(2^{p_0}+\sqrt{\delta_*})(1+r)^{-c}.
\]
Using $1+r_k\leq2(1+r_{k-1})$, each summand in \eqref{eq:local-gap-summability} is controlled by
\[
 (1+r_k)^{d+a}\bigl(\epsilon(r_{k-1})+\epsilon(r_k)\bigr)
 \leq(2^c+1)(2^{p_0}+\sqrt{\delta_*})\,2^{-(k-1)(c-d-a)}.
\]
The resulting geometric series gives
\[
 \sum_{k\geq1}(1+r_k)^{d+a}\bigl(\epsilon(r_{k-1})+\epsilon(r_k)\bigr)
 \leq\frac{(2^c+1)(2^{p_0}+\sqrt{\delta_*})}{1-2^{-(c-d-a)}}<\infty.
\]
Substitution into \eqref{eq:local-gap-summability} proves \eqref{eq:local-gap-proof-target}. Applying \eqref{eq:local-gap-auxiliary} then proves \eqref{eq:comparison-target}, with $C_{\mathrm{loc}}$ independent of $s$ and the volume.
\end{proof}

\subsection{Proof for one-dimensional systems}\label{app:one-dimensional}\label{app:domination-chain}
We verify Condition~\ref{ass:local-gap-state-domination} for the chains in \eqref{eq:chain-setting}, with constants depending only on $\beta,q$. We use the weighted Pauli operator \eqref{eq:pauli-form}. We first prove \eqref{eq:local-gap-auxiliary}. The bounds on $\mathsf W=\sGamma_{\sigma_s}^{1/2}\sGamma_\tau^{-1/2}$ in \eqref{eq:chain-cut-map} then give \eqref{eq:BC-input} and \eqref{eq:local-gap-kernel-assumption}. These verify the hypotheses of Proposition~\ref{prop:local-gap-criterion} with $a=1$ and $p_0=b=10$.

\proofstep{Step 1: compare the weighted Pauli operator with the sampler}
\begin{proof}[Proof of \eqref{eq:local-gap-auxiliary} in one dimension]
According to~\cite[Theorem~VI.1, Eq.~(189)]{BC2025},
\begin{equation}\label{eq:chain-pauli-comparison-input}
 \sum_{x,u}\|[S_{x,u},X]\|_{2,\sigma_s}^2
 \leq c_{\beta,q}\,\mathcal E_{2,\sigma_s}(X),
 \qquad X\in\mathcal B(\mathcal H),\quad 0\leq s\leq\beta.
\end{equation}
Substituting $|Y\rangle=\sGamma_{\sigma_s}^{1/2}|X\rangle$ gives
\[
 \mathsf A_s\leq c_{\beta,q}(-\mathsf L_{\sigma_s}).
\]
Also, $\mathsf B_{s,x,u}|\Omega_{\sigma_s}\rangle=0$. If all these commutators vanish, then $X$ commutes with every on-site matrix, hence $X$ is scalar. Thus $\ker\mathsf A_s=\operatorname{span}\{|\Omega_{\sigma_s}\rangle\}$, proving \eqref{eq:local-gap-auxiliary}.
\end{proof}

\proofstep{Step 2: estimate the change of weighting across the boundary}
For an interval $I$, removing its at most two boundary bonds gives
\begin{equation}\label{eq:chain-cut-reference}
 H^{\mathrm{cut}}=H_I+H_{I^c}^{\mathrm{cut}},\qquad
 \tau=\tau_I\otimes\tau_{I^c},\qquad
 \mathsf{W}=\sGamma_{\sigma_s}^{1/2}\sGamma_\tau^{-1/2},
\end{equation}
where $H_I$ and $H_{I^c}^{\mathrm{cut}}$ contain the interaction terms supported entirely inside and outside $I$, respectively, and $\tau_I\propto e^{-sH_I}$, $\tau_{I^c}\propto e^{-sH_{I^c}^{\mathrm{cut}}}$.
Write $\mathsf A_{s,I}=\sum_{x\in I,u}\mathsf B_{s,x,u}^*\mathsf B_{s,x,u}$ and let $\mathsf P_{s,I}$ be its kernel projection. An operator commuting with every on-site Pauli in $I$ is of the form $I_I\otimes X$, so
\[
 \ker\mathsf A_{s,I}
 =\{\sGamma_{\sigma_s}^{1/2}|I_I\otimes X\rangle:
                         X\in\mathcal B(\mathcal H_{I^c})\}.
\]
Replacing $\sigma_s$ by the product state $\tau$ gives the weighted Pauli operator of the isolated interval,
\[
 \mathsf{A}_{s,I}^{\mathrm{cut}}=\sum_{x\in I,u}(\mathsf{B}_{s,x,u}^{\mathrm{cut}})^*\mathsf{B}_{s,x,u}^{\mathrm{cut}},\qquad
 \mathsf{B}_{s,x,u}^{\mathrm{cut}}=\sGamma_\tau^{1/2}\Ad_{S_{x,u}}\sGamma_\tau^{-1/2}\,.
\]
Furthermore,
\[
 \ker\mathsf A_{s,I}^{\mathrm{cut}}
 =\{\sGamma_\tau^{1/2}|I_I\otimes X\rangle:
                         X\in\mathcal B(\mathcal H_{I^c})\}.
\]
Since $\mathsf W\sGamma_\tau^{1/2}=\sGamma_{\sigma_s}^{1/2}$, the kernels and Gibbs vectors therefore satisfy
\begin{equation}\label{eq:chain-cut-kernel}
 \mathsf{W}|\Omega_\tau\rangle=|\Omega_{\sigma_s}\rangle,\qquad
 \ran \mathsf{P}_{s,I}=\mathsf{W}\ker\mathsf{A}_{s,I}^{\mathrm{cut}}.
\end{equation}

\begin{lemma}[Bounds and locality of the boundary weighting]\label{lem:chain-cut-locality}
There are $K\geq1$ and $c>0$ such that $\mathsf{W}$ has approximants $\widetilde{\mathsf W}_L$ supported within distance $L$ of the cut bonds, with
\begin{equation}\label{eq:chain-cut-map}
 \|\mathsf{W}\|+\|\mathsf{W}^{-1}\|+\|\widetilde{\mathsf W}_L\|\leq K,\qquad
 \|\mathsf{W}-\widetilde{\mathsf W}_L\|\leq Ke^{-cL\log(2+L)}.
\end{equation}
The bounds are uniform in $0\leq s\leq\beta$, the interval, and the chain length; $K,c$ depend only on $\beta,q$.
\end{lemma}
\begin{proof}
Write $Z=\Tr e^{-sH}$ and $Z^{\mathrm{cut}}=\Tr e^{-sH^{\mathrm{cut}}}$. Set $E_I=e^{-sH/4}e^{sH^{\mathrm{cut}}/4}$. Since $\mathsf W$ represents $\Gamma_{\sigma_s}^{1/2}\circ\Gamma_\tau^{-1/2}$, a direct computation gives
\[
\mathcal W(Y)=(Z^{\mathrm{cut}}/Z)^{1/2}E_IYE_I^*,
 \qquad \mathsf W|Y\rangle=|\mathcal W(Y)\rangle.
\]
We first state the precise estimates used to approximate $E_I$. For one removed bond, let $H_-$ be the Hamiltonian after its removal and put $E=e^{-sH/4}e^{sH_-/4}$. The one-dimensional estimates in~\cite[Lemma~A.3 and Corollary~A.1]{BC2025} give
\[
 \|E\|+\|E^{-1}\|\leq K_\beta,\qquad
 E=\sum_{n\geq0}E_n,\qquad
 \|E_n\|\leq A_\beta\frac{B_\beta^n}{n!},
\]
where $E_n$ is supported within distance $n$ of that bond. The constants are uniform in the bond, the chain length, and $0\leq s\leq\beta$. Truncating this series at $n=L$ gives a supported approximation $E^{[L]}$ with
\[
 \|E-E^{[L]}\|
 \leq A_\beta\sum_{n>L}\frac{B_\beta^n}{n!}
 \leq C_\beta e^{-c_\beta L\log(2+L)}.
\]
For the last inequality, use $n!\geq(n/e)^n$ and bound the tail by a geometric series once $L$ exceeds $2B_\beta$; the finitely many smaller $L$ are absorbed in $C_\beta$.

There are at most two removed bonds. If $H^{(0)}=H$, $H^{(1)}$ removes the first bond, and $H^{(2)}=H^{\mathrm{cut}}$, then
\[
 E_I=\bigl(e^{-sH^{(0)}/4}e^{sH^{(1)}/4}\bigr)
      \bigl(e^{-sH^{(1)}/4}e^{sH^{(2)}/4}\bigr).
\]
Apply the preceding one-bond estimates to each factor; after the first cut, the second factor acts on the remaining component containing the second bond. Replacing each by its truncation gives $E_{I,L}$ supported on the union of at most two boundary strips,
\[
N_L:=
\bigcup_{\{u,v\}\,\text{a removed bond}}
\left\{x\in\Lambda:\min\{\operatorname{dist}(x,u),\operatorname{dist}(x,v)\}\le L\right\},
\]
with
\[
 \|E_I\|+\|E_I^{-1}\|+\|E_{I,L}\|\leq C_\beta,
 \qquad \|E_I-E_{I,L}\|\leq C_\beta e^{-c_\beta L\log(2+L)}.
\]
The cases of one or no removed bonds use one or no factors, respectively; in the latter case, $E_I=E_{I,L}=I$.

The scalar normalization is also uniform. Interpolating between $H^{\mathrm{cut}}$ and $H$ and differentiating the log partition function gives
\[
 |\log(Z/Z^{\mathrm{cut}})|
 \leq s\|H-H^{\mathrm{cut}}\|\leq2\beta.
\]
For $a_I=(Z^{\mathrm{cut}}/Z)^{1/2}$, the induced Hilbert--Schmidt norms therefore satisfy
\[
 \|\mathsf W\|\leq a_I\|E_I\|^2,\qquad
 \|\mathsf W^{-1}\|\leq a_I^{-1}\|E_I^{-1}\|^2.
\]
Define $\widetilde{\mathcal W}_L(Y)=a_I E_{I,L}YE_{I,L}^*$ and let $\widetilde{\mathsf W}_L$ be its vectorized representative. It is supported on the doubled site factors in $N_L$ and satisfies
\[
 \|\widetilde{\mathsf W}_L\|\leq a_I\|E_{I,L}\|^2\leq C_\beta.
\]
Moreover,
\[
 \|\mathsf W-\widetilde{\mathsf W}_L\|
 \leq a_I(\|E_I\|+\|E_{I,L}\|)\|E_I-E_{I,L}\|
 \leq C_\beta e^{-c_\beta L\log(2+L)}.
\]
Together with the bounds on $\mathsf W$ and $\mathsf W^{-1}$, these inequalities prove \eqref{eq:chain-cut-map} after enlarging $K$.
\end{proof}

\proofstep{Step 3: prove the local spectral gap}
We use Lemma~\ref{lem:chain-cut-locality} to prove the following uniform local spectral gap bound for intervals.
\begin{lemma}[Local spectral gaps for intervals]\label{lem:chain-local-poincare}
There is $\lambda=\lambda(\beta,q)>0$ such that, for every interval $I\subseteq\Lambda$, $0\leq s\leq\beta$, and every $|Y\rangle$,
\begin{equation}\label{eq:BC-input}
 \|(\one-\mathsf{P}_{s,I})|Y\rangle\|_2^2
 \leq\lambda^{-1}\langle Y|\mathsf{A}_{s,I}|Y\rangle,
 \qquad \one-\mathsf{P}_{s,I}\leq\lambda^{-1}\mathsf{A}_{s,I}.
\end{equation}
\end{lemma}
\begin{proof} Cut the boundary bonds of $I$ and write $\tau=\tau_I\otimes\tau_{I^c}$. Lemma~\ref{lem:chain-cut-locality} bounds the change of weighting $\mathsf{W}=\sGamma_{\sigma_s}^{1/2}\sGamma_\tau^{-1/2}$ and its inverse, so
\begin{equation}\label{eq:chain-cut-norm-comparison}
 K^{-1}\|X\|_{2,\tau}\leq\|X\|_{2,\sigma_s}
                         \leq K\|X\|_{2,\tau}.
\end{equation}
Because $H_I$ is a chain of the same interaction strength as $H$, according to~\cite[Theorem~V.1]{BC2025}, we have
\begin{equation}\label{eq:chain-isolated-gap-input}
 \inf_{c\in\mathbb C}\|X_I-cI_I\|_{2,\tau_I}^2
 \leq\lambda_{\mathrm{cut}}^{-1}
        \sum_{x\in I,u}\|[S_{x,u},X_I]\|_{2,\tau_I}^2,
\end{equation}
with $\lambda_{\mathrm{cut}}>0$ uniform in $I$ and $0\leq s\leq\beta$.

To extend this estimate to matrices on $I\cup I^c$, choose an orthonormal basis $\{T_j\}$ in the $\tau_{I^c}$-weighted matrix space and expand $X=\sum_jX_j\otimes T_j$. The product structure of $\tau$ gives
\[
 \begin{aligned}
 \inf_Z\|X-I_I\otimes Z\|_{2,\tau}^2=\sum_j\inf_{c_j\in\mathbb C}\|X_j-c_jI_I\|_{2,\tau_I}^2\leq\lambda_{\mathrm{cut}}^{-1}
             \sum_{x\in I,u}\|[S_{x,u},X]\|_{2,\tau}^2.
 \end{aligned}
\]
We can now compare the two weighted norms:
\[
 \begin{aligned}
 &\inf_Z\|X-I_I\otimes Z\|_{2,\sigma_s}^2
 \leq K^2\inf_Z\|X-I_I\otimes Z\|_{2,\tau}^2\\
 &\leq\frac{K^2}{\lambda_{\mathrm{cut}}}
             \sum_{x\in I,u}\|[S_{x,u},X]\|_{2,\tau}^2\leq\frac{K^4}{\lambda_{\mathrm{cut}}}
             \sum_{x\in I,u}\|[S_{x,u},X]\|_{2,\sigma_s}^2.
 \end{aligned}
\]
For $|Y\rangle=\sGamma_{\sigma_s}^{1/2}|X\rangle$, the left side is $\|(\one-\mathsf{P}_{s,I})|Y\rangle\|_2^2$ and the final sum is $\langle Y|\mathsf{A}_{s,I}|Y\rangle$. This proves \eqref{eq:BC-input} with $\lambda=\lambda_{\mathrm{cut}}K^{-4}$.
\end{proof}

\proofstep{Step 4: prove the local state domination}
\begin{lemma}[Local state domination in one dimension]\label{lem:chain-domination-transfer}
There is $\delta_*>0$, depending only on $\beta,q$, such that, for every $0\leq s\leq\beta$, $x\in\Lambda$, and $r\geq0$, with $I=\Ball_x(r)$ and $J=\Ball_x(\lfloor r/2\rfloor)$, every unit $|\psi_s\rangle\in\ker\mathsf{A}_{s,I}$ and every $0\leq\mathsf{O}_J\leq\one$ supported on $J$ satisfy
\begin{equation}\label{eq:chain-state-domination}
 \langle\psi_s|\mathsf O_J|\psi_s\rangle
 \leq2K^6\langle\Omega_{\sigma_s}|\mathsf O_J|\Omega_{\sigma_s}\rangle
       +\delta_*(1+r)^{-20},
\end{equation}
where $K$ is from Lemma~\ref{lem:chain-cut-locality}. In particular, \eqref{eq:local-gap-kernel-assumption} holds with $C_{\mathrm{dom}}=2K^6$ and $b=10$.
\end{lemma}
\begin{proof}
Set $T=\sqrt{\mathsf O_J}$, $D=[T,\mathsf W]$, and $|v\rangle=\mathsf W^{-1}|\psi_s\rangle$. The identities
\[
 \mathsf B_{s,x,u}^{\mathrm{cut}}\mathsf W^{-1}
 =\mathsf W^{-1}\mathsf B_{s,x,u},\qquad x\in I,
\]
show that $|v\rangle$ belongs to the cut kernel. Hence $|v\rangle=|\Omega_{\tau_I}\rangle\otimes|v_{I^c}\rangle$. Since $T$ is supported in $J\subseteq I$,
\[
 \|T|v\rangle\|_2=\|v\|_2\,\|T|\Omega_\tau\rangle\|_2,
 \qquad \|v\|_2\leq K.
\]
Write $a=\|T|\Omega_{\sigma_s}\rangle\|_2$ and $\varepsilon=\|D\|$. From $\mathsf W|\Omega_\tau\rangle=|\Omega_{\sigma_s}\rangle$ we obtain
\[
 \mathsf W T|\Omega_\tau\rangle
 =T|\Omega_{\sigma_s}\rangle-D|\Omega_\tau\rangle,
 \qquad \|T|\Omega_\tau\rangle\|_2\leq K(a+\varepsilon).
\]
Also $T|\psi_s\rangle=\mathsf W T|v\rangle+D|v\rangle$. Combining the last two bounds gives
\[
 \begin{aligned}
 \|T|\psi_s\rangle\|_2\leq K\|v\|_2\,\|T|\Omega_\tau\rangle\|_2+\varepsilon\|v\|_2\leq K^3a+(K^3+K)\varepsilon\leq K^3a+2K^3\varepsilon.
 \end{aligned}
\]
Squaring and using $(a+2\varepsilon)^2\leq2a^2+8\varepsilon^2$ proves
\begin{equation}\label{eq:chain-domination-reduction}
 \langle\psi_s|\mathsf O_J|\psi_s\rangle
 \leq2K^6\langle\Omega_{\sigma_s}|\mathsf O_J|\Omega_{\sigma_s}\rangle
       +8K^6\|[\sqrt{\mathsf O_J},\mathsf W]\|^2.
\end{equation}

To bound the last term in \eqref{eq:chain-domination-reduction}, let $r\geq4$ and choose $L=\lfloor r/4\rfloor$. The support of $\widetilde{\mathsf W}_L$ in \eqref{eq:chain-cut-map} is disjoint from $J=\Ball_x(\lfloor r/2\rfloor)$, so
\[
 \begin{aligned}
 \|[\sqrt{\mathsf O_J},\mathsf W]\|=\|[\sqrt{\mathsf O_J},\mathsf W-\widetilde{\mathsf W}_L]\|\leq2\|\mathsf W-\widetilde{\mathsf W}_L\|
 \leq Ce^{-cr\log(2+r)}.
 \end{aligned}
\]
If there are no cut bonds, $\mathsf W=\one$ and this commutator is zero. Substitution into \eqref{eq:chain-domination-reduction} proves \eqref{eq:chain-state-domination} for $r\geq4$, because $e^{-cr\log(2+r)}$ decays faster than every fixed inverse power. Enlarging $\delta_*$ covers $r<4$ and completes the proof.
\end{proof}

\proofstep{Step 5: conclude the one-dimensional proof}
\begin{proof}[Proof of Condition~\ref{cond:local-action} for one-dimensional systems]
Step~1 proves \eqref{eq:local-gap-auxiliary} with $c_{\mathrm{cmp}}=c_{\beta,q}$. Equation~\eqref{eq:BC-input} gives \eqref{eq:local-gap-assumption} with $\gamma_*=\lambda$ and $a=1$, and \eqref{eq:chain-state-domination} gives \eqref{eq:local-gap-kernel-assumption} with $C_{\mathrm{dom}}=2K^6$ and $b=10$. These constants are uniform in $s$ and the chain length. Taking $p_0=10$ gives $\min\{p_0,b\}=10>d+a=2$, so Proposition~\ref{prop:local-gap-criterion} proves \eqref{eq:comparison-target}.
\end{proof}

\subsection{Proof for high-temperature systems}\label{app:high-temperature}
For the high-temperature system, we let $\mathsf A_s=\mathsf U_s(-\mathsf L_{\sigma_0})\mathsf U_s^*$.
\proofstep{Step 1: identify the regional kernels and gaps}
\begin{proof}[Proof of \eqref{eq:ht-regional-kernel}]
At $s=0$, \eqref{eq:native-product-generator} gives $\mathsf A_0=\nu_q\sum_z\mathsf Q_z$, where $\mathsf Q_z$ projects onto the complement of the single-site Gibbs vector:
\[
\mathsf Q_z
=
\left(\mathbf1_z-
|\Omega_{\sigma_{0,z}}\rangle
\langle\Omega_{\sigma_{0,z}}|\right)
\otimes\mathbf1_{\Lambda\setminus\{z\}},\quad |\Omega_{\sigma_{0,z}}\rangle
=\frac{|I_z\rangle}{2^{q/2}}.
\]
Set
\[
 \mathsf a_{s,z}=\nu_q\mathsf U_s\mathsf Q_z\mathsf U_s^*,
 \qquad \mathsf P_I^0=\prod_{z\in I}(\one-\mathsf Q_z).
\]
The projections commute. If $|\psi\rangle$ is a common eigenvector of $\left\{\mathsf Q_z\right\}_{z\in I}$ outside $\ran\mathsf P_I^0$, at least one $\mathsf Q_z$ has eigenvalue one. Therefore
\[
 \begin{aligned}
 \mathsf A_{0,I}=\nu_q\sum_{z\in I}\mathsf Q_z
 &\geq\nu_q(\one-\mathsf P_I^0),\\
 \mathsf A_{s,I}=\mathsf U_s\mathsf A_{0,I}\mathsf U_s^*
 &\geq\nu_q(\one-\mathsf P_{s,I}),
 \qquad \mathsf P_{s,I}=\mathsf U_s\mathsf P_I^0\mathsf U_s^*.
 \end{aligned}
\]
The second line is \eqref{eq:ht-regional-kernel}. Since $\mathsf Q_z\mathsf P_I^0=0$ for $z\in I$, the first line shows that $\mathsf P_I^0$ projects onto $\ker\mathsf A_{0,I}$, which fixes the product Gibbs vector on $I$:
\[
 \ker\mathsf A_{0,I}=\ran\mathsf P_I^0
 =\operatorname{span}\{|\Omega_{\sigma_0,I}\rangle\}\otimes\mathcal H_{I^c}^{\mathrm{dbl}}.
\]
Since $\mathsf A_{s,I}=\mathsf U_s\mathsf A_{0,I}\mathsf U_s^*$, the regional kernels are $\ker\mathsf A_{s,I}=\mathsf U_s\ker\mathsf A_{0,I}$, which is \eqref{eq:ht-regional-kernel}. Exact transport of the Gibbs vector gives $\ker\mathsf A_s=\operatorname{span}\{|\Omega_{\sigma_s}\rangle\}$.
\end{proof}

\proofstep{Step 2: approximate transported operators on larger regions}
We control the action of the transport $\mathsf U_s$ on (approximately) local operators.
\begin{lemma}[Locality bound]\label{lem:ht-poisson-action}
Assume $\beta J\leq\theta_*$ as in \eqref{eq:ht-mixing-input}. For every fixed $p\geq0$, $x\in\Lambda$, and operator $\mathsf R$ on the doubled Hilbert space with $\|\mathsf R\|_{p+d,x}<\infty$, and all $0\leq s\leq\beta$,
\begin{equation}\label{eq:ht-poisson-action}
 \|\mathsf{U}_s^*\mathsf{R}\mathsf{U}_s\|_{p,x}\leq C_p\|\mathsf{R}\|_{p+d,x}.
\end{equation}
If $\mathsf{O}_J$ is supported on $J\subseteq I$, there exists an operator $\mathsf{O}_I$ supported on $I$ such that
\begin{equation}\label{eq:ht-observable-truncation}
 \|\mathsf{U}_s^*\mathsf{O}_J\mathsf{U}_s-\mathsf{O}_I\|\leq C_p|J|\|\mathsf{O}_J\|(1+\operatorname{dist}(J,I^c))^{-p}.
\end{equation}
When $I=\Lambda$, the approximation is exact. The constant $C_p$ depends only on $p,d,R_0,q,\theta_*$.
\end{lemma}
\begin{proof}[Proof of Lemma~\ref{lem:ht-poisson-action}] Fix $p\geq0$, put $b=p+2$ and $\zeta=b+d+1$, and choose $m>\zeta+d+1$. Local rapid mixing holds at every fixed order in the high-temperature regime, so Proposition~\ref{prop:flow-locality} gives
\begin{equation}\label{eq:ht-integrated-flow-strength}
 \sup_{0\leq a\leq\beta}\|\mathsf G_a\|_{m,\mathrm{loc}}
 \leq C_G\leq CJ,
 \qquad \int_0^s\|\mathsf G_a\|_{m,\mathrm{loc}}\,da\leq CsJ.
\end{equation}
To obtain a continuous supported decomposition, use $\mathsf G_{a,x}=[\mathsf M_{f_a},\mathsf L_{\sigma_a,x}]$ from \eqref{eq:flow-locality-reduction}. Its terms are continuous and skew-adjoint, and the same estimate at order $m+2$ gives $\sup_{a,x}\|\mathsf G_{a,x}\|_{m+2,x}\leq CJ$. Let $\mathbb E_I$ denote normalized partial trace onto the doubled site factors in $I$, and set
\[
 \begin{aligned}
 \mathsf G_{a,x,0}&=\mathbb E_{\Ball_x(0)}(\mathsf G_{a,x}),\\
 \mathsf G_{a,x,r}&=(\mathbb E_{\Ball_x(r)}-\mathbb E_{\Ball_x(r-1)})(\mathsf G_{a,x}),
 \qquad r\geq1.
 \end{aligned}
\]
For $r\geq1$, subtracting a truncation supported on $\Ball_x(r-1)$ from both partial traces and using contractivity gives $\|\mathsf G_{a,x,r}\|\leq CJ(1+r)^{-m-2}$. Thus $\mathsf G_a=\sum_{x,r}\mathsf G_{a,x,r}$, with
\[
 \supp\mathsf G_{a,x,r}\subseteq\Ball_x(r),\qquad
 \sup_{a,x}\sum_r(1+r)^m\|\mathsf G_{a,x,r}\|\leq CJ.
\]
The partial traces preserve adjoints, so each $i\mathsf G_{a,x,r}$ is Hermitian and continuous in $a$. These are the interaction terms for the evolution $\partial_a\mathsf U_a=-i(i\mathsf G_a)\mathsf U_a$.

For sites $z,w$, any ball containing both has $r\geq\dist(z,w)/2$, and for each $r$ there are at most $C(1+r)^d$ possible centers. Consequently,
\begin{equation}\label{eq:ht-flow-interaction-bound}
 \sum_{\substack{x,r:\ z,w\in\Ball_x(r)}}\|\mathsf G_{a,x,r}\|
 \leq CJ\sum_{r\geq\dist(z,w)/2}(1+r)^{d-m}
 \leq CJ(1+\dist(z,w))^{-\zeta}.
\end{equation}
Write $F_\zeta(z,w)=(1+\dist(z,w))^{-\zeta}$. Since $\zeta>d$,
\[
 \sup_z\sum_wF_\zeta(z,w)<\infty,
 \qquad
 \sum_yF_\zeta(z,y)F_\zeta(y,w)\leq C_\zeta F_\zeta(z,w).
\]
For the second inequality, split the sum according to which of $\dist(z,y)$ and $\dist(y,w)$ is at least $\dist(z,w)/2$, and use the first inequality on the other factor. Equation~\eqref{eq:ht-flow-interaction-bound} therefore verifies the hypothesis of~\cite[Theorem~3.1]{NSY2019}. This implies that, if $J$ and $K$ are disjoint,
\begin{equation}\label{eq:ht-flow-commutator}
 \|[\mathsf U_s^*\mathsf O_J\mathsf U_s,\mathsf B_K]\|
 \leq C\bigl(e^{CJs}-1\bigr)\|\mathsf O_J\|\|\mathsf B_K\|
       \sum_{z\in J}\sum_{w\in K}F_\zeta(z,w).
\end{equation}
The exponential is uniformly bounded because $sJ\leq\theta_*$, and the constants are independent of the size of $K$.

We now construct the approximation in \eqref{eq:ht-observable-truncation}. If $I=\Lambda$, take $\mathsf O_I=\mathsf U_s^*\mathsf O_J\mathsf U_s$. Otherwise, set
\[
 \mathsf O_I=\mathbb E_I(\mathsf U_s^*\mathsf O_J\mathsf U_s).
\]
The operators and partial trace here act on the doubled site factors. Averaging over unitaries $V$ on $I^c$ gives the identity
\[
 \mathsf O-\mathbb E_I(\mathsf O)
 =\int_{\mathcal U(I^c)}[\mathsf O,V]V^*\,dV,
\]
where $dV$ is normalized Haar measure. Apply \eqref{eq:ht-flow-commutator} with $K=I^c$. For $D=\dist(J,I^c)$, polynomial growth of balls gives
\[
 \sum_{z\in J}\sum_{w\in I^c}F_\zeta(z,w)
 \leq C|J|\sum_{n\geq D}(1+n)^{d-\zeta}
 \leq C_b|J|(1+D)^{-b}.
\]
This proves \eqref{eq:ht-observable-truncation}. The approximation is Hermitian whenever $\mathsf O_J$ is Hermitian.

Finally, we prove \eqref{eq:ht-poisson-action} by explicitly telescoping these approximations. Write $\mathsf R=\sum_n\mathsf R_n$ with $\supp\mathsf R_n\subseteq\Ball_x(n)$, and define
\[
 \begin{aligned}
 T_{n,\ell}&=\mathbb E_{\Ball_x(n+\ell)}
                    (\mathsf U_s^*\mathsf R_n\mathsf U_s),\\
 \mathsf R_{n,0}&=T_{n,0},\qquad
 \mathsf R_{n,\ell}=T_{n,\ell}-T_{n,\ell-1}\quad(\ell\geq1).
 \end{aligned}
\]
The approximation error just proved is at most
$C_p(1+n)^d(1+\ell)^{-p-2}\|\mathsf R_n\|$.
The triangle inequality for two successive errors, and contractivity for $\ell=0$, therefore give
\[
 \begin{gathered}
 \mathsf U_s^*\mathsf R_n\mathsf U_s=\sum_{\ell\geq0}\mathsf R_{n,\ell},
 \qquad \supp\mathsf R_{n,\ell}\subseteq\Ball_x(n+\ell),\\
 \|\mathsf R_{n,\ell}\|\leq C_p(1+n)^d(1+\ell)^{-p-2}\|\mathsf R_n\|.
 \end{gathered}
\]
Inserting this decomposition into the definition of the quasi-local norm yields
\[
 \begin{aligned}
 \|\mathsf U_s^*\mathsf R\mathsf U_s\|_{p,x}\leq\sum_{n,\ell\geq0}(1+n+\ell)^p\|\mathsf R_{n,\ell}\|\leq C_p\sum_{n\geq0}(1+n)^{p+d}\|\mathsf R_n\|
                       \sum_{\ell\geq0}(1+\ell)^{-2}.
 \end{aligned}
\]
The last series converges. Taking the infimum over the decompositions of $\mathsf R$ proves \eqref{eq:ht-poisson-action}.
\end{proof}

\proofstep{Step 3: prove local state domination}
\begin{proof}[Proof of \eqref{eq:local-gap-kernel-assumption} at high temperature]
Fix $I=\Ball_x(r)$, $J=\Ball_x(\lfloor r/2\rfloor)$, a unit vector $|\psi_s\rangle\in\ker\mathsf A_{s,I}$, and $0\leq\mathsf O_J\leq\one$ supported on $J$. By \eqref{eq:ht-regional-kernel},
\[
 \mathsf U_s^*|\psi_s\rangle=|\Omega_{\sigma_0,I}\rangle\otimes|v\rangle,
 \qquad \|v\|_2=1.
\]
Thus the transformed vector has the same factor on $I$ as $|\Omega_{\sigma_0}\rangle$. For every Hermitian operator $\mathsf{O}_I$ supported on $I$,
\begin{equation}\label{eq:ht-reference-expectations}
 \langle\psi_s|\mathsf{U}_s\mathsf{O}_I\mathsf{U}_s^*|\psi_s\rangle=\langle\Omega_{\sigma_0}|\mathsf{O}_I|\Omega_{\sigma_0}\rangle=\langle\Omega_{\sigma_s}|\mathsf{U}_s\mathsf{O}_I\mathsf{U}_s^*|\Omega_{\sigma_s}\rangle,
\end{equation}
where the last equality uses $\mathsf{U}_s|\Omega_{\sigma_0}\rangle=|\Omega_{\sigma_s}\rangle$. Hence
\begin{equation}\label{eq:ht-domination-transfer}
 \begin{alignedat}{2}
 &\langle\psi_s|\mathsf{O}_J|\psi_s\rangle\leq\langle\Omega_{\sigma_s}|\mathsf{U}_s\mathsf{O}_I\mathsf{U}_s^*|\Omega_{\sigma_s}\rangle+\|\mathsf{U}_s^*\mathsf{O}_J\mathsf{U}_s-\mathsf{O}_I\|&\qquad(\text{using \eqref{eq:ht-reference-expectations}})&\\
 &\leq\langle\Omega_{\sigma_s}|\mathsf{O}_J|\Omega_{\sigma_s}\rangle+2\|\mathsf{U}_s^*\mathsf{O}_J\mathsf{U}_s-\mathsf{O}_I\|.
 \end{alignedat}
\end{equation}
To bound the last term, take
\[
 \mathsf O_I=\mathbb E_I(\mathsf U_s^*\mathsf O_J\mathsf U_s).
\]
This operator is Hermitian. Applying \eqref{eq:ht-observable-truncation} with $p=5d+5$, using $|J|\leq C(1+r)^d$ and $\dist(J,I^c)\geq r/2$, gives
\begin{equation}\label{eq:ht-interior-approximation}
 \|\mathsf U_s^*\mathsf O_J\mathsf U_s-\mathsf O_I\|
 \leq C(1+r)^d(1+r/2)^{-5d-5}
 \leq C'(1+r)^{-4d-5}.
\end{equation}
If $I=\Lambda$, the error is zero. Substitution into \eqref{eq:ht-domination-transfer} yields
\begin{equation}\label{eq:ht-state-domination}
 \langle\psi_s|\mathsf O_J|\psi_s\rangle
 \leq\langle\Omega_{\sigma_s}|\mathsf O_J|\Omega_{\sigma_s}\rangle
       +2C'(1+r)^{-4d-5}.
\end{equation}
After enlarging the error coefficient, this implies \eqref{eq:local-gap-kernel-assumption} with $C_{\mathrm{dom}}=2$ and $b=2d+2$. All constants are uniform for $\beta J\leq\theta_*$.
\end{proof}

\proofstep{Step 4: compare local actions with the product energy}

We compare the local actions with the product energy at $s=0$ in the following lemma.

\begin{lemma}[Local actions at the product state]\label{lem:ht-poisson-product}
If $p>0$ and $\mathsf{R}|\Omega_{\sigma_0}\rangle=0$, then
\begin{equation}\label{eq:ht-product-spatial}
 \mathsf{R}^*\mathsf{R}\leq C_p\|\mathsf{R}\|_{p+1,z}^2\sum_wF_p(z,w)\mathsf{Q}_w,
\end{equation}
where $F_p(z,w)=(1+\operatorname{dist}(z,w))^{-p}$.
Consequently, for $p>d$ and a family with $\mathsf{R}_\alpha|\Omega_{\sigma_0}\rangle=0$,
\begin{equation}\label{eq:ht-poisson-product-quadratic}
 \sum_\alpha\|\mathsf{R}_\alpha|Y\rangle\|_2^2
 \leq C_p\left(\sup_z\sum_{\alpha:x_\alpha=z}\|\mathsf{R}_\alpha\|_{p+1,z}^2\right)
 \langle Y|(-\mathsf{L}_{\sigma_0})|Y\rangle.
\end{equation}
\end{lemma}
\begin{proof}[Proof of Lemma~\ref{lem:ht-poisson-product}]
Put $m=\|\mathsf R\|_{p+1,z}$; the case $m=0$ is immediate. Let $\mathsf P_n^0=\prod_{w\in\Ball_z(n)}(\one-\mathsf Q_w)$ and truncate a decomposition in \eqref{eq:polynorm} with weighted sum at most $2m$ to an operator $\mathsf R^{[n]}$ supported on $\Ball_z(n)$. Product factorization and $\mathsf R|\Omega_{\sigma_0}\rangle=0$ give
\[
 \begin{aligned}
 \|\mathsf R^{[n]}\mathsf P_n^0\|
 &=\|\mathsf R^{[n]}|\Omega_{\sigma_0}\rangle\|_2
 \leq\|\mathsf R-\mathsf R^{[n]}\|,\\
 \|\mathsf R\mathsf P_n^0\|
 &\leq2\|\mathsf R-\mathsf R^{[n]}\|
 \leq4m(1+n)^{-p-1}.
 \end{aligned}
\]
These projections are nested, $\|\mathsf R\|\leq m$, and $\mathsf R\mathsf P_N^0=0$ once $\Ball_z(N)=\Lambda$. Applying Lemma~\ref{lem:nested-kernel-bound} with $\varepsilon_n=4(1+n)^{-p-1}$, and using $\one-\mathsf P_n^0\leq\sum_{w\in\Ball_z(n)}\mathsf Q_w$, yields
\[
 \begin{aligned}
 \mathsf R^*\mathsf R\leq C_pm^2\sum_{n\geq0}(1+n)^{-p-1}(\one-\mathsf P_n^0)\leq C_pm^2\sum_w
       \left[\sum_{n\geq\dist(z,w)}(1+n)^{-p-1}\right]\mathsf Q_w.
 \end{aligned}
\]
The inner sum is at most $(1+p^{-1})F_p(z,w)$ by integral comparison, proving \eqref{eq:ht-product-spatial}. For $p>d$, sum this inequality over $\alpha$, group terms by $x_\alpha=z$, and use $\sup_w\sum_zF_p(z,w)<\infty$. Substituting $\sum_w\mathsf Q_w=\nu_q^{-1}(-\mathsf L_{\sigma_0})$ from \eqref{eq:native-product-generator} and taking the quadratic form on $|Y\rangle$ gives \eqref{eq:ht-poisson-product-quadratic}.
\end{proof}

\proofstep{Step 5: compare the transported product form with the sampler}

We now use Lemma~\ref{lem:ht-poisson-product} to compare the transported product form with the sampler.

\begin{lemma}[Comparison with the transported product form]\label{lem:ht-comparison}
For sufficiently small $\beta J$, the form \eqref{eq:ht-transported-form} satisfies
\begin{equation}\label{eq:ht-sampler-comparison}
 \mathsf{A}_s\leq4(-\mathsf{L}_{\sigma_s})\qquad(0\leq s\leq\beta).
\end{equation}
\end{lemma}
\begin{proof}[Proof of Lemma~\ref{lem:ht-comparison}]
By \eqref{eq:ht-transported-form}, it suffices to prove $-\mathsf{L}_{\sigma_0}\leq4\mathsf{U}_s^*(-\mathsf{L}_{\sigma_s})\mathsf{U}_s$. For any $|Y\rangle$, the representation \eqref{eq:native-dirichlet-form} and the triangle inequality give
\begin{equation}\label{eq:ht-comparison-reduction}
 \begin{aligned}
 &\sqrt{\langle Y|(-\mathsf{L}_{\sigma_0})|Y\rangle}
 =\left(\sum_{x,u}\int\|\widehat{\mathsf{B}}_{0,x,u}|Y\rangle\|_2^2\,d\mu\right)^{1/2}\\
 &\leq\sqrt{\langle Y|\mathsf{U}_s^*(-\mathsf{L}_{\sigma_s})\mathsf{U}_s|Y\rangle}
 +\left(\sum_{x,u}\int\|(\mathsf{U}_s^*\widehat{\mathsf{B}}_{s,x,u}\mathsf{U}_s-\widehat{\mathsf{B}}_{0,x,u})|Y\rangle\|_2^2\,d\mu\right)^{1/2}.
 \end{aligned}
\end{equation}
We will bound the last term in \eqref{eq:ht-comparison-reduction} by $CsJ\sqrt{\langle Y|(-\mathsf L_{\sigma_0})|Y\rangle}$, with $C$ uniform for $\beta J\leq\theta_*$. Choosing $C\beta J\leq1/2$ will then prove \eqref{eq:ht-sampler-comparison}.

To bound that term, integrate \eqref{eq:moving-frame-derivative}:
\[
 \mathsf{U}_s^*\widehat{\mathsf{B}}_{s,x,u}\mathsf{U}_s-\widehat{\mathsf{B}}_{0,x,u}
 =\int_0^s\mathsf{U}_a^*\mathsf{F}_{a,x,u}\mathsf{U}_a\,da,
\]
where $\mathsf{F}_{a,x,u}$ is defined in \eqref{eq:moving-frame-derivative}. Minkowski's inequality therefore gives
\begin{equation}\label{eq:ht-gradient-change}
 \left(\sum_{x,u}\int\|(\mathsf{U}_s^*\widehat{\mathsf{B}}_{s,x,u}\mathsf{U}_s-\widehat{\mathsf{B}}_{0,x,u})|Y\rangle\|_2^2\,d\mu\right)^{1/2}
 \leq\int_0^s\left(\sum_{x,u}\int\|\mathsf{U}_a^*\mathsf{F}_{a,x,u}\mathsf{U}_a|Y\rangle\|_2^2\,d\mu\right)^{1/2}da.
\end{equation}
We bound the expression under the square root on the right-hand side of \eqref{eq:ht-gradient-change} using \eqref{eq:ht-poisson-product-quadratic}. To verify its quasi-local norm hypothesis, take $p=d+1$, so $p+1+d=2d+2$, and apply \eqref{eq:ht-poisson-action} to $\mathsf F_{a,x,u}$:
\begin{equation}\label{eq:ht-transported-gradient-locality}
 \begin{aligned}
 \sup_x\sum_u\int\|\mathsf{U}_a^*\mathsf{F}_{a,x,u}\mathsf{U}_a\|_{p+1,x}^2\,d\mu\leq C_p\sup_x\sum_u\int\|\mathsf{F}_{a,x,u}\|_{p+1+d,x}^2\,d\mu\leq CJ^2.
 \end{aligned}
\end{equation}
The last bound follows from \eqref{eq:differentiated-gradient-locality} at order $2d+2$: \eqref{eq:native-gradient-locality}, \eqref{eq:dirichlet-locality-target}, and \eqref{eq:ht-integrated-flow-strength} give $C_B\leq C$, $C_{\dot B}\leq CJ$, and $C_G\leq CJ$, uniformly for $\beta J\leq\theta_*$. Local rapid mixing holds at every fixed order in this regime, so the hypotheses of that derivative estimate are satisfied. Its proof in Appendix~\ref{sec:energy-reduction} uses no energy comparison. The cancellation required by \eqref{eq:ht-poisson-product-quadratic} follows from \eqref{eq:derivative-gibbs-cancellation}:
\[
(\mathsf{U}_a^*\mathsf{F}_{a,x,u}\mathsf{U}_a)|\Omega_{\sigma_0}\rangle=\mathsf{U}_a^*\mathsf{F}_{a,x,u}|\Omega_{\sigma_a}\rangle=0\,.
\]

Applying \eqref{eq:ht-poisson-product-quadratic} and then \eqref{eq:ht-transported-gradient-locality} gives
\begin{equation}\label{eq:ht-transported-gradient-action}
 \begin{aligned}
 \sum_{x,u}\int\|\mathsf{U}_a^*\mathsf{F}_{a,x,u}\mathsf{U}_a|Y\rangle\|_2^2\,d\mu
 &\leq C_p\left(\sup_x\sum_u\int\|\mathsf{U}_a^*\mathsf{F}_{a,x,u}\mathsf{U}_a\|_{p+1,x}^2\,d\mu\right)
 \langle Y|(-\mathsf{L}_{\sigma_0})|Y\rangle\\
 &\leq CJ^2\langle Y|(-\mathsf{L}_{\sigma_0})|Y\rangle.
 \end{aligned}
\end{equation}
Substitution of \eqref{eq:ht-transported-gradient-action} into \eqref{eq:ht-gradient-change} bounds its right-hand side by $CsJ\sqrt{\langle Y|(-\mathsf L_{\sigma_0})|Y\rangle}$. Equation~\eqref{eq:ht-comparison-reduction} therefore becomes
\[
 \sqrt{\langle Y|(-\mathsf{L}_{\sigma_0})|Y\rangle}
 \leq\sqrt{\langle Y|\mathsf{U}_s^*(-\mathsf{L}_{\sigma_s})\mathsf{U}_s|Y\rangle}+CsJ\sqrt{\langle Y|(-\mathsf{L}_{\sigma_0})|Y\rangle}.
\]
Choose $0<\beta_0\leq\theta_*$ with $C\beta_0\leq1/2$, and assume $\beta J\leq\beta_0$. Moving the last term to the left gives $\langle Y|(-\mathsf{L}_{\sigma_0})|Y\rangle\leq4\langle Y|\mathsf{U}_s^*(-\mathsf{L}_{\sigma_s})\mathsf{U}_s|Y\rangle$. Conjugating by $\mathsf U_s$ proves \eqref{eq:ht-sampler-comparison}. This choice of $\beta_0$ depends only on $d,R_0,q$.
\end{proof}

\proofstep{Step 6: conclude the high-temperature proof}
\begin{proof}[Proof of Condition~\ref{cond:local-action} at high temperature]
Choose $\beta_0\leq\theta_*$ as in Lemma~\ref{lem:ht-comparison} and assume $\beta J\leq\beta_0$. Step~1 identifies the kernel of $\mathsf A_s$, and \eqref{eq:ht-sampler-comparison} completes \eqref{eq:local-gap-auxiliary} with $c_{\mathrm{cmp}}=4$. Equation~\eqref{eq:ht-regional-kernel} gives \eqref{eq:local-gap-assumption} with $\gamma_*=\nu_q$ and $a=1$. Equation~\eqref{eq:ht-state-domination} gives \eqref{eq:local-gap-kernel-assumption} with $C_{\mathrm{dom}}=2$ and $b=2d+2$. Taking $p_0=2d+2$ gives $\min\{p_0,b\}=2d+2>d+a=d+1$. Proposition~\ref{prop:local-gap-criterion} therefore proves \eqref{eq:comparison-target}, uniformly in $s$ and the volume.
\end{proof}

\FloatBarrier

\end{document}